\documentclass[11pt]{article}

\pdfoutput=1

\usepackage{setspace}
\usepackage{appendix}
\usepackage{xfrac}
\usepackage{amssymb,amsfonts}
\usepackage[figuresright]{rotating}
\usepackage{amsmath,amssymb}
\usepackage{dcolumn}
\usepackage{bm}
\usepackage{amscd,amsthm}
\usepackage{ifthen}
\usepackage{mathtools}
\usepackage{graphicx}
\usepackage{caption}
\usepackage{subcaption}
\usepackage{pgf,tikz}
\usepackage{pgfplots}
\usetikzlibrary{calc,shadows}
\usetikzlibrary{pgfplots.groupplots}
\usepackage{url}
\usetikzlibrary{external}
\usepackage{amsthm}
\usepackage{mathtools}

\usepackage[nosort]{cite}        
\usepackage{url} 
\usepackage{bbm}
\usepackage{paralist}
\usepackage{caption}
\usepackage[printonlyused]{acronym}
\usepackage{xcolor}

\newlength{\figwidth}
\pgfplotsset{every non boxed x axis/.append style={x axis line style=->,>=latex},every non boxed y axis/.append style={y axis line style=->,>=latex}}

\usepackage{amsmath}
\usepackage{amssymb}
\usepackage{amsfonts}
\usepackage{amsthm}
\usepackage{mathrsfs}
\usepackage{xspace}
\usepackage{bm}
\usepackage{fancyref}
\usepackage{textcomp}
\usepackage[Symbol]{upgreek}

\newcommand{\safemath}[2]{\newcommand{#1}{\ensuremath{#2}\xspace}}

\newcommand{\ssa}{\mathsf{a}}
\newcommand{\ssb}{\mathsf{b}}
\newcommand{\ssc}{\mathsf{c}}
\newcommand{\ssd}{\mathsf{d}}
\newcommand{\sse}{\mathsf{e}}
\newcommand{\ssf}{\mathsf{f}}
\newcommand{\ssg}{\mathsf{g}}
\newcommand{\ssh}{\mathsf{h}}
\newcommand{\ssi}{\mathsf{i}}
\newcommand{\ssj}{\mathsf{j}}
\newcommand{\ssk}{\mathsf{k}}
\newcommand{\ssl}{\mathsf{l}}
\newcommand{\ssm}{\mathsf{m}}
\newcommand{\ssn}{\mathsf{n}}
\newcommand{\sso}{\mathsf{o}}
\newcommand{\ssp}{\mathsf{p}}
\newcommand{\ssq}{\mathsf{q}}
\newcommand{\ssr}{\mathsf{r}}
\newcommand{\sss}{\mathsf{s}}
\newcommand{\sst}{\mathsf{t}}
\newcommand{\ssu}{\mathsf{u}}
\newcommand{\ssv}{\mathsf{v}}
\newcommand{\ssw}{\mathsf{w}}
\newcommand{\ssx}{\mathsf{x}}
\newcommand{\ssy}{\mathsf{y}}
\newcommand{\ssz}{\mathsf{z}}

\safemath{\bmsa}{\bm{\ssa}}
\safemath{\bmsb}{\bm{\ssb}}
\safemath{\bmsc}{\bm{\ssc}}
\safemath{\bmsd}{\bm{\ssd}}
\safemath{\bmse}{\bm{\sse}}
\safemath{\bmsf}{\bm{\ssf}}
\safemath{\bmsg}{\bm{\ssg}}
\safemath{\bmsh}{\bm{\ssh}}
\safemath{\bmsi}{\bm{\ssi}}
\safemath{\bmsj}{\bm{\ssj}}
\safemath{\bmsk}{\bm{\ssk}}
\safemath{\bmsl}{\bm{\ssl}}
\safemath{\bmsm}{\bm{\ssm}}
\safemath{\bmsn}{\bm{\ssn}}
\safemath{\bmso}{\bm{\sso}}
\safemath{\bmsp}{\bm{\ssp}}
\safemath{\bmsq}{\bm{\ssq}}
\safemath{\bmsr}{\bm{\ssr}}
\safemath{\bmss}{\bm{\sss}}
\safemath{\bmst}{\bm{\sst}}
\safemath{\bmsu}{\bm{\ssu}}
\safemath{\bmsv}{\bm{\ssv}}
\safemath{\bmsw}{\bm{\ssw}}
\safemath{\bmsx}{\bm{\ssx}}
\safemath{\bmsy}{\bm{\ssy}}
\safemath{\bmsz}{\bm{\ssz}}

\bmdefine{\bmualphad}{\upalpha}
\bmdefine{\bmubetad}{\upbeta}
\bmdefine{\bmuchid}{\upchi}
\bmdefine{\bmudeltad}{\updelta}
\bmdefine{\bmuepsilond}{\upepsilon}
\bmdefine{\bmuvarepsilond}{\upvarepsilon}
\bmdefine{\bmuetad}{\upeta}
\bmdefine{\bmugammad}{\upgamma}
\bmdefine{\bmuiotad}{\upiota}
\bmdefine{\bmukappad}{\upkappa}
\bmdefine{\bmulambdad}{\uplambda}
\bmdefine{\bmumud}{\upmu}
\bmdefine{\bmunud}{\upnu}
\bmdefine{\bmuomegad}{\upomega}
\bmdefine{\bmuphid}{\upphi}
\bmdefine{\bmuvarphid}{\upvarphi}
\bmdefine{\bmupid}{\uppi}
\bmdefine{\bmuvarpid}{\upvarpi}
\bmdefine{\bmupsid}{\uppsi}
\bmdefine{\bmurhod}{\uprho}
\bmdefine{\bmuvarrhod}{\upvarrho}
\bmdefine{\bmusigmad}{\upsigma}
\bmdefine{\bmuvarsigmad}{\upvarsigma}
\bmdefine{\bmutaud}{\uptau}
\bmdefine{\bmuthetad}{\uptheta}
\bmdefine{\bmuvarthetad}{\upvartheta}
\bmdefine{\bmuupsilond}{\upupsilon}
\bmdefine{\bmuxid}{\upxi}
\bmdefine{\bmuzetad}{\upzeta}

\safemath{\bmua}{\mathbf{a}}
\safemath{\bmub}{\mathbf{b}}
\safemath{\bmuc}{\mathbf{c}}
\safemath{\bmud}{\mathbf{d}}
\safemath{\bmue}{\mathbf{e}}
\safemath{\bmuf}{\mathbf{f}}
\safemath{\bmug}{\mathbf{g}}
\safemath{\bmuh}{\mathbf{h}}
\safemath{\bmui}{\mathbf{i}}
\safemath{\bmuj}{\mathbf{j}}
\safemath{\bmuk}{\mathbf{k}}
\safemath{\bmul}{\mathbf{l}}
\safemath{\bmum}{\mathbf{m}}
\safemath{\bmun}{\mathbf{n}}
\safemath{\bmuo}{\mathbf{o}}
\safemath{\bmup}{\mathbf{p}}
\safemath{\bmuq}{\mathbf{q}}
\safemath{\bmur}{\mathbf{r}}
\safemath{\bmus}{\mathbf{s}}
\safemath{\bmut}{\mathbf{t}}
\safemath{\bmuu}{\mathbf{u}}
\safemath{\bmuv}{\mathbf{v}}
\safemath{\bmuw}{\mathbf{w}}
\safemath{\bmux}{\mathbf{x}}
\safemath{\bmuy}{\mathbf{y}}
\safemath{\bmuz}{\mathbf{z}}

\safemath{\bmualpha}{\bmualphad}
\safemath{\bmubeta}{\bmubetad}
\safemath{\bmuchi}{\bumchid}
\safemath{\bmudelta}{\bmudeltad}
\safemath{\bmuepsilon}{\bmuepsilond}
\safemath{\bmuvarepsilon}{\bmuvarepsilond}
\safemath{\bmueta}{\bmuetad}
\safemath{\bmugamma}{\bmugammad}
\safemath{\bmuiota}{\bmuiotad}
\safemath{\bmukappa}{\bmukappad}
\safemath{\bmulambda}{\bmulambdad}
\safemath{\bmumu}{\bmumud}
\safemath{\bmunu}{\bmunud}
\safemath{\bmuomega}{\bmuomegad}
\safemath{\bmuphi}{\bmuphid}
\safemath{\bmuvarphi}{\bmuvarphid}
\safemath{\bmupi}{\bmupid}
\safemath{\bmuvarpi}{\bmuvarpid}
\safemath{\bmupsi}{\bmupsid}
\safemath{\bmurho}{\bmurhod}
\safemath{\bmuvarrho}{\bmuvarrhod}
\safemath{\bmusigma}{\bmusigmad}
\safemath{\bmuvarsigma}{\bmuvarsigmad}
\safemath{\bmutau}{\bmutaud}
\safemath{\bmutheta}{\bmuthetad}
\safemath{\bmuvartheta}{\bmuvarthetad}
\safemath{\bmuupsilon}{\bmuupsilond}
\safemath{\bmuxi}{\bmuxid}
\safemath{\bmuzeta}{\bmuzetad}

\bmdefine{\bmiad}{a}
\bmdefine{\bmibd}{b}
\bmdefine{\bmicd}{c}
\bmdefine{\bmidd}{d}
\bmdefine{\bmied}{e}
\bmdefine{\bmifd}{f}
\bmdefine{\bmigd}{g}
\bmdefine{\bmihd}{h}
\bmdefine{\bmiid}{i}
\bmdefine{\bmijd}{j}
\bmdefine{\bmikd}{k}
\bmdefine{\bmild}{l}
\bmdefine{\bmimd}{m}
\bmdefine{\bmind}{n}
\bmdefine{\bmiod}{o}
\bmdefine{\bmipd}{p}
\bmdefine{\bmiqd}{q}
\bmdefine{\bmird}{r}
\bmdefine{\bmisd}{s}
\bmdefine{\bmitd}{t}
\bmdefine{\bmiud}{u}
\bmdefine{\bmivd}{v}
\bmdefine{\bmiwd}{w}
\bmdefine{\bmixd}{x}
\bmdefine{\bmiyd}{y}
\bmdefine{\bmizd}{z}

\bmdefine{\bmialphad}{\alpha}
\bmdefine{\bmibetad}{\beta}
\bmdefine{\bmichid}{\chi}
\bmdefine{\bmideltad}{\delta}
\bmdefine{\bmiepsilond}{\epsilon}
\bmdefine{\bmivarepsilond}{\varepsilon}
\bmdefine{\bmietad}{\eta}
\bmdefine{\bmigammad}{\gamma}
\bmdefine{\bmiiotad}{\iota}
\bmdefine{\bmikappad}{\kappa}
\bmdefine{\bmivarkappad}{\varkappa}
\bmdefine{\bmilambdad}{\lambda}
\bmdefine{\bmimud}{\mu}
\bmdefine{\bminud}{\nu}
\bmdefine{\bmiomegad}{\omega}
\bmdefine{\bmiphid}{\phi}
\bmdefine{\bmivarphid}{\varphi}
\bmdefine{\bmipid}{\pi}
\bmdefine{\bmivarpid}{\varpi}
\bmdefine{\bmipsid}{\psi}
\bmdefine{\bmirhod}{\rho}
\bmdefine{\bmivarrhod}{\varrho}
\bmdefine{\bmisigmad}{\sigma}
\bmdefine{\bmivarsigmad}{\varsigma}
\bmdefine{\bmitaud}{\tau}
\bmdefine{\bmithetad}{\theta}
\bmdefine{\bmivarthetad}{\vartheta}
\bmdefine{\bmiupsilond}{\upsilon}
\bmdefine{\bmixid}{\xi}
\bmdefine{\bmizetad}{\zeta}

\safemath{\bmia}{\bmiad}
\safemath{\bmib}{\bmibd}
\safemath{\bmic}{\bmicd}
\safemath{\bmid}{\bmidd}
\safemath{\bmie}{\bmied}
\safemath{\bmif}{\bmifd}
\safemath{\bmig}{\bmigd}
\safemath{\bmih}{\bmihd}
\safemath{\bmii}{\bmiid}
\safemath{\bmij}{\bmijd}
\safemath{\bmik}{\bmikd}
\safemath{\bmil}{\bmild}
\safemath{\bmim}{\bmimd}
\safemath{\bmin}{\bmind}
\safemath{\bmio}{\bmiod}
\safemath{\bmip}{\bmipd}
\safemath{\bmiq}{\bmiqd}
\safemath{\bmir}{\bmird}
\safemath{\bmis}{\bmisd}
\safemath{\bmit}{\bmitd}
\safemath{\bmiu}{\bmiud}
\safemath{\bmiv}{\bmivd}
\safemath{\bmiw}{\bmiwd}
\safemath{\bmix}{\bmixd}
\safemath{\bmiy}{\bmiyd}
\safemath{\bmiz}{\bmizd}

\safemath{\bmialpha}{\bmialphad}
\safemath{\bmibeta}{\bmibetad}
\safemath{\bmichi}{\bmichid}
\safemath{\bmidelta}{\bmideltad}
\safemath{\bmiepsilon}{\bmiepsilond}
\safemath{\bmivarepsilon}{\bmivarepsilond}
\safemath{\bmieta}{\bmietad}
\safemath{\bmigamma}{\bmigammad}
\safemath{\bmiiota}{\bmiiotad}
\safemath{\bmikappa}{\bmikappad}
\safemath{\bmivarkappa}{\bmivarkappad}
\safemath{\bmilambda}{\bmilambdad}
\safemath{\bmimu}{\bmimud}
\safemath{\bminu}{\bminud}
\safemath{\bmiomega}{\bmiomegad}
\safemath{\bmiphi}{\bmiphid}
\safemath{\bmivarphi}{\bmivarphid}
\safemath{\bmipi}{\bmipid}
\safemath{\bmivarpi}{\bmivarpid}
\safemath{\bmipsi}{\bmipsid}
\safemath{\bmirho}{\bmirhod}
\safemath{\bmivarrho}{\bmivarrhod}
\safemath{\bmisigma}{\bmisigmad}
\safemath{\bmivarsigma}{\bmivarsigmad}
\safemath{\bmitau}{\bmitaud}
\safemath{\bmitheta}{\bmithetad}
\safemath{\bmivartheta}{\bmivarthetad}
\safemath{\bmiupsilon}{\bmiupsilond}
\safemath{\bmixi}{\bmixid}
\safemath{\bmizeta}{\bmizetad}

\bmdefine{\bmuDeltad}{\Updelta}
\bmdefine{\bmuGammad}{\Upgamma}
\bmdefine{\bmuLambdad}{\Uplambda}
\bmdefine{\bmuOmegad}{\Upomega}
\bmdefine{\bmuPhid}{\Upphi}
\bmdefine{\bmuPid}{\Uppi}
\bmdefine{\bmuPsid}{\Uppsi}
\bmdefine{\bmuSigmad}{\Upsigma}
\bmdefine{\bmuThetad}{\Uptheta}
\bmdefine{\bmuUpsilond}{\Upupsilon}
\bmdefine{\bmuXid}{\Upxi}

\safemath{\bmuA}{\mathbf{A}}
\safemath{\bmuB}{\mathbf{B}}
\safemath{\bmuC}{\mathbf{C}}
\safemath{\bmuD}{\mathbf{D}}
\safemath{\bmuE}{\mathbf{E}}
\safemath{\bmuF}{\mathbf{F}}
\safemath{\bmuG}{\mathbf{G}}
\safemath{\bmuH}{\mathbf{H}}
\safemath{\bmuI}{\mathbf{I}}
\safemath{\bmuJ}{\mathbf{J}}
\safemath{\bmuK}{\mathbf{K}}
\safemath{\bmuL}{\mathbf{L}}
\safemath{\bmuM}{\mathbf{M}}
\safemath{\bmuN}{\mathbf{N}}
\safemath{\bmuO}{\mathbf{O}}
\safemath{\bmuP}{\mathbf{P}}
\safemath{\bmuQ}{\mathbf{Q}}
\safemath{\bmuR}{\mathbf{R}}
\safemath{\bmuS}{\mathbf{S}}
\safemath{\bmuT}{\mathbf{T}}
\safemath{\bmuU}{\mathbf{U}}
\safemath{\bmuV}{\mathbf{V}}
\safemath{\bmuW}{\mathbf{W}}
\safemath{\bmuX}{\mathbf{X}}
\safemath{\bmuY}{\mathbf{Y}}
\safemath{\bmuZ}{\mathbf{Z}}

\safemath{\bmuZero}{\mathbf{0}}
\safemath{\bmuOne}{\mathbf{1}}

\safemath{\bmuDelta}{\bmuDeltad}
\safemath{\bmuGamma}{\bmuGammad}
\safemath{\bmuLambda}{\bmuLambdad}
\safemath{\bmuOmega}{\bmuOmegad}
\safemath{\bmuPhi}{\bmuPhid}
\safemath{\bmuPi}{\bmuPid}
\safemath{\bmuPsi}{\bmuPsid}
\safemath{\bmuSigma}{\bmuSigmad}
\safemath{\bmuTheta}{\bmuThetad}
\safemath{\bmuUpsilon}{\bmuUpsilond}
\safemath{\bmuXi}{\bmuXid}

\bmdefine{\bmiAd}{A}
\bmdefine{\bmiBd}{B}
\bmdefine{\bmiCd}{C}
\bmdefine{\bmiDd}{D}
\bmdefine{\bmiEd}{E}
\bmdefine{\bmiFd}{F}
\bmdefine{\bmiGd}{G}
\bmdefine{\bmiHd}{H}
\bmdefine{\bmiId}{I}
\bmdefine{\bmiJd}{J}
\bmdefine{\bmiKd}{K}
\bmdefine{\bmiLd}{L}
\bmdefine{\bmiMd}{M}
\bmdefine{\bmiOd}{N}
\bmdefine{\bmiPd}{O}
\bmdefine{\bmiQd}{P}
\bmdefine{\bmiRd}{R}
\bmdefine{\bmiSd}{S}
\bmdefine{\bmiTd}{T}
\bmdefine{\bmiUd}{U}
\bmdefine{\bmiVd}{V}
\bmdefine{\bmiWd}{W}
\bmdefine{\bmiXd}{X}
\bmdefine{\bmiYd}{Y}
\bmdefine{\bmiZd}{Z}

\bmdefine{\bmiDeltad}{\Delta}
\bmdefine{\bmiGammad}{\Gamma}
\bmdefine{\bmiLambdad}{\Lambda}
\bmdefine{\bmiOmegad}{\Omega}
\bmdefine{\bmiPhid}{\Phi}
\bmdefine{\bmiPid}{\Pi}
\bmdefine{\bmiPsid}{\Psi}
\bmdefine{\bmiSigmad}{\Sigma}
\bmdefine{\bmiThetad}{\Theta}
\bmdefine{\bmiUpsilond}{\Upsilon}
\bmdefine{\bmiXid}{\Xi}

\safemath{\bmiA}{\bmiAd}
\safemath{\bmiB}{\bmiBd}
\safemath{\bmiC}{\bmiCd}
\safemath{\bmiD}{\bmiDd}
\safemath{\bmiE}{\bmiEd}
\safemath{\bmiF}{\bmiFd}
\safemath{\bmiG}{\bmiGd}
\safemath{\bmiH}{\bmiHd}
\safemath{\bmiI}{\bmiId}
\safemath{\bmiJ}{\bmiJd}
\safemath{\bmiK}{\bmiKd}
\safemath{\bmiL}{\bmiLd}
\safemath{\bmiM}{\bmiMd}
\safemath{\bmiN}{\bmiNd}
\safemath{\bmiO}{\bmiOd}
\safemath{\bmiP}{\bmiPd}
\safemath{\bmiQ}{\bmiQd}
\safemath{\bmiR}{\bmiRd}
\safemath{\bmiS}{\bmiSd}
\safemath{\bmiT}{\bmiTd}
\safemath{\bmiU}{\bmiUd}
\safemath{\bmiV}{\bmiVd}
\safemath{\bmiW}{\bmiWd}
\safemath{\bmiX}{\bmiXd}
\safemath{\bmiY}{\bmiYd}
\safemath{\bmiZ}{\bmiZd}

\safemath{\bmiDelta}{\bmiDeltad}
\safemath{\bmiGamma}{\bmiGammad}
\safemath{\bmiLambda}{\bmiLambdad}
\safemath{\bmiOmega}{\bmiOmegad}
\safemath{\bmiPhi}{\bmiPhid}
\safemath{\bmiPi}{\bmiPid}
\safemath{\bmiPsi}{\bmiPsid}
\safemath{\bmiSigma}{\bmiSigmad}
\safemath{\bmiTheta}{\bmiThetad}
\safemath{\bmiUpsilon}{\bmiUpsilond}
\safemath{\bmiXi}{\bmiXid}

\safemath{\evA}{\mathcal{A}}
\safemath{\evB}{\mathcal{B}}
\safemath{\evC}{\mathcal{C}}
\safemath{\evD}{\mathcal{D}}
\safemath{\evE}{\mathcal{E}}
\safemath{\evF}{\mathcal{F}}
\safemath{\evG}{\mathcal{G}}
\safemath{\evH}{\mathcal{H}}
\safemath{\evI}{\mathcal{I}}
\safemath{\evJ}{\mathcal{J}}
\safemath{\evK}{\mathcal{K}}
\safemath{\evL}{\mathcal{L}}
\safemath{\evM}{\mathcal{M}}
\safemath{\evN}{\mathcal{N}}
\safemath{\evO}{\mathcal{O}}
\safemath{\evP}{\mathcal{P}}
\safemath{\evQ}{\mathcal{Q}}
\safemath{\evR}{\mathcal{R}}
\safemath{\evS}{\mathcal{S}}
\safemath{\evT}{\mathcal{T}}
\safemath{\evU}{\mathcal{U}}
\safemath{\evV}{\mathcal{V}}
\safemath{\evW}{\mathcal{W}}
\safemath{\evX}{\mathcal{X}}
\safemath{\evY}{\mathcal{Y}}
\safemath{\evZ}{\mathcal{Z}}

\safemath{\setA}{\mathcal{A}}
\safemath{\setB}{\mathcal{B}}
\safemath{\setC}{\mathcal{C}}
\safemath{\setD}{\mathcal{D}}
\safemath{\setE}{\mathcal{E}}
\safemath{\setF}{\mathcal{F}}
\safemath{\setG}{\mathcal{G}}
\safemath{\setH}{\mathcal{H}}
\safemath{\setI}{\mathcal{I}}
\safemath{\setJ}{\mathcal{J}}
\safemath{\setK}{\mathcal{K}}
\safemath{\setL}{\mathcal{L}}
\safemath{\setM}{\mathcal{M}}
\safemath{\setN}{\mathcal{N}}
\safemath{\setO}{\mathcal{O}}
\safemath{\setP}{\mathcal{P}}
\safemath{\setQ}{\mathcal{Q}}
\safemath{\setR}{\mathcal{R}}
\safemath{\setS}{\mathcal{S}}
\safemath{\setT}{\mathcal{T}}
\safemath{\setU}{\mathcal{U}}
\safemath{\setV}{\mathcal{V}}
\safemath{\setW}{\mathcal{W}}
\safemath{\setX}{\mathcal{X}}
\safemath{\setY}{\mathcal{Y}}
\safemath{\setZ}{\mathcal{Z}}
\safemath{\emptySet}{\varnothing}

\safemath{\colA}{\mathscr{A}}
\safemath{\colB}{\mathscr{B}}
\safemath{\colC}{\mathscr{C}}
\safemath{\colD}{\mathscr{D}}
\safemath{\colE}{\mathscr{E}}
\safemath{\colF}{\mathscr{F}}
\safemath{\colG}{\mathscr{G}}
\safemath{\colH}{\mathscr{H}}
\safemath{\colI}{\mathscr{I}}
\safemath{\colJ}{\mathscr{J}}
\safemath{\colK}{\mathscr{K}}
\safemath{\colL}{\mathscr{L}}
\safemath{\colM}{\mathscr{M}}
\safemath{\colN}{\mathscr{N}}
\safemath{\colO}{\mathscr{O}}
\safemath{\colP}{\mathscr{P}}
\safemath{\colQ}{\mathscr{Q}}
\safemath{\colR}{\mathscr{R}}
\safemath{\colS}{\mathscr{S}}
\safemath{\colT}{\mathscr{T}}
\safemath{\colU}{\mathscr{U}}
\safemath{\colV}{\mathscr{V}}
\safemath{\colW}{\mathscr{W}}
\safemath{\colX}{\mathscr{X}}
\safemath{\colY}{\mathscr{Y}}
\safemath{\colZ}{\mathscr{Z}}

\safemath{\opA}{\mathbb{A}}
\safemath{\opB}{\mathbb{B}}
\safemath{\opC}{\mathbb{C}}
\safemath{\opD}{\mathbb{D}}
\safemath{\opE}{\mathbb{E}}
\safemath{\opF}{\mathbb{F}}
\safemath{\opG}{\mathbb{G}}
\safemath{\opH}{\mathbb{H}}
\safemath{\opI}{\mathbb{I}}
\safemath{\opJ}{\mathbb{J}}
\safemath{\opK}{\mathbb{K}}
\safemath{\opL}{\mathbb{L}}
\safemath{\opM}{\mathbb{M}}
\safemath{\opN}{\mathbb{N}}
\safemath{\opO}{\mathbb{O}}
\safemath{\opP}{\mathbb{P}}
\safemath{\opQ}{\mathbb{Q}}
\safemath{\opR}{\mathbb{R}}
\safemath{\opS}{\mathbb{S}}
\safemath{\opT}{\mathbb{T}}
\safemath{\opU}{\mathbb{U}}
\safemath{\opV}{\mathbb{V}}
\safemath{\opW}{\mathbb{W}}
\safemath{\opX}{\mathbb{X}}
\safemath{\opY}{\mathbb{Y}}
\safemath{\opZ}{\mathbb{Z}}
\safemath{\opZero}{\mathbb{O}}
\safemath{\identityop}{\opI}

\safemath{\sca}{a}
\safemath{\scb}{b}
\safemath{\scc}{c}
\safemath{\scd}{d}
\safemath{\sce}{e}
\safemath{\scf}{f}
\safemath{\scg}{g}
\safemath{\sch}{h}
\safemath{\sci}{i}
\safemath{\scj}{j}
\safemath{\sck}{k}
\safemath{\scl}{l}
\safemath{\scm}{m}
\safemath{\scn}{n}
\safemath{\sco}{o}
\safemath{\scp}{p}
\safemath{\scq}{q}
\safemath{\scr}{r}
\safemath{\scs}{s}
\safemath{\sct}{t}
\safemath{\scu}{u}
\safemath{\scv}{v}
\safemath{\scw}{w}
\safemath{\scx}{x}
\safemath{\scy}{y}
\safemath{\scz}{z}

\safemath{\scA}{A}
\safemath{\scB}{B}
\safemath{\scC}{C}
\safemath{\scD}{D}
\safemath{\scE}{E}
\safemath{\scF}{F}
\safemath{\scG}{G}
\safemath{\scH}{H}
\safemath{\scI}{I}
\safemath{\scJ}{J}
\safemath{\scK}{K}
\safemath{\scL}{L}
\safemath{\scM}{M}
\safemath{\scN}{N}
\safemath{\scO}{O}
\safemath{\scP}{P}
\safemath{\scQ}{Q}
\safemath{\scR}{R}
\safemath{\scS}{S}
\safemath{\scT}{T}
\safemath{\scU}{U}
\safemath{\scV}{V}
\safemath{\scW}{W}
\safemath{\scX}{X}
\safemath{\scY}{Y}
\safemath{\scZ}{Z}

\safemath{\scalpha}{\alpha}
\safemath{\scbeta}{\beta}
\safemath{\scchi}{\chi}
\safemath{\scdelta}{\delta}
\safemath{\scepsilon}{\epsilon}
\safemath{\scvarepsilon}{\varepsilon}
\safemath{\sceta}{\eta}
\safemath{\scgamma}{\gamma}
\safemath{\sciota}{\iota}
\safemath{\sckappa}{\kappa}
\safemath{\scvarkappa}{\varkappa}
\safemath{\sclambda}{\lambda}
\safemath{\scmu}{\mu}
\safemath{\scnu}{\nu}
\safemath{\scomega}{\omega}
\safemath{\scphi}{\phi}
\safemath{\scvarphi}{\varphi}
\safemath{\scpi}{\pi}
\safemath{\scvarpi}{\varpi}
\safemath{\scpsi}{\psi}
\safemath{\scrho}{\rho}
\safemath{\scvarrho}{\varrho}
\safemath{\scsigma}{\sigma}
\safemath{\scvarsigma}{\varsigma}
\safemath{\sctau}{\tau}
\safemath{\sctheta}{\theta}
\safemath{\scvartheta}{\vartheta}
\safemath{\scupsilon}{\upsilon}
\safemath{\scxi}{\xi}
\safemath{\sczeta}{\zeta}

\safemath{\veca}{\mathbf{a}}
\safemath{\vecb}{\mathbf{b}}
\safemath{\vecc}{\mathbf{c}}
\safemath{\vecd}{\mathbf{d}}
\safemath{\vece}{\mathbf{e}}
\safemath{\vecf}{\mathbf{f}}
\safemath{\vecg}{\mathbf{g}}
\safemath{\vech}{\mathbf{h}}
\safemath{\veci}{\mathbf{i}}
\safemath{\vecj}{\mathbf{j}}
\safemath{\veck}{\mathbf{k}}
\safemath{\vecl}{\mathbf{l}}
\safemath{\vecm}{\mathbf{m}}
\safemath{\vecn}{\mathbf{n}}
\safemath{\veco}{\mathbf{o}}
\safemath{\vecp}{\mathbf{p}}
\safemath{\vecq}{\mathbf{q}}
\safemath{\vecr}{\mathbf{r}}
\safemath{\vecs}{\mathbf{s}}
\safemath{\vect}{\mathbf{t}}
\safemath{\vecu}{\mathbf{u}}
\safemath{\vecv}{\mathbf{v}}
\safemath{\vecw}{\mathbf{w}}
\safemath{\vecx}{\mathbf{x}}
\safemath{\vecy}{\mathbf{y}}
\safemath{\vecz}{\mathbf{z}}
\safemath{\veczero}{\mathbf{0}}
\safemath{\vecone}{\mathbf{1}}

\safemath{\vecalpha}{\upalpha}
\safemath{\vecbeta}{\upbeta}
\safemath{\vecchi}{\upchi}
\safemath{\vecdelta}{\updelta}
\safemath{\vecepsilon}{\upepsilon}
\safemath{\vecvarepsilon}{\upvarepsilon}
\safemath{\veceta}{\upeta}
\safemath{\vecgamma}{\upgamma}
\safemath{\veciota}{\upiota}
\safemath{\veckappa}{\upkappa}
\safemath{\veclambda}{\uplambda}
\safemath{\vecmu}{\text{\textmu}}
\safemath{\vecnu}{\upnu}
\safemath{\vecomega}{\upomega}
\safemath{\vecphi}{\upphi}
\safemath{\vecvarphi}{\upvarphi}
\safemath{\vecpi}{\uppi}
\safemath{\vecvarpi}{\upvarpi}
\safemath{\vecpsi}{\uppsi}
\safemath{\vecrho}{\uprho}
\safemath{\vecvarrho}{\upvarrho}
\safemath{\vecsigma}{\upsigma}
\safemath{\vecvarsigma}{\upvarsigma}
\safemath{\vectau}{\uptau}
\safemath{\vectheta}{\uptheta}
\safemath{\vecvartheta}{\upvartheta}
\safemath{\vecupsilon}{\upupsilon}
\safemath{\vecxi}{\upxi}
\safemath{\veczeta}{\upzeta}

\safemath{\vecac}{a}
\safemath{\vecbc}{b}
\safemath{\veccc}{c}
\safemath{\vecdc}{d}
\safemath{\vecec}{e}
\safemath{\vecfc}{f}
\safemath{\vecgc}{g}
\safemath{\vechc}{h}
\safemath{\vecic}{i}
\safemath{\vecjc}{j}
\safemath{\veckc}{k}
\safemath{\veclc}{l}
\safemath{\vecmc}{m}
\safemath{\vecnc}{n}
\safemath{\vecoc}{o}
\safemath{\vecpc}{p}
\safemath{\vecqc}{q}
\safemath{\vecrc}{r}
\safemath{\vecsc}{s}
\safemath{\vectc}{t}
\safemath{\vecuc}{u}
\safemath{\vecvc}{v}
\safemath{\vecwc}{w}
\safemath{\vecxc}{x}
\safemath{\vecyc}{y}
\safemath{\veczc}{z}

\safemath{\matA}{\mathbf{A}}
\safemath{\matB}{\mathbf{B}}
\safemath{\matC}{\mathbf{C}}
\safemath{\matD}{\mathbf{D}}
\safemath{\matE}{\mathbf{E}}
\safemath{\matF}{\mathbf{F}}
\safemath{\matG}{\mathbf{G}}
\safemath{\matH}{\mathbf{H}}
\safemath{\matI}{\mathbf{I}}
\safemath{\matJ}{\mathbf{J}}
\safemath{\matK}{\mathbf{K}}
\safemath{\matL}{\mathbf{L}}
\safemath{\matM}{\mathbf{M}}
\safemath{\matN}{\mathbf{N}}
\safemath{\matO}{\mathbf{O}}
\safemath{\matP}{\mathbf{P}}
\safemath{\matQ}{\mathbf{Q}}
\safemath{\matR}{\mathbf{R}}
\safemath{\matS}{\mathbf{S}}
\safemath{\matT}{\mathbf{T}}
\safemath{\matU}{\mathbf{U}}
\safemath{\matV}{\mathbf{V}}
\safemath{\matW}{\mathbf{W}}
\safemath{\matX}{\mathbf{X}}
\safemath{\matY}{\mathbf{Y}}
\safemath{\matZ}{\mathbf{Z}}
\safemath{\matzero}{\mathbf{0}}

\safemath{\matDelta}{\Updelta}
\safemath{\matGamma}{\Upgammma}
\safemath{\matLambda}{\Uplambda}
\safemath{\matOmega}{\Upomega}
\safemath{\matPhi}{\Upphi}
\safemath{\matPi}{\Uppi}
\safemath{\matPsi}{\Uppsi}
\safemath{\matSigma}{\Upsigma}
\safemath{\matTheta}{\Uptheta}
\safemath{\matUpsilon}{\Upupsilon}
\safemath{\matXi}{\Upxi}

\safemath{\matidentity}{\matI}
\safemath{\vecunit}{\vece} 
\safemath{\matone}{\matO}

\safemath{\matAc}{a}
\safemath{\matBc}{b}
\safemath{\matCc}{c}
\safemath{\matDc}{d}
\safemath{\matEc}{e}
\safemath{\matFc}{f}
\safemath{\matGc}{g}
\safemath{\matHc}{h}
\safemath{\matIc}{i}
\safemath{\matJc}{j}
\safemath{\matKc}{k}
\safemath{\matLc}{l}
\safemath{\matMc}{m}
\safemath{\matNc}{n}
\safemath{\matOc}{o}
\safemath{\matPc}{p}
\safemath{\matQc}{q}
\safemath{\matRc}{r}
\safemath{\matSc}{s}
\safemath{\matTc}{t}
\safemath{\matUc}{u}
\safemath{\matVc}{v}
\safemath{\matWc}{w}
\safemath{\matXc}{x}
\safemath{\matYc}{y}
\safemath{\matZc}{z}

\safemath{\rnda}{\bmia}
\safemath{\rndb}{\bmib}
\safemath{\rndc}{\bmic}
\safemath{\rndd}{\bmid}
\safemath{\rnde}{\bmie}
\safemath{\rndf}{\bmif}
\safemath{\rndg}{\bmig}
\safemath{\rndh}{\bmih}
\safemath{\rndi}{\bmii}
\safemath{\rndj}{\bmij}
\safemath{\rndk}{\bmik}
\safemath{\rndl}{\bmil}
\safemath{\rndm}{\bmim}
\safemath{\rndn}{\bmin}
\safemath{\rndo}{\bmio}
\safemath{\rndp}{\bmip}
\safemath{\rndq}{\bmiq}
\safemath{\rndr}{\bmir}
\safemath{\rnds}{\bmis}
\safemath{\rndt}{\bmit}
\safemath{\rndu}{\bmiu}
\safemath{\rndv}{\bmiv}
\safemath{\rndw}{\bmiw}
\safemath{\rndx}{\bmix}
\safemath{\rndy}{\bmiy}
\safemath{\rndz}{\bmiz}

\safemath{\rndA}{\bmiA}
\safemath{\rndB}{\bmiB}
\safemath{\rndC}{\bmiC}
\safemath{\rndD}{\bmiD}
\safemath{\rndE}{\bmiE}
\safemath{\rndF}{\bmiF}
\safemath{\rndG}{\bmiG}
\safemath{\rndH}{\bmiH}
\safemath{\rndI}{\bmiI}
\safemath{\rndJ}{\bmiJ}
\safemath{\rndK}{\bmiK}
\safemath{\rndL}{\bmiL}
\safemath{\rndM}{\bmiM}
\safemath{\rndN}{\bmiN}
\safemath{\rndO}{\bmiO}
\safemath{\rndP}{\bmiP}
\safemath{\rndQ}{\bmiQ}
\safemath{\rndR}{\bmiR}
\safemath{\rndS}{\bmiS}
\safemath{\rndT}{\bmiT}
\safemath{\rndU}{\bmiU}
\safemath{\rndV}{\bmiV}
\safemath{\rndW}{\bmiW}
\safemath{\rndX}{\bmiX}
\safemath{\rndY}{\bmiY}
\safemath{\rndZ}{\bmiZ}

\safemath{\rndalpha}{\bmialpha}
\safemath{\rndbeta}{\bmibeta}
\safemath{\rndchi}{\bmichi}
\safemath{\rnddelta}{\bmidelta}
\safemath{\rndepsilon}{\bmiepsilon}
\safemath{\rndvarepsilon}{\bmivarepsilon}
\safemath{\rndeta}{\bmieta}
\safemath{\rndgamma}{\bmigamma}
\safemath{\rndiota}{\bmiiota}
\safemath{\rndkappa}{\bmikappa}
\safemath{\rndlambda}{\bmilambda}
\safemath{\rndmu}{\bmimu}
\safemath{\rndnu}{\bminu}
\safemath{\rndomega}{\bmiomega}
\safemath{\rndphi}{\bmiphi}
\safemath{\rndvarphi}{\bmivarphi}
\safemath{\rndpi}{\bmipi}
\safemath{\rndvarpi}{\bmivarpi}
\safemath{\rndpsi}{\bmipsi}
\safemath{\rndrho}{\bmirho}
\safemath{\rndvarrho}{\bmivarrho}
\safemath{\rndsigma}{\bmisigma}
\safemath{\rndvarsigma}{\bmivarsigma}
\safemath{\rndtau}{\bmitau}
\safemath{\rndtheta}{\bmitheta}
\safemath{\rndvartheta}{\bmivartheta}
\safemath{\rndupsilon}{\bmiupsilon}
\safemath{\rndxi}{\bmixi}
\safemath{\rndzeta}{\bmizeta}

\safemath{\rveca}{\bmua}
\safemath{\rvecb}{\bmub}
\safemath{\rvecc}{\bmuc}
\safemath{\rvecd}{\bmud}
\safemath{\rvece}{\bmue}
\safemath{\rvecf}{\bmuf}
\safemath{\rvecg}{\bmug}
\safemath{\rvech}{\bmuh}
\safemath{\rveci}{\bmui}
\safemath{\rvecj}{\bmuj}
\safemath{\rveck}{\bmuk}
\safemath{\rvecl}{\bmul}
\safemath{\rvecm}{\bmum}
\safemath{\rvecn}{\bmun}
\safemath{\rveco}{\bmuo}
\safemath{\rvecp}{\bmup}
\safemath{\rvecq}{\bmuq}
\safemath{\rvecr}{\bmur}
\safemath{\rvecs}{\bmus}
\safemath{\rvect}{\bmut}
\safemath{\rvecu}{\bmuu}
\safemath{\rvecv}{\bmuv}
\safemath{\rvecw}{\bmuw}
\safemath{\rvecx}{\bmux}
\safemath{\rvecy}{\bmuy}
\safemath{\rvecz}{\bmuz}

\safemath{\rvecalpha}{\bmualpha}
\safemath{\rvecbeta}{\bmubeta}
\safemath{\rvecchi}{\bmuchi}
\safemath{\rvecdelta}{\bmudelta}
\safemath{\rvecepsilon}{\bmuepsilon}
\safemath{\rvecvarepsilon}{\bmuvarepsilon}
\safemath{\rveceta}{\bmueta}
\safemath{\rvecgamma}{\bmugamma}
\safemath{\rveciota}{\bmuiota}
\safemath{\rveckappa}{\bmukappa}
\safemath{\rveclambda}{\bmulambda}
\safemath{\rvecmu}{\bmumu}
\safemath{\rvecnu}{\bmunu}
\safemath{\rvecomega}{\bmuomega}
\safemath{\rvecphi}{\bmuphi}
\safemath{\rvecvarphi}{\bmuvarphi}
\safemath{\rvecpi}{\bmupi}
\safemath{\rvecvarpi}{\bmuvarpi}
\safemath{\rvecpsi}{\bmupsi}
\safemath{\rvecrho}{\bmurho}
\safemath{\rvecvarrho}{\bmuvarrho}
\safemath{\rvecsigma}{\bmusigma}
\safemath{\rvecvarsigma}{\bmuvarsigma}
\safemath{\rvectau}{\bmutau}
\safemath{\rvectheta}{\bmutheta}
\safemath{\rvecvartheta}{\bmuvartheta}
\safemath{\rvecupsilon}{\bmuupsilon}
\safemath{\rvecxi}{\bmuxi}
\safemath{\rveczeta}{\bmuzeta}

\safemath{\rmatA}{\bmuA}
\safemath{\rmatB}{\bmuB}
\safemath{\rmatC}{\bmuC}
\safemath{\rmatD}{\bmuD}
\safemath{\rmatE}{\bmuE}
\safemath{\rmatF}{\bmuF}
\safemath{\rmatG}{\bmuG}
\safemath{\rmatH}{\bmuH}
\safemath{\rmatI}{\bmuI}
\safemath{\rmatJ}{\bmuJ}
\safemath{\rmatK}{\bmuK}
\safemath{\rmatL}{\bmuL}
\safemath{\rmatM}{\bmuM}
\safemath{\rmatN}{\bmuN}
\safemath{\rmatO}{\bmuO}
\safemath{\rmatP}{\bmuP}
\safemath{\rmatQ}{\bmuQ}
\safemath{\rmatR}{\bmuR}
\safemath{\rmatS}{\bmuS}
\safemath{\rmatT}{\bmuT}
\safemath{\rmatU}{\bmuU}
\safemath{\rmatV}{\bmuV}
\safemath{\rmatW}{\bmuW}
\safemath{\rmatX}{\bmuX}
\safemath{\rmatY}{\bmuY}
\safemath{\rmatZ}{\bmuZ}

\safemath{\rmatDelta}{\bmuDelta}
\safemath{\rmatGamma}{\bmuGamma}
\safemath{\rmatLambda}{\bmuLambda}
\safemath{\rmatOmega}{\bmuOmega}
\safemath{\rmatPhi}{\bmuPhi}
\safemath{\rmatPi}{\bmuPi}
\safemath{\rmatPsi}{\bmuPsi}
\safemath{\rmatSigma}{\bmuSigma}
\safemath{\rmatTheta}{\bmuTheta}
\safemath{\rmatUpsilon}{\bmuUpsilon}
\safemath{\rmatXi}{\bmuXi}

\newenvironment{textbmatrix}{	\setlength{\arraycolsep}{2.5pt}%
								\big[\begin{matrix}}{\end{matrix}\big]%
								\raisebox{0.08ex}{\vphantom{M}}}

 \def\btm{\begin{textbmatrix}}
 \def\etm{\end{textbmatrix}}

\newcommand{\lefto}{\mathopen{}\left}

\DeclareMathOperator{\diag}{diag}			
\DeclareMathOperator{\adj}{adj}				
\newcommand{\vectorize}[1]{\mathrm{vec}(#1)} 
\DeclareMathOperator{\sign}{sign}			
\DeclareMathOperator*{\argmin}{arg\;min}		
\DeclareMathOperator{\kron}{\otimes}			
\DeclareMathOperator{\Exop}{\opE}			
\DeclareMathOperator{\conv}{\star}			

\DeclareMathOperator{\landauO}{\mathcal{O}}

\safemath{\fun}{\scf}						
\safemath{\altfun}{\scg}						
\safemath{\vectr}{\veca}	 	
\safemath{\vectrc}{\vecac}	 	

\safemath{\vrbl}{t}						
\safemath{\altvrbl}{y}						
\safemath{\aaltvrbl}{z}						

\safemath{\vvrbl}{\vecx}						
\safemath{\altvvrbl}{\vecy}						
\safemath{\aaltvvrbl}{\vecz}						

\safemath{\aaltfun}{\sch}
\safemath{\bel}{\sce}					
\safemath{\altbel}{\sce}					
\safemath{\frel}{g}					
\safemath{\altfrel}{g}					
\safemath{\dfrel}{\tilde{g}}					
\safemath{\altdfrel}{\tilde{g}}					

\safemath{\mat}{\matA}						
\safemath{\matc}{\matAc}						
\safemath{\altmat}{\matB}						
\safemath{\altmatc}{\matBc}						

\safemath{\altvectr}{\vecv}						
\safemath{\altvectrc}{\vecvc}						

\newcommand{\nullspace}{\setN}	 			
\newcommand{\simplifiedmathchoice}[2]{\mathchoice{#1}{#2}{#2}{#2}}

\newcommand{\Ex}[2]{\simplifiedmathchoice{\ensuremath{\Exop_{#1}\lefto[#2\right]}}{\ensuremath{\Exop_{#1}\bigl[#2\bigr]}}} 	
\newcommand{\abs}[1]{\mathchoice{{\left\lvert#1\right\rvert}}{{\bigl\lvert#1\bigr\rvert}}{{\left\lvert#1\right\rvert}}{{\left\lvert#1\right\rvert}}}		

\newcommand{\union}{\cup}					

\newcommand{\intersect}{\cap}				

\newcommand{\twonorm}[1]{\lVert#1\rVert_2}		
\newcommand{\onenorm}[1]{\lVert#1\rVert_1}		
\newcommand{\infnorm}[1]{\lVert#1\rVert_{\infty}}		
\newcommand{\zeronorm}[1]{\lVert#1\rVert_0}		
\newcommand{\vecnorm}[1]{\lVert#1\rVert}		
\newcommand{\conj}[1]{\ensuremath{#1^{*}}} 	
\newcommand{\tp}[1]{\ensuremath{#1^{\mathsf{T}}}} 		
\newcommand{\herm}[1]{\ensuremath{#1^{\mathsf{H}}}} 	
\newcommand{\inv}[1]{\ensuremath{#1^{-1}}} 	

\safemath{\dirac}{\delta}					
\safemath{\diracp}{\dirac(\time)}			
\safemath{\krond}{\dirac}					
\safemath{\indfun}{I}						
\safemath{\stepfun}{u}						

\safemath{\upto}{\uparrow}
\safemath{\downto}{\downarrow}
\safemath{\iu}{\mathrm{i}}							
\safemath{\maj}{\succ}
\newcommand{\dftmat}[1]{\matF_{#1}}			
\safemath{\mdft}{\dftmat{}}					
\safemath{\runity}{\beta}					
\safemath{\eval}{\lambda}					

\safemath{\veval}{\veclambda}				
\safemath{\littleo}{\sco}					

\let\im\undefined
\safemath{\re}{\Re}				
\safemath{\im}{\Im}				

\safemath{\euclidspace}{\complexset}			
\safemath{\confunspace}{\setC}				
\newcommand{\banachseqspace}[1]{l^{#1}}		
\safemath{\hilseqspace}{\banachseqspace{2}}	
\newcommand{\banachfunspace}[1]{\setL^{#1}}	
\safemath{\hilfunspace}{\banachfunspace{2}}	
\safemath{\hilfunspacep}{\hilfunspace(\complexset)}	
\safemath{\schwarzspace}{\setS}				

\newcommand{\hadj}[1]{#1^{\star}}			

\safemath{\SNR}{\rho} 				
\safemath{\SINR}{\text{\sc sinr}} 				
\safemath{\No}{N_0}							
\safemath{\Es}{E_s}							
\safemath{\Eb}{E_b}							
\safemath{\EbNo}{\frac{\Eb}{\No}}
\safemath{\EsNo}{\frac{\Es}{\No}}
\safemath{\NoVar}{\variance}                 

\let\time\undefined
\safemath{\time}{\sct}						
\safemath{\dtime}{\sck}						
\safemath{\delay}{\sctau}					
\safemath{\ddelay}{\scl}						
\safemath{\doppler}{\scnu}					
\safemath{\ddoppler}{\scm}					
\safemath{\freq}{\scf}						
\safemath{\dfreq}{\scn}						
\safemath{\Dtime}{\Delta\time}
\safemath{\Dfreq}{\Delta\freq}
\safemath{\Ddtime}{\dtime}
\safemath{\Ddfreq}{\dfreq}
\safemath{\bandwidth}{\scB}
\safemath{\maxdoppler}{\doppler_{0}}			
\safemath{\maxdelay}{\delay_{0}}				
\safemath{\spread}{\Delta_{\CHop}}			
\DeclareMathOperator{\CHop}{\ensuremath{\opH}} 
\safemath{\kernel}{\rndk_{\CHop}}			
\safemath{\kernelp}{\kernel(\time,\time')}	
\safemath{\tvir}{\rndh_{\CHop}}				
\safemath{\tvirp}{\tvir(\time,\delay)}		
\safemath{\tvirc}{\conj{\rndh}_{\CHop}}		
\safemath{\tvtf}{\rndl_{\CHop}}				
\safemath{\tvtfp}{\tvtf(\time,\freq)}			
\safemath{\tvtfc}{\conj{\rndl}_{\CHop}}		
\safemath{\spf}{\rnds_{\CHop}}				
\safemath{\spfp}{\spf(\doppler,\delay)}		
\safemath{\spfc}{\conj{\rnds}_{\CHop}}		
\safemath{\bff}{\rndb_{\CHop}}				
\safemath{\bffp}{\bff(\doppler,\freq)}		

\safemath{\irc}{\scr_{\rndh}}				
\safemath{\tfc}{\scr_{\rndl}}				
\safemath{\spc}{\scr_{\rnds}}				
\safemath{\bfc}{\scr_{\rndb}}				

\safemath{\scaf}{\scc_{\rnds}}				
\safemath{\scafp}{\scaf(\doppler,\delay)}		
\safemath{\ccf}{\scc_{\rndl}}				
\safemath{\ccfp}{\ccf(\Dtime,\Dfreq)}			
\safemath{\cic}{\scc_{\rndh}}				
\safemath{\cicp}{\cic(\Dtime,\delay)}			

\safemath{\mi}{\scI}							
\safemath{\capacity}{\scC}					

\DeclareMathOperator{\Prob}{\opP}		

\safemath{\normal}{\mathcal{N}}			
\safemath{\jpg}{\mathcal{CN}}			
\safemath{\uniform}{\mathcal{U}}				
\safemath{\mchain}{\leftrightarrow}		

\safemath{\dB}{\,\mathrm{dB}}
\safemath{\dBm}{\,\mathrm{dBm}}
\safemath{\Hz}{\,\mathrm{Hz}}
\safemath{\kHz}{\,\mathrm{kHz}}
\safemath{\MHz}{\,\mathrm{MHz}}
\safemath{\GHz}{\,\mathrm{GHz}}
\safemath{\THz}{\,\mathrm{THz}}
\safemath{\s}{\,\mathrm{s}}
\safemath{\ms}{\,\mathrm{ms}}
\safemath{\mus}{\,\mathrm{\text{\textmu}s}}
\safemath{\ns}{\,\mathrm{ns}}
\safemath{\ps}{\,\mathrm{ps}}
\safemath{\meter}{\,\mathrm{m}}
\safemath{\mm}{\,\mathrm{mm}}
\safemath{\cm}{\,\mathrm{cm}}
\safemath{\nm}{\,\mathrm{nm}}
\safemath{\m}{\,\mathrm{m}}
\safemath{\W}{\,\mathrm{W}}
\safemath{\mW}{\, \mathrm{mW}}
\safemath{\J}{\,\mathrm{J}}
\safemath{\K}{\,\mathrm{K}}
\safemath{\bit}{\,\mathrm{bit}}
\safemath{\nat}{\,\mathrm{nat}}

\safemath{\define}{\triangleq}					

\providecommand{\inner}[2]{\ensuremath{\left\langle#1,#2\right\rangle}}
\safemath{\equivalent}{\sim}
\safemath{\distas}{\sim}					
\safemath{\sdiff}{\Delta}				
\safemath{\setdiff}{\setminus}				

\safemath{\reals}{\mathbb R}
\safemath{\positivereals}{\reals^{+}}
\safemath{\integers}{\mathbb Z}
\safemath{\posint}{\integers^{+}}
\safemath{\naturals}{\mathbb N}
\safemath{\posnaturals}{\naturals^{+}}
\safemath{\complexset}{\mathbb C}
\safemath{\rationals}{\mathbb Q}

\newcommand{\natseg}[2]{[#1 \text{\phantom{\tiny{.}}:\phantom{\tiny{.}}} #2]} 

\DeclarePairedDelimiter\floor{\lfloor}{\rfloor}

\safemath{\iSet}{\setI}
\safemath{\rel}{\bowtie}					
\safemath{\eqrel}{\sim}					
\safemath{\rlord}{\leq}					
\safemath{\slord}{<}						
\safemath{\rpord}{\preceq}				
\safemath{\rrpord}{\succeq}				
\safemath{\spord}{\prec}					
\safemath{\sig}{\sigma}					
\safemath{\metric}{d}					

\safemath{\setfun}{\Phi}					
\safemath{\measure}{\mu}					
\safemath{\altmeasure}{\lambda}					
\newcommand{\outerm}[1]{#1^{\star}}		
\newcommand{\innerm}[1]{#1_{\star}}		
\safemath{\omeasure}{\outerm{\measure}}		
\safemath{\imeasure}{\innerm{\measure}}		
\safemath{\aecol}{\colS^{\star}_{\measure}} 
\safemath{\emeasure}{\bar{\measure}_{0}}	
\safemath{\rmeasure}{\tilde{\measure}}	
\safemath{\bmeasure}{\measure_{0}}		
\safemath{\glength}{\measure_{\altfun}}	
\safemath{\lebmea}{\lambda}				
\safemath{\blebmea}{\lebmea_{0}}			
\safemath{\sfun}{s}						
\safemath{\absintspace}{\colL^{1}}		
\safemath{\sqintspace}{\colL^{2}}		
\safemath{\abssumspace}{l^{1}}		
\safemath{\sqsumspace}{l^{2}}		

\safemath{\sfield}{\setF}				
\safemath{\vectors}{\setV}				
\safemath{\vecspace}{(\vectors,\sfield)}	
\safemath{\linspace}{\setV}				
\safemath{\clinspace}{(\linspace,\sfield)} 
\safemath{\nspace}{\setU}				
\safemath{\metspace}{\setM}				
\safemath{\bspace}{\setB}				
\safemath{\ipspace}{\setG}				
\safemath{\hilspace}{\setH}				
\safemath{\blospace}{\setG}				
\safemath{\lop}{\opT}					
\safemath{\altlop}{\opS}					
\safemath{\nullsp}{\nullspace(\lop)}		
\safemath{\lfun}{l}						
\safemath{\altlfun}{g}					
\newcommand{\dual}[1]{#1^{'}}			
\safemath{\dsum}{\oplus}					
\safemath{\funspace}{\colL}				
\renewcommand{\adj}[1]{#1^{\dagger}}		
\safemath{\adjlop}{\adj{\lop}}			
\safemath{\hadjlop}{\hadj{\lop}}			
\safemath{\tow}{\xrightarrow{w}}			
\safemath{\tows}{\xrightarrow{w^{*}}}		
\safemath{\cparam}{\lambda}
\safemath{\lopl}{\lop_{\cparam}}		
\safemath{\iop}{\opI}					
\safemath{\resolop}{\opR}				
\safemath{\resolvent}{\resolop_{\cparam}(\lop)}	
\safemath{\reset}{\setQ}
\safemath{\spectrum}{\setS}
\safemath{\resolset}{\reset(\lop)}		
\safemath{\lopspec}{\spectrum(\lop)}		
\safemath{\pspec}{\spectrum_{p}(\lop)}	
\safemath{\cspec}{\spectrum_{c}(\lop)}	
\safemath{\rspec}{\spectrum_{r}(\lop)}	
\safemath{\ev}{\cparam}					
\newcommand{\specrad}[1]{r_{#1}}			
\safemath{\lopsrad}{\specrad{\lop}}		
\safemath{\pop}{\opP}					

\safemath{\specfam}{\colE}				
\safemath{\specop}{\opE_{\cparam}}		
\safemath{\altspecop}{\opE_{\mu}}		
\safemath{\vmulti}{\vecone}				
\safemath{\unitaryop}{\opU}				
\safemath{\sval}{\sigma}					
\safemath{\corrcoef}{\rho}				
\safemath{\sangle}{\theta}				

\let\time\undefined
\safemath{\iset}{\setI}				
\safemath{\shift}{\nu}
\safemath{\scale}{\alpha}
\safemath{\time}{t}
\safemath{\specfreq}{\theta}	
\newcommand{\transopgen}[1]{\opT_{#1}} 
\safemath{\transop}{\transopgen{\delay}}
\newcommand{\modopgen}[1]{\opM_{#1}}	
\safemath{\modop}{\modopgen{\shift}}
\newcommand{\dilopgen}[1]{\opD_{#1}}	
\safemath{\dilop}{\dilopgen{\scale}}
\safemath{\fram}{\setF}				
\safemath{\dfram}{\dual{\fram}}		
\safemath{\ufb}{B}					
\safemath{\lfb}{A}					
\safemath{\sop}{\hadj{\aop}}				
\safemath{\aop}{\opT}			
\safemath{\fop}{\opS}				
\safemath{\daop}{\tilde\opT}			
\safemath{\dsop}{\hadj{\tilde\opT}}				

\safemath{\ifop}{\inv{\fop}}			
\safemath{\rifop}{\fop^{-1/2}}			
\safemath{\transeq}{\setT}			
\safemath{\nfun}{\Phi}				
\safemath{\funvec}{\vecf}			
\safemath{\altfunvec}{\vecg}

\safemath{\samplespace}{\Omega}
\safemath{\probspace}{(\samplespace,\sfield,\Prob)}	
\safemath{\ccoef}{\rho}			

\safemath{\infstate}{\vecpi}				
\safemath{\typset}{\setA_{\epsilon}^{(N)}}	
\safemath{\expequal}{\doteq}				
\safemath{\code}{C}						
\safemath{\dstringset}{\setD^{\star}}		
\safemath{\cwlength}{l}					
\safemath{\elength}{L}					
\safemath{\extension}{C^{\star}}			
\safemath{\approaches}{\rightarrow}		
\safemath{\evnt}{\setA}					
\safemath{\altevnt}{\setB}					
\safemath{\rv}{\rndx}					
\safemath{\altrv}{\rndy}					
\safemath{\complexrv}{\rndu}					
\safemath{\altcrv}{\rndv}				
\safemath{\rvec}{\rvecx}					
\safemath{\altrvec}{\rvecy}				
\safemath{\crvec}{\rvecu}				
\safemath{\altcrvec}{\rvecv}				
\safemath{\variance}{\sigma^{2}}			
\safemath{\map}{T}						
\safemath{\jacobian}{\matJ}					
\safemath{\wvec}{\rvecw}					
\safemath{\wrv}{\rndw}					
\safemath{\orthmat}{\matQ}				
\safemath{\evmat}{\matLambda}			
\safemath{\identity}{\matidentity}		
\safemath{\innovec}{\vecv}				
\safemath{\convas}{\xrightarrow{\text{a.s.}}}	
\safemath{\convr}{\xrightarrow{\text{r}}}	
\safemath{\convp}{\xrightarrow{\text{P}}}	
\safemath{\convd}{\xrightarrow{\text{D}}}	
\safemath{\ltis}{\opL}				
\safemath{\ir}{h}					
\safemath{\tf}{\MakeUppercase{\ir}}	

\newcommand*{\fancyrefparlabelprefix}{par}		
\newcommand*{\fancyrefchalabelprefix}{cha}		
\newcommand*{\fancyrefapplabelprefix}{app}		
\newcommand*{\fancyrefthmlabelprefix}{thm}		
\newcommand*{\fancyreflemlabelprefix}{lem}		
\newcommand*{\fancyrefcorlabelprefix}{cor}		
\newcommand*{\fancyrefdeflabelprefix}{def}		
\newcommand*{\fancyrefproplabelprefix}{prop}		

\frefformat{vario}{\fancyrefparlabelprefix}{Part~#1}
\frefformat{vario}{\fancyrefchalabelprefix}{Chapter~#1}
\frefformat{vario}{\fancyrefseclabelprefix}{Section~#1}
\frefformat{vario}{\fancyrefthmlabelprefix}{Theorem~#1}
\frefformat{vario}{\fancyreflemlabelprefix}{Lemma~#1}
\frefformat{vario}{\fancyrefcorlabelprefix}{Corollary~#1}
\frefformat{vario}{\fancyrefdeflabelprefix}{Definition~#1}
\frefformat{vario}{\fancyreffiglabelprefix}{Figure~#1}
\frefformat{vario}{\fancyrefapplabelprefix}{Appendix~#1}
\frefformat{vario}{\fancyrefeqlabelprefix}{(#1)}
\frefformat{vario}{\fancyrefproplabelprefix}{Property~#1}

\newtheorem{thm}{Theorem}
\newtheorem{lem}{Lemma}

\theoremstyle{definition}
\newtheorem{dfn}{Definition}

\newtheorem*{rem}{Remark}

\newcommand{\xseq}{\hat u}	 	
\newcommand{\yseq}{\hat v}	 	
\newcommand{\zseq}{\hat w}	 	
\newcommand{\xseqf}{u}	 	
\newcommand{\yseqf}{v}	 	
\newcommand{\zseqf}{w}	 	

\newcommand{\otime}{t}	 	
\newcommand{\osp}{\vecv}	 	
\newcommand{\isp}{\vecw}	 	
\newcommand{\ospc}{v}	 	
\newcommand{\gfun}[2]{\chi(#1,#2)}	 	
\newcommand{\cnst}[1]{C(#1)}	 	
\newcommand{\cnstl}[1]{C_L(#1)}	 	
\newcommand{\cnstnum}{C}	 	
\newcommand{\cutf}{{\alpha}}	 	
\newcommand{\cutfn}{{\alpha}}	 	
\newcommand{\calpha}{C(\cutf)}	 	
\newcommand{\calphan}{C(\cutf)}	 	
\newcommand{\shld}{\beta}	 	
\newcommand{\rsp}{r}	 	
\newcommand{\rspt}{\omega}	 	
\newcommand{\vinp}{\vecx}	 	
\newcommand{\vinpdd}{{\vecx_\mathrm{2D}}}	 	
\newcommand{\vinpc}{x}	 		
\newcommand{\vinpfilt}{\vecs}	 	
\newcommand{\vinpfiltc}{s}	 		
\newcommand{\vinpfiltnn}{\tilde\vecs}	 	
\newcommand{\vinpfiltnndd}{\tilde\vecs_\mathrm{2D}}	 	
\newcommand{\vinpest}{\hat\vecx}
\newcommand{\vinper}{\vech}
\newcommand{\vinperc}{h}	
\newcommand{\idim}{N}			
\newcommand{\sparsity}{S}       
\newcommand{\ninterv}{d}       
\newcommand{\voutdd}{\vecy_\mathrm{2D}}	 	
\newcommand{\odim}{n}			
\newcommand{\odimn}{n}			
\newcommand{\fc}{f_c}			
\newcommand{\fcn}{f_c}			
\newcommand{\fccoh}{\bar{f}_c}			
\newcommand{\lambdac}{\lambda_c}			
\newcommand{\tlambdac}{{\tilde\lambda}_c}			
\newcommand{\lambdacn}{\lambda_c}			

\newcommand{\support}{\mathrm{supp}}			
\newcommand{\projnoise}{\vecz}			
\newcommand{\dualp}{\vecq}			
\newcommand{\dualpc}{q}			
\newcommand{\Fdd}{\matF_{\mathrm{2D}}}			
\newcommand{\Q}{\matQ}			
\newcommand{\hQdd}{{\hat\matQ}_{\mathrm{2D}}}			
\newcommand{\rclass}[2]{\setR_1(#1,#2)}		
\newcommand{\rclasstwo}[2]{\setR_2(#1,#2)}		
\newcommand{\rclassidxf}[4]{\setR_1(#1,#2;#3,#4)}			
\newcommand{\rclassp}[2]{\setR_1^+(#1,#2)}			
\newcommand{\rclasspdd}[2]{\setR^+_2(#1,#2)}			
\newcommand{\SRF}{\mathrm{SRF}}			
\newcommand{\SRFn}{\mathrm{SRF}}			
\newcommand{\NAF}{\mathrm{NAF}}			
\newcommand{\MC}{\mathrm{MC}}			
\newcommand{\glow}{g_{\mathrm{low}}}			
\newcommand{\flow}{f_{\mathrm{low}}}			
\newcommand{\FDO}[1]{\Delta^{\rspt}_{#1}[\glow](t)}			
\newcommand{\FDOarg}[2]{\Delta^{\rspt}_{#1}[\glow]\left(#2\right)}			
\newcommand{\ucirc}{\mathbb{T}}			
\newcommand{\inpnorm}[1]{\vecnorm{#1}}			
\newcommand{\outnorm}[1]{\vecnorm{#1}}			

\usepackage{hyperref}

\hypersetup{
    unicode=false,          
    pdftoolbar=true,        
    pdfmenubar=true,        
    pdffitwindow=false,     
    pdfstartview={FitH},    
    pdfauthor={Veniamin Morgenshtern},     
    pdfcreator={Veniamin Morgenshtern},   
    pdfproducer={Producer}, 
    pdfnewwindow=true,      
    colorlinks=true,       
	 linkcolor=blue!60!black,          
    citecolor=blue!60!black,        
    filecolor=blue!60!black,      
    urlcolor=blue!60!black,
	 hypertexnames=false
}
\usepackage{autonum}
 
\begin{document}
	\acrodef{NAF}{noise amplification factor}
	\acrodef{MC}{modulus of continuity}
	\acrodef{SNR}{signal-to-noise ratio}
	\acrodef{SRF}{super-resolution factor}
	\acrodef{JPG}{jointly proper Gaussian}
	\acrodef{PDF}{probability density function}
	\acrodef{DFT}{discrete Fourier transform}
	\acrodef{FFT}{fast Fourier transform}
	\acrodef{DoF}{degree-of-freedom}
	\acrodef{RHS}{right-hand side}
	\acrodef{LHS}{left-hand side}
	\acrodef{IO}{input-output}
	\acrodef{RV}{random variable}
	\acrodef{SVD}{singular-value decomposition}
	\acrodef{MS}{mass spectrometry}
	\acrodef{MCR}{mass-to-charge ratio}
	\acrodef{IOS}{incoherent optical system}
	\acrodef{A/D}{analog-to-digital}
	\acrodef{2D}{two-dimensional}
	\acrodef{1D}{one-dimensional}
	\acrodef{LP}{linear program}

\title{Super-Resolution of Positive Sources:\\ the Discrete Setup}

\pdfinfo{
 	/Title		(Stable Super-Resolution of Positive Sources: the Discrete Setup)
 	/Author		(Veniamin~I.~Morgenshtern and Emmanuel J. Cand{\`e}s)
 	/Keywords	()
}

\author{
\parbox{\linewidth}{\centering
Veniamin~I.~Morgenshtern$^{1}$ and Emmanuel J. Cand{\`e}s$^{1,2}$ \\\vspace{0.3cm}
$^{1}$Dept.~of Statistics, Stanford University, CA \\
$^{2}$Dept.~of Mathematics, Stanford University, CA 
}
}

\newcommand{\ejc}[1]{\textcolor{blue}{[EJC: #1]}}
\newcommand{\ven}[1]{\textcolor{orange}{[VM: #1]}}

\maketitle
\begin{abstract} {In single-molecule microscopy it is necessary to
    locate with high precision point sources from noisy observations
    of the spectrum of the signal at frequencies capped by $\fc$,
    which is just about the frequency of natural light. This paper
    rigorously establishes that this super-resolution problem can be
    solved via linear programming in a stable manner. We prove that
    the quality of the reconstruction crucially depends on the
    Rayleigh regularity of the support of the signal; that is, on the
    maximum number of sources that can occur within a square of side
    length about $1/\fc$. The theoretical performance guarantee is
    complemented with a converse result showing that our simple convex
    program convex is nearly optimal. Finally, numerical experiments
    illustrate our methods.}
\end{abstract}

\section{Introduction}

The problem of super-resolution arises in many areas of science and
engineering including mass-spectrometry, radar imaging, and wireless
communication. In optics, for example, the natural resolution of
microscopes is inversely proportional to the wavelength of light used
for observation. This happens because of the diffraction of light, and
makes it fundamentally difficult to study sub-wavelength features of
the object; e.g.~to resolve nearby sources located at distances
smaller than the diffraction limit.  This paper is about this problem:
namely, the super-resolution of positive sources, e.g.~fluorescing
molecules as in single-molecule imaging.

Formally, consider a high-frequency signal
\begin{equation}
	\label{eq:xcont}
	x(\isp)=\sum_i x_i \delta(\isp-\isp_i)
\end{equation}
consisting of positive point sources located at unknown positions
$\isp_i$ and of unknown intensity $x_i>0$. The signal is observed
through a convolution of the form
\begin{equation}
\label{eq:firstconv}
s(\osp)=\int \flow(\osp-\isp) x(\isp) d\isp+z(\osp),
\end{equation}
where $\flow(\cdot)$ is a low-frequency kernel that erases the
high-frequency components of the signal and $z(\cdot)$ is noise. The
goal of super-resolution is to accurately estimate $x(\cdot)$,
i.e.,~the source locations and intensities.

\subsection{Super-resolution microscopy}
\label{sec:microscopy}

Since our mathematical models and theoretical results are motivated by
very concrete contemporary problems in single-molecule imaging, we
find it best to pause and introduce some background material; for
details beyond those we provide below, please check the wonderful book
by J.~Goodman~\cite{goodman88}.

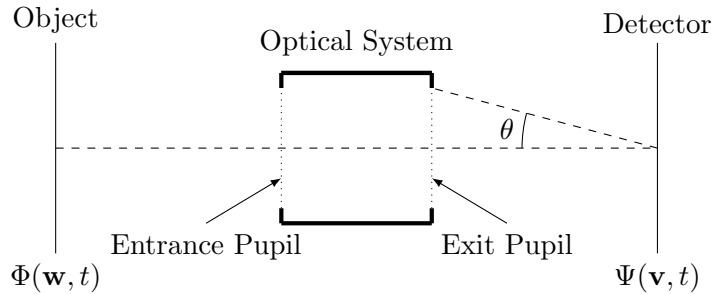
\begin{figure}[ht]
        \centering
\begin{tikzpicture}[scale=2,auto]
	\begin{scope}[ultra thick]
  		\draw (0,0) -- (0,0.1);
  	  	\draw (0,0.9) -- (0,1);
		\draw[dotted,thin] (0,0.1) -- (0,0.9);
		\draw[dotted,thin] (1,0.1) -- (1,0.9);
  		\draw (1,0) -- (1,0.1);
  	  	\draw (1,0.9) -- (1,1);
  	  	\draw (0,0) -- (1,0);
  	  	\draw (0,1) -- node[above=2pt] {Optical System} (1,1);
  	  	\draw (0,1) -- (1,1);
  	  	
  \end{scope}
  
  \begin{scope}[thin]
	  \draw[dashed] (-1.5,0.5) -- (2.5,0.5);
	  \draw[-,>=latex] (-1.5,-0.2) node[anchor=north] {$\Phi(\isp,\otime)$} -- (-1.5,1.2) node[anchor=south] {Object};
	  \draw[-,>=latex] (2.5,-0.2) node[anchor=north] {$\Psi(\osp,\otime)$} -- (2.5,1.2) node[anchor=south] {Detector};
	  \draw[dashed] (2.5,0.5) -- (1,0.9);
	  \draw[<-,>=latex] (0,0.3) -- (-0.5,0) node[anchor=north] {Entrance Pupil};
	  \draw[<-,>=latex] (1,0.3) -- (1.5,0) node[anchor=north] {Exit Pupil};
	  \draw(1.6,0.5) arc(180:165:0.9);
	  \draw(1.5,0.63) node {$\theta$};
	  
  \end{scope} 
\end{tikzpicture}
\caption{Model of an optical system.}
\label{fig:optics}
\end{figure}

To understand where \eqref{eq:firstconv} comes from, we derive the
input-output relation of a simple imaging system as shown
in~\fref{fig:optics}. While the laws of optics are governed by
Maxwell's equations, which are linear, the vectorial nature of the
electric and magnetic fields can be neglected in Fourier optics and
the physics fully described via the time-varying
phasor~\cite[Sec.~3.2]{goodman88}, a term assigned to any of the three
components of these two fields. Assume that a narrow-band (not
necessarily monochromatic) light is used for illumination, and let
$\Phi(\isp,\otime)$ and $\Psi(\osp,\otime)$ respectively denote the
input/output phasors describing the field emitted by the object being
imaged and the field generated at the receiver of the system.  Here,
$\isp,\osp\in\reals^2$ are indexing spatial coordinates in the object
plane and in the detector plane, respectively, and $\otime\in\reals$
is indexing time. We assume, for convenience, that the phasors
$\Phi(\isp,\otime),\Psi(\osp,\otime)$ have been frequency-shifted (as
a function of $\otime$) to be centered around the mean frequency of
the optical wave~\cite[p.~132]{goodman88}, so that, for example,
$E(\isp,\otime)=\re[\Phi(\isp,\otime) e^{2\pi\iu\bar\nu \otime}]$,
where $E$ is one of the components of the electric field and $\bar\nu$
is the average frequency of emitted light.  The diffraction of light
in the optical system can be described by the Fraunhofer approximation
leading to~\cite[Eq (6-6)]{goodman88}
\begin{equation}
	\label{eq:opticel1}
	\Psi(\osp,\otime)= \int h(\osp-\isp) \Phi(\isp,\otime) d\isp,
\end{equation}
where $h(\osp)$ is the point-spread function (PSF) of the optical
system. In general, the Fourier transform of $h(\cdot)$ is
proportional to the indicator function of the aperture and because the
aperture is finite, $h(\cdot)$ is band-limited. To be concrete, assume
that the entrance and the exit pupils in \fref{fig:optics} are
square. In this case~\cite[Sec. 6.2.2]{goodman88}
\begin{equation}
	\label{eq:kerdirichlet}
        h(\osp) \propto \frac{1}{\sqrt{2 \fccoh}}\frac{\sin(2 \pi \fccoh \ospc_1)}{\pi \ospc_1}	\frac{1}{\sqrt{2 \fccoh}}\frac{\sin(2 \pi \fccoh \ospc_2)}{\pi \ospc_2},\quad \osp=\tp{[\ospc_1,  \ospc_2]}.
\end{equation} 
The spatial frequency cut-off of the optical system is given by
\begin{equation}
	\fccoh=\frac{\sin(\theta)}{\lambda},
\end{equation}
where $\lambda$ is the wavelength of emitted light (average wavelength
in the narrow-band illumination case) and $\theta$ is half of the
angle spanned by the exit pupil as seen from the center of the image
plane (see~\fref{fig:optics}). Note that due to the narrow-band
illumination assumption, $h(\cdot)$ depends upon the average
wavelength of the optical wave, but not upon the specific frequencies
in the illuminating spectrum so that the system model is described by
the simple convolution equation~\fref{eq:opticel1}.

In optics, the carrier frequency $\bar\nu\sim 500 \THz$ is much higher
than the frequency $f_\mathrm{HET}\sim 10 \GHz$, which electronic
components can respond to, e.~g.~the frequency of heterodyne used to
down-convert the signal. Consequently, in optics only the time-average
of the instantaneous intensity of received light (called received
intensity) is directly observable \cite[Eq (6-8)]{goodman88}:
\begin{equation}
	\label{eq:quadratic1}
	\tilde s(\osp)\define\left<\Psi(\osp,\otime)\conj{\Psi}(\osp,\otime)\right>,
\end{equation}	
where $\conj{\Psi}(\cdot)$ denotes the complex conjugate of
$\Psi(\cdot)$ and $\left<\cdot\right>$ stands for time averaging:
\begin{equation}
	\left<g(\otime)\right>=f_\mathrm{HET} \int_{0}^{1/f_\mathrm{HET}} g(\otime) d\otime.
\end{equation}
In a majority of microscopy applications, the object emits incoherent
light. Mathematically, this situation is described by assuming that
frequencies of spatially separated emitters vary in statistically
independent fashions. This idealized property may be represented by
the equation~\cite[Eq (6-14)]{goodman88}
\begin{equation}
	\label{eq:noncohdelta1}
	\left<\Phi(\isp,\otime)\conj{\Phi}(\isp',\otime)\right> =   \delta(\isp-\isp') x(\isp).
\end{equation}
The quantity $x(\isp)$ is
the time-average of the instantaneous intensity of light emitted by
the object and is called {\em emitted intensity}. Substituting
\fref{eq:opticel1} into \fref{eq:quadratic1} and then using
\fref{eq:noncohdelta1} we obtain the following input-output relation
\begin{equation}
	\label{eq:IOnoncohcont1}
	\tilde s(\osp) = \int \flow(\osp-\isp) x(\isp) d\isp,\quad \flow(\osp-\isp)\define\abs{h(\osp-\isp)}^2. 
\end{equation}
Observe that \fref{eq:IOnoncohcont1} is a linear convolution equation
with respect to emitted intensity; compare to~\fref{eq:opticel1},
which is a linear convolution equation with respect to the components
of the emitted field. The low-frequency kernel $\flow(\cdot)$ is the
square of the \ac{2D} sinc kernel \fref{eq:kerdirichlet} and has a
spatial frequency cut-off at $\fc=2\fccoh$ (twice that of the kernel
$h(\cdot)$).  The emitted intensity $x(\cdot)$ is a nonnegative
function, a property that is crucially important for all results in
this paper. Finally, the $\ell_1$ norm of the signal, 
\begin{equation}
	\label{eq:wenergy}
	\onenorm{x(\cdot)}=\int \abs{x(\isp)} d\isp, 
\end{equation}
has the meaning of cumulative emitted intensity or total energy of
light emitted per second. As a side remark, note that when the sample
is illuminated by coherent light, as in X-ray crystallography, the
resulting input-output relations is no longer linear, in stark
contrast to \fref{eq:IOnoncohcont1}, and the phase retrieval problem
needs to be solved. For the interested reader, this point is explained
in~\fref{app:cohopt}. 

Our goal is to reconstruct the signal $x(\cdot)$ from the observations
$\tilde s(\cdot)$ in \fref{eq:IOnoncohcont1}. Without additional
structural assumptions on $x(\cdot)$, this is clearly not possible,
because the high-frequency components of $x(\cdot)$ are lost. The
details of $x(\cdot)$ that are smaller than the Rayleigh diffraction
limit,\footnote{The specific value of the constant, $1.22$, is largely
  a historical convention; the point here is that the details of the
  image that are much smaller than $1/\fccoh$ are blurred.}
$1.22/\fccoh$, cannot be distinguished~\cite[Sec 6.5.2]{goodman88}.
In single-molecule
microscopy~\cite{betzig06-06,dickson97on,klar00fluor}, a modern
imaging technique, the signal $x(\cdot)$ consists of several disjoint
molecules emitting light. Here, the size of each molecule is about
$4\nm$, which is much smaller than $1/\fccoh\approx 200\nm$, and yet
it is absolutely necessary to estimate the locations of these molecules
with precision that is significantly higher than the Rayleigh
diffraction limit. 

The main contribution of this paper is to show that under the
structural model \fref{eq:xcont}, it is possible to estimate $x(\isp)$
via linear programming stably from noisy data---all imaging systems
are fundamentally noisy---with resolution beyond the diffraction limit. 
Further, the quality of estimation fundamentally depends on how
regularly (in the sense explained in~\fref{sec:mainres}) the
sources/molecules are distributed in the image domain.

\subsection{Mathematical models and methods}

The super-resolution theory developed in this paper is discrete, which
means that the input signal $x(\cdot)$ is assumed to be supported on a
fine grid. The nonzero elements of this discrete signal are 
suggestively called ``spikes''.  In optics, there is no grid, of
course; the spikes in~\fref{eq:xcont} can be in arbitrary (continuous)
locations, and the companion paper~\cite{candes-14-2} shows how to
generalize our key result to the continuous setting. In truth, the
analysis of the continuous-space problem is far more technical than
that presented here; however, the final result---the stability
estimate in~\fref{eq:mainbnd}---is essentially the same. For now, the
advantage of working with a discrete model is that we can explain the
key concepts without bothering with heavy mathematical machinery.

\subsubsection{Discrete setup}

A noiseless discrete model is of the form
\begin{equation}
		\label{eq:IO3}
		\vinpfiltnn=\matQ\vinp,
\end{equation}
where $\vinp$ is either a one- or two-dimensional discrete array of
intensities, $\matQ$ models the (discrete) convolution equation and
$\vinpfiltnn$ is the output data, assumed to be of the same
dimension(s) as the input vector $\vinp$. We have already seen
examples of PSFs or convolutions; for instance, $\flow(\cdot)$
in~\fref{eq:IOnoncohcont1} is the square of the sinc kernel (in each
direction), the sinc kernel being an ideal low-pass filter whose
frequency response is a box function. Therefore, the frequency
response of $\flow(\cdot)$ is a triangle function in \acs{1D} and a
pyramid in 2D. In
\fref{eq:Q1dflat},~\fref{eq:Q1dtri},~\fref{eq:Q2dflat} and
\fref{eq:Q2dtri} below, we consider natural PSFs in one and two
dimensions so that in the remainder, $\matQ$ in \fref{eq:IO3} or
\fref{eq:IO3altnoise1d} may be given by any of these.

\subsubsection{Noise} 

In modern microscopy
applications, the intensities of emitted/received light are very low
and in such regimes, the main source of noise is due to
quantum-mechanical effects. We have argued that a component of
$\vinpfiltnn$ represents the expected number of photons to be recorded
per unit time at a given pixel on the detector. The actual number of
photons detected may be modeled as a Poisson-distributed random
variable so that $\vinpfilt\sim\mathrm{Pois}\left(\vinpfiltnn\right)$,
meaning that we have independent Poisson variables with means given by
\fref{eq:IO3}. In this paper, we shall work with a slightly more
general signal-dependent additive noise $\projnoise=\vinpfilt-\Q\vinp$
so that the \ac{IO} relation becomes
\begin{equation}
	\label{eq:IO3altnoise1d}
	\vinpfilt=\Q\vinp+\projnoise.
\end{equation}

\subsubsection{Recovery}

Our recovery method from the observations $\vinpfilt$ in
\fref{eq:IO3altnoise1d} is extremely simple: solve
\begin{align}
	\label{eq:find0}\tag{$\mathrm{CVX}$}
	\text{min}_{\hat\vinp} \quad \vecnorm{\vinpfilt-\matQ\hat\vinp}_1 \quad \text{s.t.}\quad \hat\vinp\ge \veczero.
\end{align}
In other words, we are looking for a superposition of positive sources
such that the mismatch in received intensities is minimum. Note that
this method does not make any assumption about the signal and does not
make use of any knowledge other than the received data $\vinpfilt$ and
the PSF $\matQ$. Furthermore, (CVX) is a simple convex optimization
program, which can be recast as a linear program since both
$\vinpfilt$ and $\matQ$ are real valued.

\subsubsection{Examples of PSFs}

We now discuss various models for the discrete convolution equation
\eqref{eq:IO3altnoise1d}.

\paragraph{1D model with flat spectrum.} In our first example,
$\vinp=\tp{[\vinpc_0\cdots\vinpc_{\idim-1}]}\in \reals^\idim$ is a
one-dimensional array, and $\matQ$ is an ideal low-pass filter in the
sense that it has a flat spectrum with a sharp cut-off at
$\fc$. Formally,
\begin{equation}
	\label{eq:Q1dflat}
	\matQ= 
        \matQ_{\mathrm{flat,1D}}=\herm\matF\hat\matQ_{\mathrm{flat,1D}}\matF, 
\end{equation}
where
\[ [\matF]_{k,l} = \frac{1}{\sqrt{N}} \, e^{-\iu 2\pi k l/N}, \quad
{-\idim/2+1\le k\le \idim/2}, \,\,  {0\le l\le \idim-1,}
\]
is the $N \times N$ \ac{DFT} and
$\hat\matQ_{\mathrm{flat,1D}}=\diag(\tp{[\hat p_{-\idim/2+1}\cdots
  \hat p_{\idim/2}]})$ with
\begin{equation}
	\label{eq:filter3p}
	\hat p_k=
	\begin{cases}
		1, & k=-\fc,\ldots, \fc,\\
		0, & \text{otherwise}.
	\end{cases}
\end{equation}
The wavelength $\lambdacn\define 1/\fcn$ gives the width of the
convolution kernel represented by~$\matQ$. We assume throughout the
paper that $\idim$ is even for simplicity.

\paragraph{1D model with triangular spectrum.}
The discrete one-dimensional analog of our imaging system with
incoherent light \fref{eq:IOnoncohcont1} is given by~\fref{eq:IO3},
where $\matQ$ is as follows:
\begin{equation}
	\label{eq:Q1dtri}
	\matQ= \matQ_{\mathrm{tri,1D}}=\herm\matF\hat\matQ_{\mathrm{tri,1D}}\matF,
\end{equation}
$\hat\matQ_{\mathrm{tri,1D}}=\diag(\hat\vecq)$ with $\hat\vecq=\tp{[\hat q_{-\idim/2+1}\cdots \hat q_{\idim/2}]}$ and
\begin{equation}
	\label{eq:filter2}
	\hat q_k=
	\begin{cases}
		1-\frac{\abs{k}}{\fcn+1}, & k=-\fcn,\ldots, \fcn\\
		0, & \text{otherwise}.
	\end{cases}
\end{equation}
In this model, the nonzero elements of $\vinp$ represent the molecules
at the corresponding locations (on the grid) whereas the components of
$\vinpfiltnn$ represent the intensity of light measured at the
corresponding pixel on the detector.

\paragraph{2D model with flat spectrum.}
Similarly, the \ac{2D} model with a flat spectrum reads 
\begin{equation}
	\label{eq:IO3alt1dd}
	\vinpfiltnndd=\herm{\Fdd}\hQdd\Fdd \vinpdd,
\end{equation}
where
$\Fdd:\complexset^\idim\times\complexset^\idim\to\complexset^\idim\times\complexset^\idim$ is the linear operator that implements the \ac{2D} Fourier transform and acts according to
\begin{equation}
	[\Fdd\vinpdd]_{k_1,k_2}=\frac{1}{\idim}\sum_{l_1=0}^{\idim-1} \sum_{l_2=0}^{\idim-1} \vinpc_{l_1,l_2} e^{-\iu 2\pi (k_1 l_1+k_2 l_2)/\idim}
\end{equation}
and
$\hQdd:\complexset^\idim\times\complexset^\idim\to\complexset^\idim\times\complexset^\idim$
is the diagonal operator in the Fourier domain,
\begin{equation}
	\label{eq:hQdd}
	[\hQdd\voutdd]_{k_1,k_2}=\hat p_{k_1-\idim/2} \hat p_{k_2-\idim/2} \,\, [\voutdd]_{k_1,k_2}.
\end{equation}
To keep the same notation, define $\vinp=\vectorize{\vinpdd}$ and
$\vinpfiltnn=\vectorize{\vinpfiltnndd}$, where the $\vectorize{\cdot}$
operation stacks the columns of a matrix into a tall vector. Using the
properties of the Kronecker product,~\fref{eq:IO3alt1dd} can be
written as~\fref{eq:IO3} with
\begin{equation}
	\label{eq:Q2dflat}
	\matQ= \matQ_{\mathrm{flat,2D}}=(\herm\matF\kron\herm\matF)(\hat\matQ_{\mathrm{flat,1D}}\kron \hat\matQ_{\mathrm{flat,1D}})(\matF\kron\matF).
\end{equation}

\paragraph{2D model with triangular spectrum.}
With the vectorized notation, the 2D model with triangular spectrum
can be written as~\fref{eq:IO3} with
\begin{equation}
	\label{eq:Q2dtri}
	\matQ= \matQ_{\mathrm{tri,2D}}=(\herm\matF\kron\herm\matF)(\hat\matQ_{\mathrm{tri,1D}}\kron \hat\matQ_{\mathrm{tri,1D}})(\matF\kron\matF).
\end{equation}

\subsubsection{Intensity normalization}

It follows from our earlier discussion that for incoherent light
(models with triangular spectra), we may interpret $\onenorm{\vinp}$
as the total intensity of light emitted by the object. Similarly,
$\onenorm{\vinpfiltnn}$ is the total intensity of light observed at
the receiver. Letting $[\vecq_0\cdots \vecq_{\idim-1}]$ denote the
columns of $\matQ$, \fref{eq:filter2} guarantees that
$\onenorm{\vecq_l}=1$ for all $l$. To see this, first note that $\vecq_l$ is a shifted version of $\frac{1}{\sqrt{N}}\herm\matF\hat\vecq$ so that $\onenorm{\vecq_l}=\onenorm{\frac{1}{\sqrt{N}}\herm\matF\hat\vecq}$. Next, write $\hat\vecq=\hat\vecd\conv \hat\vecd$ where $\hat\vecd=\tp{[\hat
  d_{-\idim/2+1}\cdots \hat d_{\idim/2}]}$ and
\begin{equation}
	\hat d_k=
	\begin{cases}
		\sqrt{\frac{1}{\fc+1}}, & k=-\fcn/2,\ldots, \fcn/2,\\
		0, & \text{otherwise},
	\end{cases}
\end{equation}
and $\conv$ denotes the discrete convolution.  Finally, use the
convolution theorem to conclude
\begin{equation}
	\onenorm{\vecq_l}=\left\lVert\frac{1}{\sqrt{N}}\herm\matF\hat\vecq\right\lVert_1=\left\lVert\frac{1}{\sqrt{N}}\herm\matF(\hat\vecd\conv \hat\vecd)\right\lVert_1 = 
	\left\lVert\overline{\herm\matF\hat\vecd}\odot\herm\matF\hat\vecd\right\lVert_1= 
	\twonorm{\herm\matF\hat\vecd}^2=1,
\end{equation}
where $\odot$ denotes the element-wise product and $\overline{\mathbf{a}}$
takes conjugate element-wise.  Therefore, using that $x_l\ge 0$ and
$\vecq_l\ge\veczero$ for all $l$,
\begin{equation}
	\label{eq:norm}
	\onenorm{\vinpfiltnn}=\Bigl\|\sum_{l=0}^{\idim-1} x_l \vecq_l \Bigr\|_1  = \sum_{l=0}^{\idim-1} x_l \onenorm{ \vecq_l} = \sum_{l=0}^{\idim-1}  x_l= \onenorm{\vinp}. 
\end{equation}
Hence, our normalization is such that the intensity of light (emitted
energy per second) is conserved in the system. In the models
\fref{eq:filter3p} and \fref{eq:hQdd} with a flat spectrum the
$\ell_1$ norm of the signal is not conserved.

\subsection{Notation}
Sets are denoted by calligraphic letters $\setA, \setB$, and so on.
Boldface letters $\matA,\matB,\ldots$ and $\veca,\vecb,\ldots$ denote
matrices (or linear operators) and vectors, respectively.  The element
in the $i$-th row and $j$-th column of a matrix $\matA$ is
$\matAc_{ij}$ or $[\matA]_{i,j}$, and the $i$-th element of the
vector~$\vectr$ is $\vectrc_i$ or $[\vectr]_i$.  For a vector
$\vectr$, $\diag(\vectr)$ stands for the diagonal matrix that has the
entries of $\vectr$ on its main diagonal.  The superscripts~$\tp{}$
and~$\herm{}$ stand for transposition and Hermitian transposition,
respectively. 
For a finite set $\setI$, we write $\abs{\setI}$ for the cardinality.
For two functions~$\fun(\cdot)$ and~$\altfun(\cdot)$, the
notation~$\fun(\cdot)=\landauO(\altfun(\cdot))$ means
that~$\limsup_{\vrbl\to \infty} \abs{\fun(\vrbl)/\altfun(\vrbl)}$ is
bounded.  For $x\in\reals$, $\lceil x\rceil\define
\min\{m\in\integers\mid m\geq x\}$.  We use $\natseg{l}{k}$ to
designate the set of natural numbers $\left\{l, l+1,\ldots,k\right\}$.
The expectation operator is $\Ex{}{\cdot}$.  For a vector
$\vectr\in\complexset^n$, $\onenorm{\vectr}=\sum_{j=0}^{n-1}
\abs{\vectrc_j}$ and $\twonorm{\vectr}=\sqrt{\sum_{j=0}^{n-1}
  \abs{\vectrc_j}^2}$ denote the $\ell_1$ and $\ell_2$ norms,
respectively; $\vecnorm{\vectr}$ means either $\onenorm{\vectr}$ or
$\twonorm{\vectr}$. The number of nonzero elements of a vector
$\vectr$ is $\zeronorm{\vectr}$.  For a matrix
$\matA\in\complexset^{n\times n}$, the operator norm is defined as
$\vecnorm{\matA}_{1,op}=\max_{i} \sum_{j=0}^{n-1} \abs{a_{ij}}$ and
$\vectorize{\matA}$ denotes the $n^2$-dimensional vector obtained by
stacking the columns of $\matA$. For vectors
$\veca$ and $\vecb$, $\veca\odot\vecb$ denotes the element-wise product; $\veca\conv\vecb$ denotes the discrete convolution; the Kronecker
product of matrices $\matA$ and $\matB$ is denoted as
$\matA\kron\matB$.

\section{Main results} 
\label{sec:mainres}
Consider the 1D model for concreteness. From~\fref{eq:filter3p},
\fref{eq:filter2} we see that we have access to $\odimn=2\fcn+1$
low-frequency observations while the total number of
degrees-of-freedom in $\vinp$ is $\idim$. The ratio $\SRFn\define
\idim/\odimn$ is called the {\em super-resolution factor} (SRF); this
is the ratio between $1/n$ and $1/N$, the scale at which we have data
and that at which we wish to see details.

As we will review below, the sparsity condition
$\zeronorm{\vinp}<\odimn/2$ is sufficient for recovery of $\vinp$ when
there is no noise. If there is noise, it turns out that sparsity is
not sufficient as our ability to estimate $\vinp$ from $\vinpfilt$ in
a stable way fundamentally depends on how regular the positions of the
spikes are, i.e., how many spikes may be clustered close together.

\subsection{Rayleigh regularity}
\label{sec:reyleigh}

Suppose we are in $D$ dimensions and think of our discrete signal
$\vinp \in \complexset^{N^D}$ as samples on the $D$-dimensional grid
$\{0,1/\idim,\ldots,1-1/\idim\}^D\subset\ucirc^D$, where $\ucirc^D$ is
the $D$-dimensional (periodic) torus---the circle in 1D. In this
paper, we can think of the ambient dimension $D$ as being either one
or two.  We introduce a definition of Rayleigh regularity inspired
by~\cite[Def.~1]{donoho92-09}.

\begin{dfn}[\bf Rayleigh regularity]
\label{def:rreg}
Fix $\idim, \odim$ and set $\lambda_c = 1/f_c = 2/(n-1)$. We say that
the set of points $\setT\subset\{0,1/\idim,\ldots,1-1/\idim\}^D
\subset\ucirc^D$ is Rayleigh regular with parameters $(\ninterv,\rsp)$
and write $\setT\in \setR_{D}(d,r; N, n)$ if it may be partitioned as
$\setT=\setT_1\union\ldots\union\setT_r$ where the $\setT_i$'s are
disjoint, and each obeys a minimum separation constraint:
	\begin{enumerate}
		\item for all $1\le i<j\le \rsp$, $\setT_i\intersect\setT_j=\emptyset$;
		\item for all square subsets $\setD\subset\ucirc^D$ of
                  sidelength $d\lambdac/2$ and all $i$,
		\begin{equation}
			\abs{\setT_i\intersect \setD}\le 1.
		\end{equation}
	\end{enumerate}
        When no ambiguity arises, we will shortly write
        $\setR_{D}(d,r)$ instead of $\setR_{D}(d,r; N, n)$. 

        With a slight abuse of notation, it is also convenient to
        define a set of Rayleigh regular signals (and nonnegative
        Rayleigh regular signals) with parameters $(\ninterv,\rsp)$:
\begin{align}
  \setR_D(d,r) &=\{\vinp\in\complexset^{\idim^D}: \support(\vinp)\in \setR_{D}(d,r)\},\\
  \setR^+_D(d,r) &=\{\vinp\in\reals^{\idim^D}_+:
  \support(\vinp)\in\setR_{D}(d,r)\},
\end{align}
where $\support(\vinp)$ is the support of $\vinp$ (the locations on
grid where $\vinp$ does not vanish). 
\end{dfn}

\begin{rem}
  Intuitively, in 1D, $\vecx\in\rclass{\ninterv}{\rsp}$ simply means
  that the signal $\vecx$ contains no more than $\rsp$ spikes in any
  $\ninterv$ consecutive Nyquist intervals; a Nyquist interval being
  of length $\lambdac/2$, which corresponds to the Nyquist-Shannon
  sampling rate of a signal that is band-limited to $[-\fc,\fc]$.
  \fref{fig:rclass} illustrates these concepts for different parameter
  values.\footnote{Clearly,
    $\rclass{d}{\rsp_1}\subset\rclass{d}{\rsp_2}$ for $\rsp_1\le
    \rsp_2$ and $\rclass{d_1}{\rsp}\subset\rclass{d_2}{\rsp}$ for
    $d_1\ge d_2$.}
\end{rem}

We discuss some examples of Rayleigh regular signals and first
consider $\vinp\in\rclass{1}{1}$. This signal may contain one spike
per Nyquist interval. Each spike is associated with two unknown
parameters: location and amplitude. Since there are $n$ Nyquist
intervals, we may have as many as $2\odim$ unknown parameters in
total, which is more than the number $n$ of observations
(cf. \fref{eq:filter2}, \fref{eq:filter3p}). Hence, recovery of
$\vinp\in\rclass{1}{1}$ is in general not possible even in the
noiseless case. If we however knew the locations of the spikes but not
the amplitudes, we could recover the signal $\vinp\in\rclass{1}{1}$ by
solving a system of linear equations.

Next take $\vinp\in\rclass{2}{1}$.  Such a signal may only contain one
spike per two Nyquist intervals. Hence, the total number of unknown
parameters is at most equal to the number of observations and recovery
of $\vinp\in\rclass{2}{1}$ is barely possible in the noiseless case.
For example, as discussed in~\fref{sec:lit}, $\vinp$ can be recovered
by Prony's method. In general, $\vinp\in\rclass{2r}{r}$ is the
absolute limit for recovery of complex-valued signals in the noiseless
case in the sense that $\vinp\in\rclass{2r-\epsilon}{r}$,
$\epsilon>0$, is in general not recoverable.

Strictly speaking, the general dimension-counting considerations above
do not hold for positive signals $\vinp\in\reals_+^\idim$ because the
positivity of $\vinp$ supplies extra information. On the one hand, it
is nevertheless possible to construct adversarial signals
$\vinp\in\rclassp{2r-\epsilon}{r}$ that will not be recoverable by any
method whatsoever. On the other hand, this paper shows that
$\vinp\in\rclassp{3.74r}{r}$ can be recovered stably in the presence
of (small) noise via the linear program (CVX).

\newcommand\cxi{0.1196}
\newcommand\cxil{0.8261}
\newcommand\cyi{0.3370}
\newcommand\cxii{0.1196}
\newcommand\cyii{0.5000}
\newcommand\cxiii{0.1196}
\newcommand\cyiii{0.7283}
\newcommand\cxiiii{0.0217}
\newcommand\cyiiii{0.1087}
\begin{figure}[ht]
        \centering
\begin{tikzpicture}

\begin{axis}[width=12cm,scale only axis, height=1.4cm,
    name= plot1,
    ytick=\empty,
	xtick={0,1},
	title={$\rclass{4}{1}$},
	every axis title/.style={at={(0.5,1)},above,yshift=-0.2cm},
	axis x line = middle,
	hide y axis,
	ymin=-0.3,
	ymax=1,
	xmin=-0.1,
	xmax=1.1]
	\addplot[domain=0:1,samples=201,line width=0.7pt,blue]{0.2*sin(deg(2*pi*(23/2)*x))};
	\addplot[ycomb,range=0:1,line width=0.1pt,violet,mark options={fill=violet,scale=1},mark=*, mark size=1pt,only marks] 
		table[x index=0,y index=1]{./figures/figrclass.txt};
	\addplot[ycomb,range=0:1,line width=1.2pt,violet] 
	 	table[x index=0,y index=1]{./figures/figrclass.txt};
	\draw[help lines] (axis cs:0,0) -- (axis cs:1,0);
	\draw[<->,>=latex] (axis cs: \cxi,0.3) -- (axis cs: \cyi,0.3) node[above,midway] {\small $\ge 2\lambdac$};
	\draw[->,>=latex] (axis cs: 0,0.3) -- (axis cs: \cxi,0.3) node[above,midway] {};
	\draw[<-,>=latex] (axis cs: \cxil,0.3) -- (axis cs: 1,0.3) node[above,midway] {\small $\ge 2\lambdac$};
\end{axis}

\begin{axis}[
    at={(plot1.below south west)},yshift=-1.6cm,
    anchor=north west,
    width=12cm,scale only axis,height=1.4cm,
    ytick=\empty,
	xtick={0,1},
	axis x line = middle,
	name= plot2,
	title={$\rclass{8}{2}$},
	every axis title/.style={at={(0.5,1)},above,yshift=-0.2cm},
	hide y axis,
	ymin=-0.3,
	ymax=1.1,
	xmin=-0.1,
	xmax=1.1]
	\addplot[domain=0:1,samples=201,line width=0.7pt,blue]{0.2*sin(deg(2*pi*(23/2)*x))};
	\draw[<->,>=latex]  (axis cs:\cxiiii,0.3) -- (axis cs:  \cyiiii,0.3) node[above,midway] {\small $\lambdac$};
	\addplot[ycomb,range=0:1,line width=0.1pt,violet,mark options={fill=violet,scale=1},mark=*, mark size=1pt,only marks] 
		table[x index=0,y index=2]{./figures/figrclass.txt};
	\addplot[ycomb,range=0:1,line width=1.2pt,violet] 
 		table[x index=0,y index=2]{./figures/figrclass.txt};
  \draw[help lines] (axis cs:0,0) -- (axis cs:1,0);
  \draw[<->,>=latex] (axis cs:\cxii,0.3) -- (axis cs: \cyii,0.3) node[above,midway] {\small $\ge 4\lambdac$};
\end{axis}

\begin{axis}[
    at={(plot2.below south west)},yshift=-1.6cm,
    anchor=north west,
    width=12cm,scale only axis,height=1.4cm,
    axis on top=false,
	ytick=\empty,
	xtick={0,1},
	axis x line = middle,
	name= plot3,
	title={$\rclass{12}{3}$},
    every axis title/.style={at={(0.5,1)},above,yshift=-0.2cm},
	hide y axis,
	ymin=-0.3,
	ymax=1.1,
	xmin=-0.1,
	xmax=1.1	]
	\addplot[domain=0:1,samples=201,line width=0.7pt,blue]{0.2*sin(deg(2*pi*(23/2)*x))};
	\draw[<->,>=latex]  (axis cs:\cxiiii,0.3) -- (axis cs:  \cyiiii,0.3) node[above,midway] {\small $\lambdac$};
	\addplot[ycomb,range=0:1,line width=0.1pt,violet,mark options={fill=violet,scale=1},mark=*, mark size=1pt,only marks] 
		table[x index=0,y index=3]{./figures/figrclass.txt};
	\addplot[ycomb,range=0:1,line width=1.2pt,violet] 
	 	table[x index=0,y index=3]{./figures/figrclass.txt};
		
  \draw[<->,>=latex]  (axis cs:\cxiii,0.3) -- (axis cs: \cyiii,0.3) node[above,midway] {\small$\ge 6\lambdac$};
\end{axis}
\end{tikzpicture}
\caption{ Examples of discrete $\idim$ dimensional signals from the
  Rayleigh classes $\rclass{4}{1}$, $\rclass{8}{2}$, $\rclass{12}{3}$
  depicted on the grid $\{0,1/\idim,\ldots,1-1/\idim\}\subset \ucirc$.
  The sine wave $\sin(2\pi \fc t)$ at the highest visible frequency is
  shown in blue for reference. Here, $\idim=92$ and $\odim=23$, so
  that $\SRF=4$ and $\lambdac=1/11$. By periodicity, the endpoints are
  identified.}
\label{fig:rclass}
\end{figure}
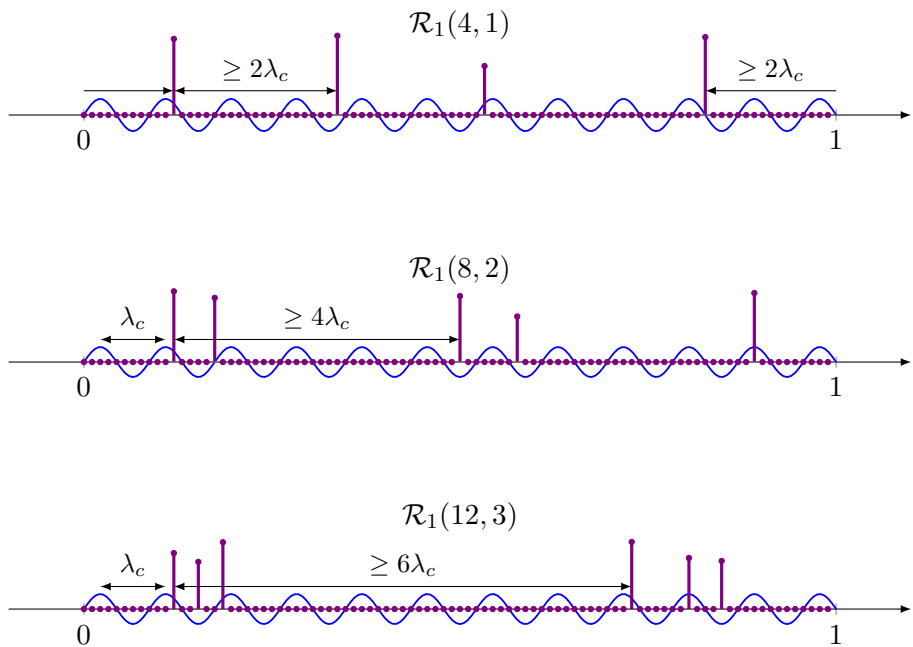

\subsection{Stable recovery}

We are now ready to present our main results; although they extend to
higher dimensions, they are stated in 1 and 2D for
simplicity. Throughout, we assume that the data $\vinpfilt$ is given
by \fref{eq:IO3altnoise1d}.

\begin{thm}[\bf Flat spectrum]
\label{thm:UB}
In 1D, take $\matQ=\matQ_{\mathrm{flat,1D}}$ and $\vinp\in
\rclassp{3.74 \rsp}{\rsp}$ with $\fc\ge 128\rsp$. Then the solution
$\hat\vinp$ to \fref{eq:find0} obeys
\begin{equation}
	\label{eq:mainbnd}
	\vecnorm{\hat\vinp-\vinp}_1\le C \cdot \left(\frac{\idim}{\odim-1}\right)^{2\rsp} \cdot \|\mathbf{z}\|_1 \approx C \cdot \SRF^{2 \rsp} \cdot  \|\mathbf{z}\|_1,
      \end{equation}
      where $C=C_1(\rsp)$, only depends on $\rsp$ (if $\SRF\ge
      3.03/\rsp$, it can be taken as in \fref{eq:c1def}).

      In 2D, take $\matQ=\matQ_{\mathrm{flat,2D}}$, $\vinp\in
      \rclasspdd{4.76 \rsp}{\rsp}$ with $\fc\ge 512\rsp$. Then
      \fref{eq:mainbnd} holds with a constant $C$ depending on $\rsp$
      only, which we do not specify for brevity.
\end{thm}


The result in~\fref{thm:UB} is not sensitive to the exact choice of
the kernel $\matQ$ and remains valid for just about any other
low-frequency kernel. To illustrate this point and to connect our
theory to super-resolution microscopy we now give the result for the
PSF discussed in Section \ref{sec:microscopy}.

\begin{thm}[\bf Triangular spectrum]
\label{thm:UBtriag}
Set $1/2\le\alpha<1$. In 1D, take $\matQ=\matQ_{\mathrm{tri,1D}}$ and
assume $\vinp\in \rclassp{3.74 \rsp/\alpha}{\rsp}$
with $\fc\ge 256\rsp$. Then the bound~\fref{eq:mainbnd} holds with a
finite constant $C=C_{1}(\rsp,\cutfn)$, namely,
\begin{equation}
\label{eq:boundtwo}
\vecnorm{\hat\vinp-\vinp}_1\le C \cdot \SRF^{2 \rsp} \cdot
\|\mathbf{z}\|_1. 
\end{equation}
(If $\SRF\ge 3.03/\rsp$, then the constant can be taken as
in~\fref{eq:c1ra}.)

In 2D, take $\matQ=\matQ_{\mathrm{tri,2D}}$, $\vinp\in
\rclasspdd{4.76\rsp/\alpha}{\rsp}$ with $\fc\ge 1024\rsp$.  Then
except for the numerical value of the constant, the same conclusion
holds.
\end{thm} 

When $\alpha\to 1$, $C_{1}(\rsp,\cutfn)\to \infty$, which reflects the
fact that, as seen from \fref{eq:filter2}, the spectrum of
$\matQ_{\mathrm{tri,1D}}$ is very small at the border of the interval
$[-\fc,\fc]$. Hence, with noise, the spectral components of the signal
can only be observed away from this border, for example on the
interval $[-0.9\fc,0.9\fc]$, which corresponds to taking $\alpha=0.9$
in~\fref{thm:UBtriag}.

\paragraph{Implications for single-molecule microscopy.}
Consider~\fref{thm:UBtriag} in 2D and remember that $\onenorm{\vecz}$
is the cumulative difference in light intensity between noiseless
(ideal) and real observations.  Then the theorem tells us that the
cumulative error in light intensity in signal estimates is bounded by
the amplified version of the cumulative error in light intensity in
the data. The \ac{NAF} behaves as $\SRFn^{2\rsp}$, where $\rsp$ is the
parameter describing the regularity of the signal support.  If the
noise level is sufficiently small and the signal is sufficiently
regular ($\rsp$ is small), i.e., not too many molecules are clustered
close together, and $\SRFn$ is modest, then the
algorithm~\fref{eq:find0} is guaranteed to achieve excellent
super-resolution results. As we will explain in~\fref{sec:tight}, no
algorithm can perform substantially better.

\paragraph{Contribution.}  Theorems \ref{thm:UB} and \ref{thm:UBtriag}
are new, and while their proofs are given in~\fref{sec:direct1}, we
would like to discuss the main technical contribution of this paper.
When $\rsp=1$ or, equivalently, when the spikes are separated by at
least $1.87\lambdac$ and not necessarily positive, a result similar
to~\fref{thm:UB} was obtained in~\cite[Th.~1.5]{candes13} using a
different convex program, see also \cite{candes21-2} for a
continuous-space version; this program, given by (L1) below, requires
knowledge of an upper bound on $\|\vecz\|_1$.  The proof
in~\cite{candes13} is based on constructing a (dual) low-frequency
trigonometric polynomial that interpolates the sign of the spikes.
The crucial observation we make in this paper is that the technique
developed in~\cite{candes13} can be extended to the important setting
when the spikes are not separated and positive. The proof is based on
a simple idea: a Rayleigh-regular set may be partitioned into subsets
with points in each subset separated by at least $1.87\lambdac$;
therefore, each set comes with a (dual) low-frequency trigonometric
polynomial constructed in~\cite{candes13}; multiplying such
polynomials together gives a low-frequency polynomial interpolating
the signal.

In the noiseless setting $(\vecz = \mathbf{0})$, our results state
that the recovery is exact. In 1D this is well known,
see~\cite{donoho90-06,fuchs05} and the review in~\fref{sec:lit}. In
\ac{2D} and higher, this is new: as explained in~\fref{sec:direct1},
this result cannot be obtained by a straightforward generalization of
the techniques in~\cite{donoho90-06,fuchs05}.



\subsection{Tightness}
\label{sec:tight}

In this section we argue that our results in \fref{thm:UBtriag} are
nearly tight. In 1D, we are interested in answers to the following two
natural questions:
\begin{enumerate}[(i)]
\item Can the assumption $\setC=\rclassp{3.74 \rsp}{\rsp}$ be
  substituted with $\setC=\rclassp{d}{\rsp}$ with $d<3.74 \rsp$
  without changing the bound \eqref{eq:boundtwo}? 

\item Can the exponent $2\rsp$ in \eqref{eq:boundtwo} be made
  smaller?
\end{enumerate}

\subsubsection{Tightness of the length of the interval}  

To answer (i), we have already argued in~\fref{sec:reyleigh} that even
in the noiseless case it is not possible to recover many of the
signals $\vinp\in\rclassp{d}{\rsp}$ with $d< 2\rsp$. Hence, $d=
3.74\rsp$ is within a factor $1.87$ of the optimum.  This factor comes
from the key result from \cite{candes13} explained above, which
concerns the existence of low-frequency polynomials interpolating
complex scalars of unit magnitude separated by $1.87 \lambda_c$. Any
improvement in this technology would yield a corresponding improvement
here, see Section \ref{sec:construction} for additional details.

\subsubsection{Tightness of the exponent}
\label{sec:convfrej}
To answer question (ii) above, we need the concept of \ac{MC}.
\begin{dfn}[\bf Modulus of continuity]
\label{def:MC}
Let $\vecnorm{\cdot}$ be a norm, $\matQ$ a linear operator, and
$\setC$ a class of signals.\footnote{For example, $\setC$ may be a
  class of sparse signals, a class of Rayleigh regular signals, and so
  on.} The \ac{MC} is defined as
\begin{equation}
	\MC[\setC,\matQ]\define\sup_{\vinp_1,\vinp_2\in \setC} \frac{\inpnorm{\vinp_1-\vinp_2}}{\outnorm{\matQ(\vinp_1-\vinp_2)}}.
\end{equation}
\end{dfn}

We also introduce the simple notion of noise amplification. 
\begin{dfn}[\bf Noise amplification factor]
\label{def:NAF}
Let $\inpnorm{\cdot}$ be a norm, $\matB$ a linear operator, and
$\setC$ a signal class. Suppose an algorithm A produces an estimator
$\hat\vinp(\vinpfilt)$ from the model
$\vinpfilt=\matB\vinp+\projnoise$ obeying the uniform stability
guarantee
	\begin{align}
		\label{eq:stab}
		\inpnorm{\hat\vinp-\vinp} &\le \NAF[A,\setC,\matB]
                \cdot \delta
	\end{align}
	for all $\vinp\in\setC$ and all $\vecz$ with $\outnorm{\vecz}
        \le \delta$.
        Then we say that the \ac{NAF} of A is (at
        most)~$\NAF[A,\setC,\matB]$.
\end{dfn}

\noindent The \ac{MC} is related to the \ac{NAF} via the following simple facts.
\begin{enumerate}
	\item 
	If the \ac{NAF} of an algorithm A is at most  $\NAF[\mathrm{A},\setC,\matQ]$, then
	\begin{equation}
		\label{eq:NAFMC}
		\NAF[\mathrm{A},\setC,\matQ]\ge \MC[\setC,\matQ].
	\end{equation}
      \item Consider the exhaustive search (ES) algorithm (in general
        intractable) for super-resolving signals in $\setC$:
	\begin{align}
		\label{eq:comb}\tag{ES}
		\text{find}\quad \hat\vinp\in\setC \quad  \text{s.t.}  \quad \outnorm{\vinpfilt-\matQ\hat\vinp}\le \delta
	\end{align}
	with $\delta$ chosen so that $\vecnorm{\vecz}\le \delta$. The \ac{NAF} of this algorithm satisfies
	\begin{equation}
		\NAF[\mathrm{ES},\setC,\matQ]\le 2\MC[\setC,\matQ].
	\end{equation}
\end{enumerate}

We now provide a lower bound on the \ac{MC} showing that if the noise
is arbitrary, no algorithm can have a \ac{NAF} smaller than $\cnstnum
\SRF^{2\rsp-1}$. Therefore, the exponent $2\rsp$ in \fref{eq:mainbnd}
is nearly optimal.

\begin{thm}
\label{thm:conv}
Take $\matQ=\matQ_{\mathrm{tri,1D}}$.  Set $\rsp$ and $d$ to be
arbitrary numbers and $\setC=\rclassp{d}{\rsp}$ so that by taking $d =
\infty$ we would have at most $r$ spikes. Then there exist signals
$\vecx=\tp{[x_0\cdots x_{\idim-1}]}, \tilde\vecx=\tp{[\tilde x_0\cdots
  \tilde x_{\idim-1}]}\in \setC$ s.t.
	\begin{equation}
		\onenorm{\vecx-\tilde\vecx}=1
	\end{equation} and when $\idim,\odimn\to\infty$, $\idim/\odimn\to\SRFn$
	\begin{equation}
		\onenorm{\matQ(\vecx-\tilde\vecx)}\to \gfun{\rsp}{\SRFn} \SRFn^{2\rsp-1}.
	\end{equation}
	For $\SRFn\to\infty$,
	\begin{equation}
		\gfun{\rsp}{\SRFn}\to \cnstl{\rsp}
	\end{equation}
	where $\cnstl{\rsp}$ depends on $\rsp$ only and is given explicitly in~\fref{eq:crsp}. 
	Consequently, letting $\vecnorm{\cdot}$ be $\onenorm{\cdot}$ in~\fref{def:MC}, 
	\begin{equation}
		\label{eq:MCbnd}
		\MC[\setC,\Q]\ge \cnstl{\rsp} \SRFn^{2\rsp-1}
	\end{equation}
	when $\idim,\odim$ and $\SRFn=\idim/\odimn$ are large.
\end{thm}

The proof, given in~\fref{app:th3}, relies on an explicit construction
of nonnegative signals $\vinp$ and $\tilde\vinp$ with disjoint
supports and such that the spikes in $\vinp-\tilde\vinp$ cancel out as
much as possible after low-pass filtering.

Comparing Theorems \ref{thm:conv} and \ref{thm:UBtriag}, we see that
the exponent of $\SRF$ in the \ac{RHS} of \fref{eq:mainbnd} is within
one unit of the best possible. It is important to point out that the
convex optimization algorithm in~\fref{eq:find0} knows nothing at all
about the regularity of the signal class~$\setC$. Yet, it is adaptive
in the sense that it has nearly optimal stability guarantee whatever
the (usually unknown) value of $r$.

\fref{thm:conv} tells us that the \ac{MC} increases exponentially with
$\rsp$. For example, for a practically interesting case where
$\SRFn=8$, it is not difficult to estimate from~\fref{eq:MCbnd} and
the numerical value of the constant that super-resolution could only
be possible if $\rsp\le 5$. For $\rsp> 5$, the modulus of continuity
is greater than $10^5$, setting unrealistic constraints on noise
levels in practical applications. This is even an optimistic estimate
and, in reality, it is nearly impossible to separate more than three
sources packed in a Nyquist interval.

It is not known whether the exponent in the lower bound
\fref{eq:MCbnd} is sharp.  In the very special case where the signal
contains exactly one spike, it is not difficult to see that a simple
matched-filter will have a bounded ratio $\NAF/\SRF$, matching the
exponent in the \ac{RHS} of \fref{eq:MCbnd}. This can be used in the
setting where the spikes are guaranteed to be so far apart, that the
overlap between their images in the output space can be neglected;
this only happens when the distance between neighboring spikes far
exceeds $\lambdac/2$ and all the spikes have roughly the same
magnitude. In general, in the interesting case where the images of
neighboring spikes can overlap in the output space, it is not clear
how one could close the small gap between \fref{eq:MCbnd} and
\fref{eq:mainbnd}. In fact, it is possible that the exponent in
\fref{eq:MCbnd} can be made larger. As we shall see, to construct
adversarial signals $\vecx,\tilde\vecx$ in the proof
of~\fref{thm:conv}, we only use signals that contain exactly $\rsp$
spikes each. However, the signals in $\rclassp{d}{\rsp}$ can have more
than $\rsp$ spikes, of course, which could allow one to construct
pairs $\vecx,\tilde\vecx$ that give a larger bound than that in the
\ac{RHS} of~\fref{eq:MCbnd}. Please also see the recent
preprint~\cite{demanet14the-r}, where the question of calculating the
exact exponent for signals with a total of $r$ spikes is addressed.

\section{Literature review and innovations}
\label{sec:lit}

\subsection{Prior art}
\paragraph*{Algebraic methods.} 
Prony's method~\cite{prony95essai} is an algebraic approach for
solving the 1D super-resolution problem from noiseless data when the
number of sources is known a priori.  The data $\vinpfilt$ is used to
form a trigonometric polynomial, whose roots coincide with the spike
locations. The polynomial is then factored, thus revealing those
locations, and the amplitudes estimated by solving a system of linear
equations. In the noiseless case, Prony's method recovers $\vinp$
perfectly provided that $\zeronorm{\vinp}< \odim/2$. No further
Rayleigh regularity assumption on the signal support is needed.  With
noise, however, the performance of Prony's method degrades
sharply. The difficulty comes from the fact that the roots of a
trigonometric polynomial constructed by an algebraic method are
completely unstable and can shift dramatically even with small changes
in the data.

Many noise-aware versions of Prony's method are used frequently in
engineering applications, for example in radar (see~\cite[Chapter
6]{stoica05}). The most popular methods are MUSIC and its numerous
variations~\cite{barabell83impro, bienvenu79influ,
  schmidt86-03,pisarenko73theretr,tufts82estim,cadzow88signa},
matrix-pencil~\cite{hua90matri}, and
ESPRIT~\cite{paulraj86a-sub,roy89-07}. For more details on algebraic
methods we refer the reader to the excellent book~\cite[Chapter
4]{stoica05}. However, the stability of noise-aware algebraic methods
is not theoretically well-understood. Asymptotic results (at high SNR)
on the stability of MUSIC in the presence of Gaussian noise are
derived in~\cite{clergeot89perfo,stoica91stati}. More recently, some
steps towards analyzing MUSIC and matrix-pencil in a non-asymptotic regime have
been taken in~\cite{liao14music} and in~\cite{moitra14the-t}, respectively. Nevertheless, to the best of our
knowledge, no strong theoretical stability guarantees like those in
Theorems \ref{thm:UB} and \ref{thm:UBtriag} are available for
algebraic methods. Hence, the search for super-resolution methods that
perform well empirically and have sharp theoretical stability
guarantees is an important open problem.

Algebraic methods have been generalized to the multi-dimensional
case. Surprisingly, the generalizations are not straightforward and
many methods (\cite{clark94-06}, \cite{clark97-11},
\cite[Sec. 4.9.7]{stoica05}) have very restrictive sparsity
constraints: namely, at most $n$ spikes when we recall, that in 2D the
total number of observations is $\odim^2$. In~\cite{jiang01-09}, the
number of spikes can be as large as $\odim^2/4$ in the noiseless case,
as one would expect from dimension-counting considerations. 

\paragraph*{Fundamental limits.} 
In the pioneering work~\cite{donoho92-09}, Donoho studied limits of
performance for the 1D super-resolution problem.  His main findings 
can be summarized as follows.  Put $\inpnorm{\cdot}=\vecnorm{\cdot}_2$
in the definition of \ac{NAF} and $\matQ=\matQ_{\mathrm{flat,1D}}$.
\begin{itemize}
\item Let $\setC=\rclass{4 \rsp}{\rsp}$, then the \ac{NAF} of the
  exhaustive search algorithm~\fref{eq:comb} obeys
	\begin{equation}
		\label{eq:donohonafUB}
		\NAF[\mathrm{ES},\setC,\matQ]\le \cnst{\rsp} \SRF^{2\rsp+1},
	\end{equation}
	where $\cnst{\rsp}$ is a positive constant that might depend on $\rsp$ but not on~$\idim$ or~$\odim$.
      \item Take an arbitrary pair $(r,d)$ and set
        $\setC=\rclass{d}{\rsp}$. Then 
	\begin{equation}
		\label{eq:donohonafUB2}
		\MC[\setC,\matQ]\ge \cnst{\rsp} \SRF^{2\rsp-1},
	\end{equation}
	where $\cnst{\rsp}$ is a positive constant that might depend on $\rsp$ but not on $\idim$ or $\odim$.
\end{itemize}
To the best of our knowledge, the analysis in~\cite{donoho92-09} has
not been generalized to the multi-dimensional case.  Unfortunately,
The algorithm~\fref{eq:comb} is not feasible because $\setC$ is not
convex, and~\cite{donoho92-09} does not propose any tractable
algorithm that would have \ac{NAF} bounded above by the \ac{RHS} of
\fref{eq:donohonafUB}.  In this respect, the key question posed by
Donoho is whether a feasible algorithm that achieves stability
in~\fref{eq:donohonafUB} exists.

Other works \cite{stoica89stati,stoica89maxim,batenkov13accuracy}
study the stability of the super-resolution problem in the presence of
noise, but likewise do not provide a tractable algorithm to perform
recovery. Work in~\cite{shahram04-05,shahram05resolv,helstrom64-10}
analyzes the detection and separation of two closely-spaced spikes,
but does not generalize to the case when there are more than two
spikes in the signal.

\paragraph*{Super-resolution under minimum separation constraint.}
Progress towards resolving the question posed in \cite{donoho92-09} in
the general situation where $\vinp\in\complexset^\idim$---in this
paper we consider the case $\reals_+^\idim$ only---has recently been
made~\cite{candes13,candes21-2}. Put
$\inpnorm{\cdot}=\vecnorm{\cdot}_1$ in the definition of the \ac{NAF},
select the PSF with a flat spectrum, $\matQ=\matQ_{\mathrm{flat,1D}}$,
and consider $\setC=\rclass{4}{1}$.  It was shown that the \ac{NAF} of
the $\ell_1$-minimization algorithm
\begin{equation}
	\tag{L1}
	\label{eq:l1min}
	\min_{\hat\vinp} \onenorm{\hat\vinp} \quad  \text{s.t.}  \quad \vecnorm{\vinpfilt-\matQ\hat\vinp}_1\le \delta
\end{equation}
with $\delta$ chosen so that $\vecnorm{\vecz}\le \delta$ is at most
\begin{equation}
	\label{eq:carlosbnd}
	\NAF[\mathrm{L1},\setC,\matQ]\le \cnstnum \cdot \SRF^{2},
\end{equation}
where $\cnstnum$ is a positive numerical constant.  The condition
$\vinp\in\setC=\rclass{4}{1}$ is restrictive because it means that the
signal $\vinp$ cannot contain spikes that are at a distance less than
$2\lambdac$. [For real-valued signals $\vinp\in\reals^\idim$, a
minimum separation of $1.87 \lambdac$ suffices.]  This is a limitation
for many applications including single-molecule microscopy, as it is
usually understood that the goal of super-resolution is to distinguish
spikes that are (significantly) closer than the Rayleigh diffraction
limit, i.e.~at a fraction of $\lambdac$ apart. Unfortunately, if there
are spikes at a distance lower than this value, $\ell_1$ minimization
does not, in general, return the correct solution even if there is no
noise.  Results in~\cite{candes13,candes21-2} also cover the
multi-dimensional case under a minimum separation constraint. On a
similar line of research, see~\cite{tang13-1} and \cite{tang13compr}
for related results on the denoising of line spectra and on the
recovery of sparse signals from a random subset of their low-pass
Fourier coefficients. The accuracy of support detection under the minimum separation constraint is analyzed in~\cite{fernandez-granda13suppo,azais15spike}. 

\paragraph*{Super-resolution of noiseless nonnegative signals.} 
The case of 1D nonnegative signal, $\vinp\in\reals_+^\idim$, was
analyzed in~\cite{donoho90-06}, see also~\cite{fuchs05} for a shorter
exposition of the same idea. Adapting to our setting, the result
in~\cite{donoho90-06} can be summarized as follows: put
$\vecnorm{\cdot}=\onenorm{\cdot}$ in the definition of \ac{NAF} and
$\matQ=\matQ_{\mathrm{flat,1D}}$.  Let $\setC$ be the class of all
signals with $\zeronorm{\vinp}<\odim/2$. Then the
$\NAF^{\mathrm{cvx}}$ of the convex feasibility program
\begin{equation}
	\tag{F}
	\label{eq:f}
	\mathrm{find}\ \hat\vinp\ge\veczero \quad \text{s.t.}  \quad \vecnorm{\vinpfilt-\matQ\hat\vinp}_2\le \delta
\end{equation}
with $\delta$ chosen so that $\twonorm{\vecz}\le \delta$, is a finite
positive constant.  The exact dependence of $\NAF^{\mathrm{cvx}}$ on
$\idim$ and $\odim$ is not specified in~\cite{donoho90-06}. As we will
see, further examination of the proof from~\cite{donoho90-06} leads to
a bound of the form
\begin{equation}
	\label{eq:donjonbnd}
	\NAF^{\mathrm{cvx}}[\setC,\matQ]\le (\cnstnum \idim)^{2 \zeronorm{\vinp}},
\end{equation} 
where $\cnstnum$ is a numerical constant. First, this does not depend
on the Rayleigh regularity of $\vinp$ but on the sparsity. Second,
this does not depend on the SRF but on the grid size. By comparing
to~\fref{eq:donohonafUB} and~\fref{eq:carlosbnd} we see that the bound
\fref{eq:donjonbnd} is weak. Indeed, consider the interesting case
$\idim,\odim\to\infty$ with $\idim/\odim=\SRF$ kept constant. In this
case the bounds in~\fref{eq:donohonafUB} and~\fref{eq:carlosbnd}
remain finite, whereas the \ac{RHS} of~\fref{eq:donjonbnd} converges
to $+\infty$ very quickly. The bound in~\fref{eq:donjonbnd} does not
depend on the frequency cut-off $f_c$ or equivalently the number $n$
of pieces of information we are given.  Whether the frequency cut-off
is $10$ or $10^6$ the bound remains the same! This cannot capture the
right behavior. 

\subsection{Innovations}
The novelty of our results can be summarized as follows.

\begin{itemize}
\item As compared to algebraic methods, Theorems~\ref{thm:UB}
  and~\ref{thm:UBtriag} show that efficient algorithms can recover the
  signal in a provably stable fashion. As we discussed earlier, strong
  worst-case stability guarantees are not available for algebraic
  methods. The flipside is that our results crucially rely on
  non-negativity of the signal; algebraic methods do not need this
  assumption.

\item As compared to \cite{donoho92-09}, our recovery algorithm is a
  simple \ac{LP} and, hence, is tractable whereas the exhaustive search
  method of~\cite{donoho92-09} is intractable and cannot be used in
  practice. The difference between stability exponents
  in~\fref{eq:mainbnd} and in~\fref{eq:donohonafUB} stem from the fact
  that \cite{donoho92-09} works with the $\ell_2$ norm while we work
  with $\ell_1$. (The stability bounds for the exhaustive search
  algorithm in~\cite{donoho92-09} do not assume the signal to be
  nonnegative.)

\item As compared to \cite{candes13}, our results do not rely on the
  restrictive minimum-separation assumption. Having said this, the
  results in~\cite{candes13} hold for the general case of complex
  amplitudes, and our proofs borrow heavily from the tools developed
  in that work.

\item As compared to work in~\cite{donoho90-06}, our stability
  estimates are far stronger, for they depend on the super-resolution
  factor, and not on the spacing on the fine grid. Further, if one
  tries to use the proof technique used in~\cite{donoho90-06} to
  generalize the noiseless results in~\cite{donoho90-06,fuchs05} to
  the \ac{2D} case, one would need to assume that our image has at
  most $\odim/2$ spikes: this is too restrictive. In sharp contrast,
  we see from Theorems \ref{thm:UB} and~\ref{thm:UBtriag} that if the
  signal support is Rayleigh regular, we may have a number of sources
  on the order of~$\odim^2$, i.e.~on the order of the number of
  measurements.
\end{itemize}

\section{Proofs}
\label{sec:direct1}
\subsection{Proof of Theorem~\ref{thm:UB} in the 1D case}
\label{sec:proofp}

The proof of the theorem is based on the following lemma.
\begin{lem}
\label{lem:fnd}
Assume that the assumptions of~\fref{thm:UB} are satisfied. Set
\begin{equation}
	\label{eq:verror}
	\vinper=\tp{[h_0\cdots h_{\idim-1}]}=\hat\vinp-\vinp
\end{equation}
and 
\begin{equation}
	\label{eq:setT}
		\setT=\{l/\idim: h_l<0\}, 
\end{equation}
and suppose there exists $\dualp=\tp{[q_0\cdots
  q_{\idim-1}]}\in\reals^\idim$ and $0< \rho<1$ such that
$\matQ\dualp=\dualp$, $\infnorm{\dualp}\le 1$, and
\begin{equation}
	\begin{cases}
		\label{eq:qbound}
		q_l= 0, & l/\idim\in \setT\\
		q_l>2\rho, & \text{otherwise}. 
	\end{cases}	
\end{equation}
Then
\begin{equation}
		\label{eq:errorbound}
		\vecnorm{\hat\vinp-\vinp}_1\le \frac{2(1-\rho)}{\rho} \cdot \|\vecz\|_1.
\end{equation}
\end{lem}
\begin{proof}
  Set $\tilde \dualp=\tp{[\tilde q_0\cdots \tilde
    q_{\idim-1}]}=\dualp-\rho$ and note that
  $\infnorm{\tilde\dualp}\le 1-\rho$ since $\rho \le 1/2$. On the one
  hand,
\begin{align}
  \abs{\inner{\tilde\dualp}{\vinper}}
  =\abs{\inner{\matQ\tilde\dualp}{\vinper}}& =\abs{\inner{\tilde\dualp}{\matQ\vinper}}\nonumber\\
  &\le \infnorm{\tilde\dualp}\onenorm{\matQ\vinper}\nonumber\\
  &\le (1-\rho) \onenorm{\matQ\vinp-\vinpfilt+\vinpfilt-\matQ\vinpest}\nonumber\\
  &\le (1-\rho) (\onenorm{\matQ\vinp-\vinpfilt}+\onenorm{\vinpfilt-\matQ\vinpest})\nonumber\\
  &\le 2(1-\rho) \onenorm{\matQ\vinp-\vinpfilt} \\ & = 2(1-\rho) \cdot
  \|\vecz\|_1.
	\label{eq:pt1}
\end{align}
On the other hand, using $\sign(\tilde q_l)=\sign(h_l)$ for all $l$
gives
\begin{align}
	\label{eq:pt2}
	\abs{\inner{\tilde\dualp}{\vinper}}=\abs{\sum_{l=0}^{\idim-1} \tilde q_l \vinperc_l}=\sum_{l=0}^{\idim-1} \tilde q_l \vinperc_l=\sum_{l=0}^{\idim-1} \abs{\tilde q_l} 	\abs{\vinperc_l} \ge \rho \onenorm{\vinper},
\end{align}	
Combining \fref{eq:pt1} and \fref{eq:pt2} yields the conclusion. 
\end{proof}

\subsubsection{Localization of trigonometric polynomials}
\label{sec:intuit}
\fref{lem:fnd} shows that in order to obtain a tight bound we need to
construct a (dual) vector $\dualp$ obeying $\infnorm{\dualp}\le 1$ and
\fref{eq:qbound} with $\rho$ as large as possible. First, observe that
since $\vinp,\hat\vinp \ge \veczero$ it follows that $\setT$
in~\fref{eq:setT} satisfies $\abs{\setT}\le\zeronorm{\vinp}< \odim/2$.
The idea is to construct a real-valued trigonometric polynomial of
largest frequency $\fc$ (recall $\odim=2\fc+1$)
\begin{equation}
	\label{eq:trigpolylow}
        q(t)=\sum_{k=-\fc}^{\fc} \hat q_k e^{-\iu 2\pi k t}\in\reals \quad \text{for all} \quad t, 
\end{equation}
obeying $\infnorm{q}\le 1$, 
\begin{equation}
		\begin{cases}
		    q(t)=0,\quad \text{for all}\quad t\in\setT,\\
			q(t)>0,\quad \text{for all} \quad t\notin\setT, 
		\end{cases}\label{eq:polycond1}
\end{equation}
and set $\dualp=\{q(l/\idim) : l \in [0:N-1]\}$.  Observe that such a
$\dualp$ would obey the conditions of \fref{lem:fnd} with
$\rho=\frac{1}{2}\argmin_{l/\idim\notin\setT}\{q(l/\idim)\}$.

A classical approach to constructing such a polynomial $q(t)$ is
\begin{equation}
	\label{eq:badq}
	q(t)=\prod_{t_0\in\setT} \frac{1}{2}\left[\cos\lefto(2\pi (t-t_0) + \pi\right)+1 \right].
\end{equation}
This approach, used in~\cite{donoho90-06,fuchs05}, works whenever
$\abs{\setT}<\odim/2$ since the degree of $q(t)$ is then at most
$\fc$. 
The problem is that $q(t)$ in \fref{eq:badq} grows extremely slowly
around its zeros, making $\rho$ very small, which then translates into
highly suboptimal stability estimates. To demonstrate this, assume
that $\setT=\{0\}$, i.e., $\abs{\setT}=1$.  Then
(see~\fref{fig:donohopoly})
\begin{equation}
	\label{eq:badq1}
	q(t)=\frac{1}{2}\left[\cos\lefto(2\pi t + \pi\right)+1\right]
\end{equation}
so that
\begin{equation}
	q(1/\idim) \le \frac{\pi^2}{\idim^2}
\end{equation}
and $\rho\le \frac{\pi^2}{2\idim^2}$. Plugging this into \fref{eq:errorbound} we get an estimate no better than
\begin{equation}
	\label{eq:errorbounddon}
	\vecnorm{\hat\vinp-\vinp}_1\le \frac{2}{\pi^2}\idim^2 \cdot \|\vecz\|_1.
\end{equation}
This is weak. In the case when $\vinp$ has one spike, the separation
condition $\setT\in\rclass{3.74}{1}$ of~\cite{candes13} is
trivially satisfied. The results in~\cite{candes13} guarantee that
$\ell_1$ minimization achieves 
\begin{equation}
	\label{eq:errorbound2}
	\vecnorm{\hat\vinp-\vinp}_1\le \cnstnum \cdot \frac{\idim^2}{\odim^2} \cdot \|\vecz\|_1 =\cnstnum \cdot \SRF^2 \cdot \|\vecz\|_1,
\end{equation}
where $\cnstnum$ is a numerical constant.  The reason
why~\cite{candes13} provides stability guarantees far stronger
than~\fref{eq:errorbounddon} is that the trigonometric polynomial
$q(t)$ constructed in~\cite{candes13} grows around its zeros much
faster than $q(t)$ in~\fref{eq:badq1}. We review the behavior of
$q(t)$ constructed in~\cite{candes13} in~\fref{lem:dual} below and
illustrate the difference between this polynomial and that
in~\fref{eq:badq1}. Based on the results of~\cite{candes13}, we then
present a novel construction for~$q(t)$ that does not
rely on the minimal separation condition $\setT\in\rclass{3.74}{1}$
needed in~\cite{candes13} and works for all signals with
Rayleigh regular support of the type
$\setT\in\rclass{3.74\rsp}{\rsp}$. At the same time, the new
polynomial~$q(t)$ grows rapidly around its zeros, which allows us to
derive strong stability guarantees.

\renewcommand\cxi{0.0057}
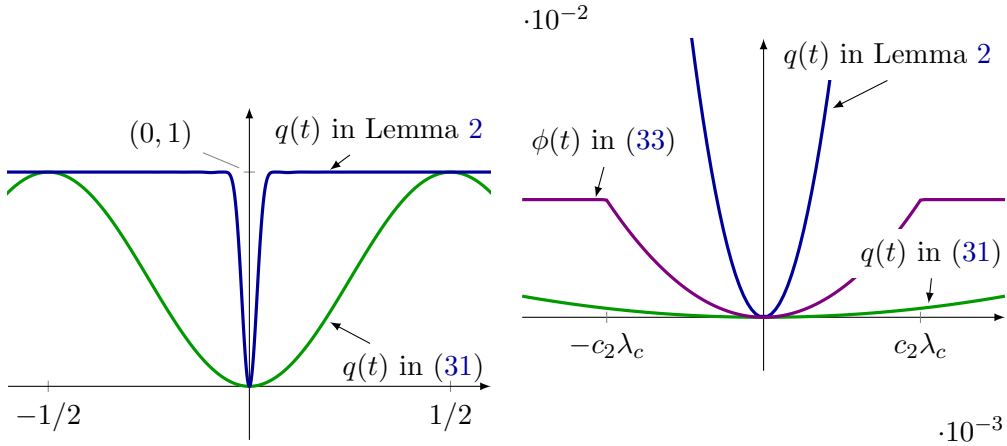
\begin{figure}[ht]
        \centering
\begin{subfigure}[b]{0.4\textwidth}
\begin{tikzpicture} 
	\begin{axis}[width=8cm,height=6cm,
    	ytick={1},
		xtick={-0.5,0.5},
    	axis y line=center,
    	axis x line=middle,
		xmin=-0.6,
		xmax=0.6,
		ymax=1.3,
		ymin=-0.25,
     	xticklabels={$-1/2$,$1/2$},
		yticklabels={}]
	
		\node[pin=165:{$(0,1)$}] at (axis cs:0,1) {};
    	\addplot[line width=1.2pt,green!60!black] table[x index=0,y index=1]{./figures/fig1poly.txt};
		\addplot[line width=1.2pt,blue!60!black] table[x index=0,y index=2]{./figures/fig1poly.txt};
		\draw[<-,>=latex] (axis cs: 0.2,0.3) -- (axis cs: 0.3,0.2) node[anchor=north] {\qquad\quad  $q(t)$ in \fref{eq:badq1}};
		\draw[<-,>=latex] (axis cs: 0.2,1.01) -- (axis cs: 0.25,1.08) node[anchor=south] {\qquad $q(t)$ in \fref{lem:dual}};
	\end{axis}
\end{tikzpicture}
\end{subfigure}%
~
\begin{subfigure}[b]{0.4\textwidth}
\begin{tikzpicture} 
	\begin{axis}[
		width=8cm,
		height=6cm,
    	axis y line=center,
    	axis x line=middle,
    	ytick=\empty,
		xtick=\empty,
		extra x ticks={-\cxi,\cxi},
		ymin=-0.00192,
		ymax=0.01,
		extra x tick labels={$-c_2\lambdac$,$c_2\lambdac$}
		]
		
		\addplot[line width=1.2pt,green!60!black] table[x index=0,y index=1]{./figures/fig2poly.txt};
		\addplot[line width=1.2pt,blue!60!black] table[x index=0,y index=2]{./figures/fig2poly.txt};
		\addplot[line width=1.2pt,violet] table[x index=0,y index=3]{./figures/fig2poly.txt};
		
		\draw[<-,>=latex] ( axis cs: -6.1e-3,4.3e-3)  -- ( axis cs: -5.8e-3,5.4e-3) node[anchor=south,fill=white] {$\phi(t)$ in \fref{eq:LBq}};
		\draw[<-,>=latex] ( axis cs: 2.5e-3,0.75e-2) -- ( axis cs: 4.5e-3,0.85e-2) node[anchor=south,fill=white] {$q(t)$ in \fref{lem:dual}};
		\draw[<-,>=latex] ( axis cs: 6e-3,0.5e-3) -- ( axis cs: 6.1e-3,1.4e-3) node[anchor=south,fill=white] {$q(t)$ in \fref{eq:badq1}};
		
	\end{axis}
\end{tikzpicture}
\end{subfigure}
\caption{Comparison of trigonometric polynomials at two different
  scales: the polynomial from \cite{candes13} (see~\fref{lem:dual})
  bounces off zero much faster than that used
  in~\cite{donoho90-06,fuchs05}.}
\label{fig:donohopoly}
\end{figure}

\subsubsection{Main building block: $q(t)$ under separation}
The following lemma is an immediate consequence
of~\cite[Lm. 2.5]{candes13} adapted to the case of real-valued signals
as explained in~\cite[Sec. 2.5]{candes13}.
\begin{lem}
\label{lem:dual}
Assume $\setT\in\setR_1({3.74},{1};{\idim},{\odim})$. As before,
$\fc=(\odim-1)/2, \lambdac=1/\fc$ and suppose $\fc\ge 128$. Then there
exists a real-valued polynomial $q(t) =\sum_{k=-\fc}^{\fc} \hat q_k
e^{-\iu 2\pi k t}$ with $\infnorm{q}\le 1$ such that
\begin{align}
	\begin{cases}
	q(t)=0, & \text{for all } t\in\setT,\\
	q(t)\ge \phi(t), & \text{for all } t,\\
	\end{cases}
\end{align}
where (see~\fref{fig:donohopoly})
\begin{align}
	\label{eq:LBq}
	\phi(t)=\begin{cases}
          c_1 \fc^2 (t_0-t)^2, & \text{for all}\  t\  \text{s.t.~} \exists t_0\in\setT\ \text{with}\ \abs{t-t_0}\le c_2\lambdac\\
          c_1 \fc^2 (c_2 \lambdac)^2=c_1 c_2^2, & \ \text{otherwise}, 
	\end{cases}
\end{align}
and $c_1=0.029$, $c_2=0.17$.		\end{lem}
The significance of this lemma is that the growth of $q(t)$ around its
zeros is nearly optimal.  Indeed, suppose we wish to construct a real
nonnegative polynomial with highest frequency $\fc$ of the
form~\fref{eq:trigpolylow} of magnitude at most one, and which grows
around its zeroes as fast as possible. How fast could it possibly
grow?  Since $q(t)$ is a superposition of harmonic functions, it
cannot outpace a pure harmonic---normalized to take on values in
$[0,1]$---at the highest available frequency. Hence, we cannot hope
for growth faster than
\begin{equation}
	\label{eq:optgrows}
	\frac{1}{2}\left[\cos\lefto(2\pi \fc (t - t_0) + \pi\right)+1\right]\approx \pi^2 \fc^2 (t-t_0)^2\quad \text{for small } t.
\end{equation}
Comparing~\fref{eq:optgrows} to~\fref{eq:LBq}, we see that
\fref{lem:dual} provides a construction that is optimal up to at most
a constant factor.

We now show how to extend the construction in~\fref{lem:dual} to the
case where the elements of $\setT$ are not necessarily well-separated,
but $\setT$ is Rayleigh regular. Together with~\fref{lem:fnd}, this
will prove~\fref{thm:UB}. The proof below is illustrated
on~\fref{fig:proof}, which the reader is encouraged to consult while
following the argument.  \renewcommand\cxi{0.1196}
\renewcommand\cxii{0.1739} \renewcommand\cxiii{0.5000}
\newcommand\cxiv{0.5761} \newcommand\cxv{0.8913}
\begin{figure}[ht]
        \centering
		\begin{tikzpicture} 
			\begin{axis}[width=12cm,scale only axis, height=3cm,
		   	xmin=-0.1,
				xmax=1.1,
				xtick={0,\cxi,\cxii,\cxiii,\cxiv,\cxv,1},
		    	ytick=\empty,
		   	ymax=4.3e-3,
				ymin=0,
		     	xticklabels={$0$,$t_1$,$t_2$,$t_3$,$t_4$,$t_5$,$1$},
				hide y axis,
				axis x line = middle]
	
		    	\addplot[line width=1.2pt,green!60!black] table[x index=0,y index=1]{./figures/figqlb.txt};
				\addplot[line width=1.2pt,blue!60!black,dashed] table[x index=0,y index=2]{./figures/figqlb.txt};
				\draw[<-,>=latex] (axis cs: 0.095,1e-3) -- (axis cs: 0.05,1e-3) node[anchor=east] {$\phi_1(t)$};
				\draw[<-,>=latex] (axis cs: 0.2,1e-3) -- (axis cs: 0.25,1e-3) node[anchor=west] {$\phi_2(t)$};
			\end{axis}
		\end{tikzpicture}
\caption{Illustration of the proof of~\fref{thm:UB} for $\rsp=2$; $\setT_1=\{t_1,t_2,t_3\}$; $\setT_2=\{t_2,t_4\}$. The trigonometric polynomials $q_1(t)$, $q_2(t)$ satisfy $q_1(t)=0$ for all $t\in \setT_1$ and $q_2(t)=0$ for all $t\in \setT_2$; they are not displayed. The lower bounds $\phi_1(t)$ and $\phi_2(t)$ defined in~\fref{eq:LBdef} are depicted.}
\label{fig:proof}
\end{figure}
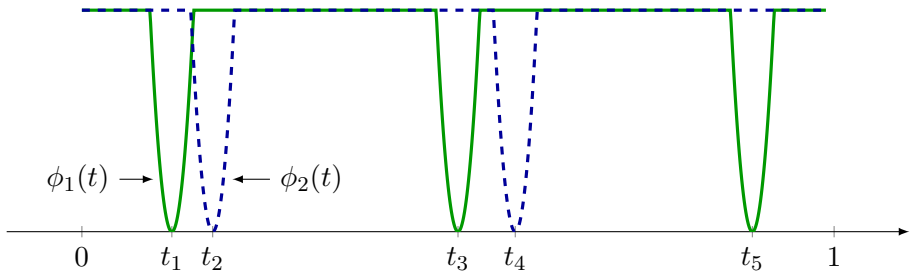

\subsubsection{Construction of $q(t)$ without separation}
\label{sec:construction}
Take $\vinp\in\rclassp{3.74\rsp}{\rsp}$ with a support of cardinality
$S$. Define $\vinper$ and $\setT=\{t_k\}_{k=1}^\sparsity$ with $t_1 <
t_2 < \ldots < t_S$ as in Lemma \ref{lem:fnd}. Since $h_l$ can only
take on negative values on $\support(\vinp)$, then
$\setT\in\setR_1(3.74\rsp,\rsp)$.  Consider the partition
$\setT=\union_{i=1}^\rsp\setT_i$, where $\setT_i=\{t_{\rsp
  k+i}\}_{k=0}^{\sparsity/\rsp-1}$.  Since $\setT \in
\rclass{3.74\rsp}{\rsp}$, $\setT_i\in \rclass{3.74 \rsp}{1}$ and by
rescaling, 
\begin{equation}
	\rclassidxf{3.74 \rsp}{1}{\idim}{\odim}=\rclassidxf{3.74}{1}{\idim}{\tilde \odim}, 
\end{equation}
where $\tilde \odim=(\odim-1)/\rsp+1$. Set\footnote{
  Strictly speaking, this requires $\fc/\rsp$ to be an integer. If
  $\fc/\rsp$ is not an integer, we can substitute $\fc$ with
  $\rsp \floor{\fc/\rsp}$ and repeat the argument for the new
  $\fc$. Since $\fc\ge 128\rsp$ by assumption, this transformation
  will result in less than a $1\%$ change meaning that
  $\rclass{3.74\rsp}{\rsp}$ would need to change into
  $\rclass{3.77\rsp}{\rsp}$. To keep things simple, we ignored this
  detail throughout the paper and implicitly assumed that $\fc/\rsp$
  is an integer.}
$\tilde\fc = (\tilde \odim-1)/2=(\odim-1)/(2\rsp)=\fc/\rsp$ and
$\tlambdac=1/\tilde\fc=\rsp/\fc$.  By~\fref{lem:dual}, there are
real-valued polynomials
$q_i(t; \idim, \tilde\odim)=\sum_{k=-\tilde\fc}^{\tilde\fc} \hat
q_{ik} e^{-\iu 2\pi k t}$ with $\infnorm{q_i}\le 1$ and
\begin{equation}
	\begin{cases}
		q_i(t)=0, &  \text{for all $t\in\setT_i$},\\
		q_i(t)\ge \phi_i(t), & \text{for all $t$}, 
	\end{cases}
\end{equation}
where (see~\fref{fig:proof})
\begin{align}
	\label{eq:LBdef}
	\phi_i(t)=\begin{cases}
          c_1 \tilde\fc^2 (t_0-t)^2, &  \text{for all $t$ s.t.~ $\exists t_0 \in \setT_i$ with $\abs{t-t_0}\le c_2\tlambdac$},\\
          c_1 \tilde\fc^2 (c_2 \tlambdac)^2=c_1 c_2^2, &
          \text{otherwise}.
	\end{cases}
\end{align}%
The trigonometric polynomial $q$ is obtained by taking the product of
the $q_i$'s:
\begin{equation}
	q(t)=\prod_{i=1}^{\rsp} q_i(t; \idim, \tilde\odim).
\end{equation}
By construction, $q(t)$ is band-limited, i.e.,
$q(t)=\sum_{k=-\fc}^{\fc} \hat q_k e^{-\iu 2\pi k t}$, $\infnorm{q}\le
1$, and
\begin{equation}
	\begin{cases}
		q(t_0)=0, & \text{for all } t_0\in\setT,\\
		q(t)\ge \prod_{i=1}^{\rsp} \phi_i(t), & \text{for all } t. 
	\end{cases}\label{eq:qbnd0}
\end{equation}
Next we further lower-bound $\prod_{i=1}^{\rsp} \phi_i(t)$. Fix $t$
and let $\setN=\{t_1,\ldots,t_{\hat\rsp}\}=\{\hat t\in\setT:
\abs{t-\hat t}\le c_2 \tlambdac\}$. Since
$\setT\in\rclass{3.74\rsp}{\rsp}$, it follows that $\hat
\rsp\le\rsp$. Let $t_0$ be the closest element of $\setN$ to $t$ so
that $(t_0-t)^2\le (t_i-t)^2$ for all $i=1,\ldots,\hat\rsp$. By the
definition of $\setN$, $c_1 \tilde\fc^2 (t_i-t)^2\le c_1
c_2^2$. Using~\fref{eq:LBdef} and these inequalities we may write
\begin{align}
		\prod_{i=1}^{\rsp} \phi_i(t)
		&\ge
		(c_1 c_2^2)^{\rsp-\hat\rsp} c_1^{\hat\rsp} \tilde f_c^{2\hat\rsp} \prod_{i=1}^{\hat\rsp} (t_i-t)^2 \\
		&\ge 
		\begin{cases}
		c_1^\rsp \tilde\fc^{2\rsp} (t_0-t)^{2\rsp},\ \text{for all}\  t\  \text{s.t.}\ \hat\rsp>0 \\
		(c_1 c_2^2)^\rsp, \ \text{otherwise}.
		\end{cases}
		\label{eq:qbnd1}
\end{align}%

The assumption $\SRF\ge 3/\rsp$ implies $\SRF= \idim/\odim>1/(2\rsp
c_2)\approx 2.94/\rsp$, which is equivalent to ${1}/{\idim}<c_2
\rsp {2}/{\odim}$ so that ${1}/{\idim}<c_2 \tlambdac$.
Therefore, from~\fref{eq:qbnd0} and \fref{eq:qbnd1}, it follows that
\begin{equation}
	\rho
	=\frac{1}{2}\argmin_{l/\idim\notin\setT}\{q(l/\idim)\}
	\ge c_1^\rsp \frac{1}{2} \tilde\fc^{2\rsp} \frac{1}{\idim^{2\rsp}}
	= c_1^\rsp \frac{1}{2 \rsp^{2\rsp}} \fc^{2\rsp} \frac{1}{\idim^{2\rsp}}
	= c_1^\rsp \frac{1}{2 (2 \rsp)^{2\rsp}}  \left(\frac{\odim-1}{\idim}\right)^{2\rsp}.
	\label{eq:rho1}
\end{equation}
Plugging this into~\fref{eq:errorbound} gives
\begin{equation}
	\label{eq:c1def}
	\vecnorm{\hat\vinp-\vinp}_1\le \underbrace{4 c_3^\rsp \rsp^{2\rsp}}_{C_{1}(\rsp)}  \left(\frac{\idim}{\odim-1}\right)^{2\rsp}\cdot \|\vecz\|_1,
\end{equation}
where $c_3=4/c_1=67.79$.  This completes the proof.  \qed

\begin{rem}[Possible Improvement]
  The constant $3.74$ in $\vinp\in \rclassp{3.74 \rsp}{\rsp}$ comes
  from the fact that our construction is built on top
  of~\fref{lem:dual} borrowed from~\cite[Sec. 2.5]{candes13}. If the
  constant $3.74$ in~\fref{lem:dual} is reduced, all our results
  automatically improve without any modification. Carlos
  Fernandez-Granda privately shared with
  us~\cite{fernandez-granda15super} that it is possible to substitute
  $3.74$ by $2.52$ in~\fref{lem:dual}.
\end{rem}

\begin{rem}[Extension]
  Consider a signal consisting of spike clusters as shown in
  \fref{fig:clust}, and violating the separation constraint. Suppose
  that within each cluster, the spikes have the same sign. Then our
  proof technique can be used to show that if the clusters are
  sufficiently separated, the signal can be recovered stably by convex
  programming. We omit the details.

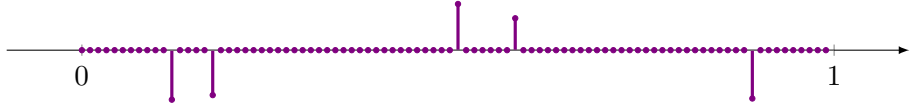
\begin{figure}[ht]
        \centering
\begin{tikzpicture}
\begin{axis}[width=12cm,scale only axis, height=1.4cm,
    name= plot1,
    ytick=\empty,
	xtick={0,1},
	axis x line = middle,
	hide y axis,
	ymin=-1,
	ymax=1,
	xmin=-0.1,
	xmax=1.1]
 	\addplot[ycomb,range=0:1,line width=0.1pt,violet,mark options={fill=violet,scale=1},mark=*, mark size=1pt,only marks] 
 		table[x index=0,y index=2]{./figures/figrclassminus.txt};
 	\addplot[ycomb,range=0:1,line width=1.2pt,violet] 
  		table[x index=0,y index=2]{./figures/figrclassminus.txt};
   \draw[help lines] (axis cs:0,0) -- (axis cs:1,0);
\end{axis}
\end{tikzpicture}
\caption{ Signal with clustered spikes that have the same sign within
  a cluster but different signs across clusters.}
\label{fig:clust}
\end{figure}
\end{rem}

\subsection{Proof of Theorem~\ref{thm:UBtriag} in the 1D case} 
\label{sec:frej}

Our strategy is to reduce the problem to that in which we have a flat
spectrum. Choose $1/2\le\cutfn< 1$ (a parameter we can optimize) so
that $\cutfn\fcn$ is an integer, and define the filter
\begin{equation}
	\label{eq:filter3}
	\hat r_k\define
	\begin{cases}
		\frac{\fcn+1}{\fcn+1-\abs{k}}, &  k=-\cutfn\fcn,\ldots, \cutfn\fcn,\\
		(\fcn+1)\left(a \abs{k}+b\right), & \abs{k}=\cutfn\fcn+1, \ldots, \fcn,\\
		0, & \text{otherwise}.
	\end{cases}
\end{equation}
with
\begin{equation}
	a\define -\frac{1}{\fcn(1-\cutfn)+1}\frac{1}{\fcn(1-\cutfn)}\quad \text{and}\quad
	b\define \frac{1}{\fcn(1-\cutfn)+1}\frac{1}{1-\cutfn}.
\end{equation}
Set $\hat\matR\define\diag(\hat\vecr)$ with $\hat\vecr\define\tp{[\hat
  r_{-\idim/2+1}\cdots \hat r_{\idim/2}]}$ and let
$\matR\define\herm\matF\hat\matR\matF$. The point is that
\[
\matT\define\matR\matQ
	=\herm\matF\hat \matR\matF \herm\matF\hat \matQ\matF
	=\herm\matF\hat \matR\hat \matQ\matF
	=\herm\matF\hat\matT\matF
\]
has a spectrum $\hat\matT\define\diag(\tp{[\hat t_{-\idim/2+1}\cdots
  \hat t_{\idim/2}]})$ given by
\begin{equation}
	\hat t_k\define
	\begin{cases}
		1,& k=-\cutfn\fcn,\ldots,  \cutfn\fcn,\\
		\left(\fcn+1-\abs{k}\right) \left(a \abs{k}+b\right),& \abs{k}=\cutfn\fcn+1, \ldots, \fcn,\\
		0,& \text{otherwise}.
	\end{cases}
\end{equation}
Note that the spectrum of $\matT$ is flat in the region
$-\cutfn\fcn\le k\le \cutfn\fcn$.  Next, construct $\vecq$ as
in~\fref{sec:proofp} with $\fc$ replaced by $\alpha\fc$. Because
$\dualp$ is band-limited to $\alpha\fc$, $\matT\dualp=\dualp$.  On the
one hand,
\begin{align}
  \abs{\inner{\dualp}{\vinper}}
  =\abs{\inner{\matT\dualp}{\vinper}} & =\abs{\inner{\dualp}{\matT\vinper}}\nonumber\\
  &\le \infnorm{\dualp}\onenorm{\matT\vinper}\nonumber\\
  &\le (1-\rho) \onenorm{\matR\matQ\vinp-\matR\vinpfilt+\matR\vinpfilt-\matR\matQ\vinpest}\nonumber\\
  &\le (1-\rho) \vecnorm{\matR}_{1,op} \left(\onenorm{\matQ\vinp-\vinpfilt}+\onenorm{\vinpfilt-\matQ\vinpest}\right)\nonumber\\
  &\le 2(1-\rho) \vecnorm{\matR}_{1,op}\onenorm{\matQ\vinp-\vinpfilt}\nonumber\\
  &\le 2 (1-\rho) \calphan \cdot \|\vecz\|_1.
			\label{eq:pt1mod}
\end{align}
The last step follows from~\fref{app:technbnd}, where we show that for all $\idim$ and all $\SRF$, 
\begin{equation}
	\label{eq:calpha}
	\vecnorm{\matR}_{1,op}\le\calphan\define 2\cutfn+\frac{2}{1-\cutfn}+\frac{1.11}{2(1-\cutfn)^2}.
\end{equation}
Note that $\calphan$ is finite as long as $\cutfn<1$. For
$\cutfn=1/2$, $\calphan=7.22$. For $\cutfn=0.75$, $\calpha=18.38$.  On
the other hand, $\abs{\inner{\dualp}{\vinper}} \ge \rho
\onenorm{\vinper}$ as before, where $\rho$ is given in~\fref{eq:rho1}
with the substitution $\odim-1=2\fc\to 2\alpha\fc$.  In conclusion,
\begin{equation}
	\label{eq:c1ra}
	C=C_{1}(\rsp,\cutfn)=C_{1}(\rsp) \calphan\left(\frac{1}{\cutfn}\right)^{2\rsp}.
\end{equation}

\subsection{Remarks on Theorems \ref{thm:UB} and \ref{thm:UBtriag} in
  2D}
\label{sec:2d}
The proof of the 2D version of~\fref{thm:UB} closely mimics that in
the 1D case. Since $\vinp\in \rclasspdd{4.76 \rsp}{\rsp}$, we can
work with a partition $\setT = \cup_{1 \le i \le r} \setT_i$ with
$\setT_i\in \rclasstwo{4.76 \rsp}{1}$. This is illustrated
in~\fref{fig:2dsets} for $\rsp=2$. The proof follows the same steps as
in~\fref{sec:proofp}. The dual trigonometric polynomial is constructed
as a product of $\rsp$ polynomials. The $i$-th term in the product has
zeros on $\setT_i$ and is constructed using~\fref{lem:dualdd} given
in~\fref{app:proof2d} for completeness; this lemma is a \ac{2D}
version of~\fref{lem:dual}, and its proof can be found in
\cite[Prop. C.1]{candes13}, \cite[Sec. D.1]{candes21-2}. The proof
of~\fref{thm:UBtriag} in the \ac{2D} case follows the steps outlined
in~\fref{sec:frej} with slight modifications, which are omitted.
\begin{figure}
	\centering
	    \begin{subfigure}[b]{0.5\textwidth}
	        \centering
	        \includegraphics[width=0.55\textwidth]{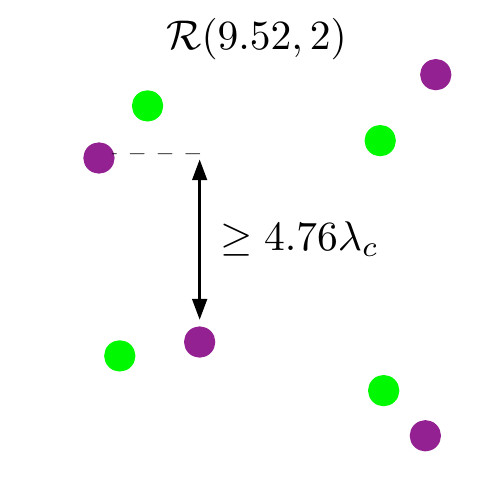}
	        \caption{}
	    \end{subfigure}%
	    ~ 
	    \begin{subfigure}[b]{0.55\textwidth}
	        \centering
	        \includegraphics[width=0.5\textwidth]{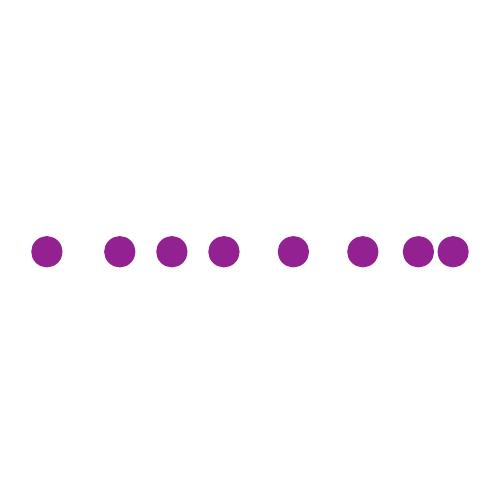}
	    	  \caption{}
	    \end{subfigure}
	
            \caption{In (a), $\setT\in\rclasstwo{9.52}{2}$ has eight
              elements and can be decomposed as
              $\setT=\setT_1\union\setT_2$ with
              $\setT_1,\setT_2\in\rclasstwo{9.52}{1}$. The points in
              $\setT_1$ are in blue and those in $\setT_2$ are in
              green. In (b), $\setT$ has eight points belonging to a
              line parallel to a coordinate axis; this is a worst-case
              scenario.}
\label{fig:2dsets}
\end{figure}

To the best of our knowledge, \fref{thm:UB} is the first result, in
the noisy and in the noiseless setting, showing that the 2D
super-resolution problem can be solved via convex optimization when
the signal is nonnegative, without assuming a separation
condition. It is instructive to discuss this point in details as the
discussion reveals interesting insights about the super-resolution
problem in higher dimensions.

Suppose one would like to obtain a noiseless result for nonnegative
signals similar to that in~\cite{donoho90-06,fuchs05}.  Following
\fref{sec:intuit} one could take a \ac{2D} version of the
 polynomial in \fref{eq:badq1},
\begin{equation}
	\label{eq:2dcos}
	\frac{1}{4}\left[\cos\lefto(2\pi t_1+ \pi\right)+1\right]\left[\cos\lefto(2\pi t_2+ \pi\right)+1\right]
\end{equation} 
and then form a product of such terms to build a low-frequency
polynomial $q(t_1,t_2)$ as done in \fref{eq:badq}. What
is the largest number of terms the product in \fref{eq:badq} could
contain in the \ac{2D} setting? Since each term of the form
\fref{eq:2dcos} costs two units in the frequency domain in each
variable, to be able to guarantee that $q(t_1,t_2)$ has frequency no
larger than $\fc$ in both variables, one can have no more that
$\odim/2$ terms of the form \fref{eq:2dcos}. Because each term in the
product is zero only at one point of the support, this technique would
not guarantee recovery of signals with more than $n/2$ spikes.  This is
discouraging since we have have $\odim^2$ observations. It is easy to
see that $\zeronorm{\vinp}<\odim/2$ is a tight bound in the worst
case: think about the situation where the support is located along a
line parallel to a coordinate axis as in~\fref{fig:2dsets}(b). In this
case, even though we have $\odim^2$ observations, the problem is
essentially one-dimensional with $\odim$ observations and no more that
$\odim/2$ spikes can possibly be resolved.

However, what happens in the typical situation where the spikes are
Rayleigh regularly spread over the domain as in~\fref{fig:2dsets}(a)?
In this case, we construct the trigonometric polynomial, which is a
product of $\rsp$ terms as in the \ac{2D} version of
\fref{eq:qbnd1}. Each term in the product has frequencies at most
$n/r$ and vanishes at $\sim\odim^2/\rsp^2$ points of the support
simultaneously. For example, in~\fref{fig:2dsets}(a) all the elements
in $\setT_1$ are roots of the first term and all those in $\setT_2$
are roots of the second. Hence, as \fref{thm:UB} shows, the number of
spikes can be as large
as $$\frac{(\odim-1)^2}{4.76^2\rsp^2}\rsp\approx
\frac{\odim^2}{4.76^2\rsp};$$ i.e.~for a fixed value of $r$,
$\zeronorm{\vinp}$ may grow linearly with the number of observations.
This is a much stronger result compared to what would be achievable
via the method from~\cite{donoho90-06,fuchs05}.

\section{Numerical results}
This section introduces a numerical simulation to illustrate the
effectiveness of our method in super-resolution microscopy.  Set
$\fc=19$, $\SRF=10$, so that $N=390$ and consider the 2D model with
$\matQ=\matQ_{\mathrm{tri,2D}}$. 
\begin{itemize}
\item The image $\vinp$ of dimensions $390\times 390$ dimensional is
  shown in \fref{fig:numb}. This image contains five different regions
  with different source densities: (i) the top-left quarter is a
  signal from $\rclasstwo{4.28}{1}$; (ii) the top-right quarter is
  from $\rclasstwo{2.14}{1}$; (iii) the lower-left quarter is from
  $\rclasstwo{4.28}{2}$; (iv) the lower-right quarter towards the
  center is from $\rclasstwo{2.24}{2}$; (v) and the lower-right
  quarter towards the corner contains three closely co-located
  spikes. All spikes were chosen to have equal magnitude set to
  $10,000$.  To be clear, we are performing one large experiment in
  which different regions of $\vinp$ exhibit different spike
  densities; we run the reconstruction algorithm only once. (Overall,
  the signal would need to belong to $\setR_2(\cdot,3)$ since it
  contains three spikes in a Nyquist cell.)

\item The observations, displayed in \fref{fig:numa}, are generated
  according to the model $\vecs=\mathrm{Pois}\left(\matQ\vinp\right)$.
\end{itemize}

We solve the \ac{LP} \fref{eq:find0} by smoothing the
  objective into $h_\mu(\vecs - \matQ \hat\vecx)$, where $h_\mu$ is
  the Huber function defined as $h_\mu(\vecy)= \sum_i h_\mu(y_i)$,
  where
\begin{equation}
	h_\mu(t)=\begin{cases}
	\frac{1}{2} t^2/\mu, & \abs{t}\le \mu,\\
	\abs{t}-\mu/2, & \abs{t}> \mu. 
	\end{cases}
\end{equation}
This is a smooth approximation to the $\ell_1$ norm, and is tight when
$\mu$ is small \cite{nesterov05smoot}.
To make sure our approximation is really tight, we set $\mu_0 = 0.1
\cdot \onenorm{\sqrt{\vecs}}/N^2\approx 0.1 \cdot
\onenorm{\vecs-\matQ\vecx}/N^2$. We then solve the smooth problem
using Lan/Lu/Monteiro's primal-dual first order
method~\cite{lan11prima} with a solver written in the framework
provided by TFOCS~\cite{becker11templ}. There are two implementation
details worth mentioning. First, we start the algorithm from an
initial guess obtained by the frequently used continuation method.
That is, we solve a series of three smoother problems (so that
convergence is faster) with $\mu \in \{10^3 \mu_0, 10^2 \mu_0, 10
\mu_0\}$, each time taking the solution to the previous problem as an
initial guess. To solve these intermediate problems, the stopping
criterion is a relative $\ell_2$ error between two consecutive
iterations below $10^{-5}$ or a number of iterations reaching $1000$,
whichever occurs first. Second, for the value of $\mu = \mu_0$, we
perform $15,000$ iterations of the Lan/Lu/Monteiro's method to obtain
a precise solution. This is an overkill but at the same time, this
guarantees that we are solving (CVX). For information, the total
computational cost is about 40,000 2D \acp{FFT} of size $390\times
390$.


%
  
The signal estimate is displayed in \fref{fig:numc}. In
\fref{fig:numd} we zoomed-in to six interesting domains of the images
in \fref{fig:numa}--\fref{fig:numc}, which are marked by white boxes
in \fref{fig:numa}. In each series of three images in~\fref{fig:numd}
we present the data, the original signal, and the estimate produced by
\fref{eq:find0}.

As we can see, in the regions (i), (ii), (iii) the algorithm performs
very well, resolving even the closely located pairs of spikes in
region (iii) (please see the zoomed-in vignettes). In region (iv) the
algorithm fails in many places, and region (v) is very poorly
resolved. The reason for the poor resolution in regions (iv) and (v)
is that in (iv), the average density of spikes is too high. In region
(v) there are too many spikes located within one Nyquist cell. 

\begin{figure}
	\centering
   	\begin{subfigure}[b]{0.32\textwidth}
       	\centering
       	\includegraphics[width=\textwidth]{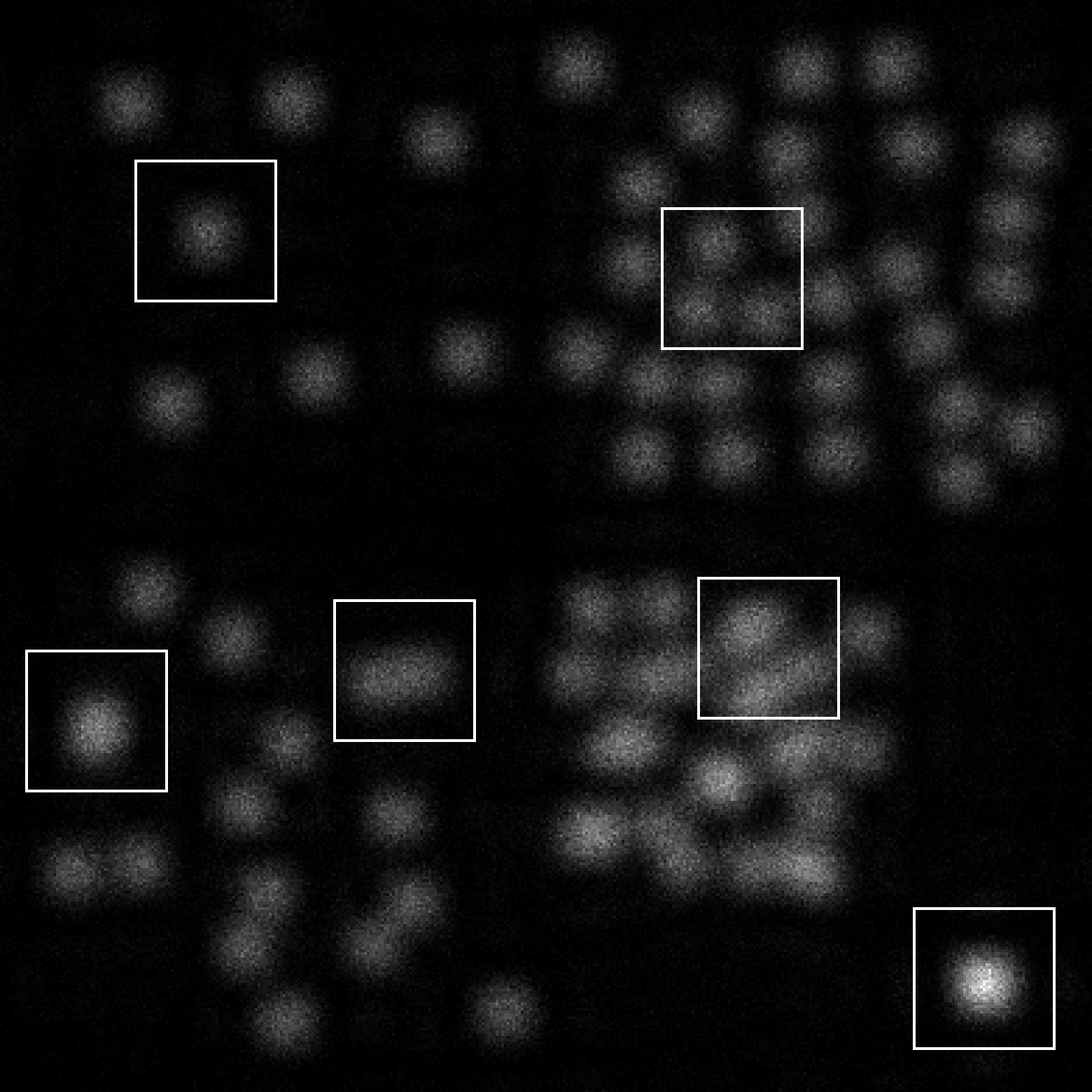}
   	  	\caption{}
			\label{fig:numa}
   	\end{subfigure}%
	   ~
		\begin{subfigure}[b]{0.32\textwidth}
	        \centering
	        \includegraphics[width=\textwidth]{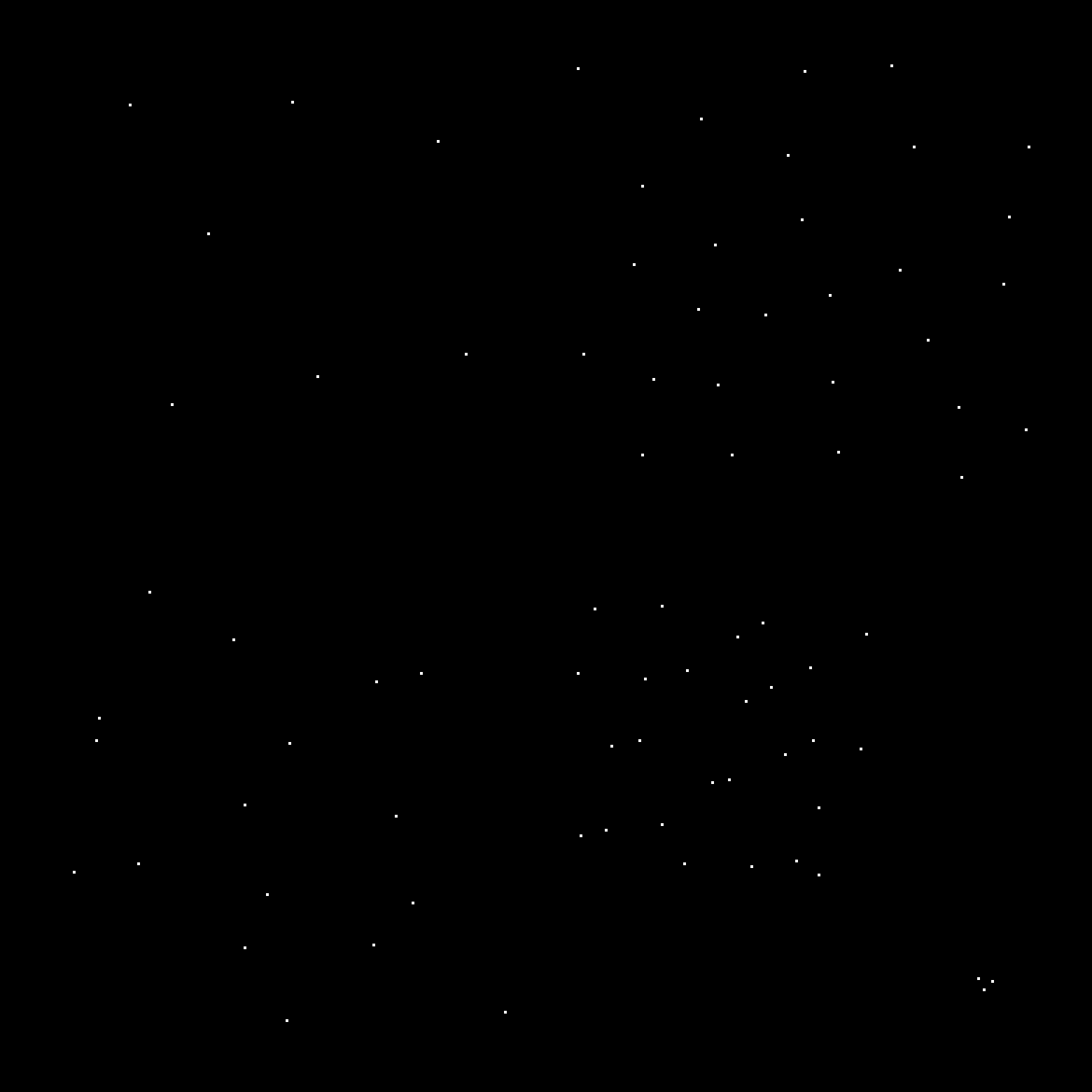}
	        \caption{}
	  			\label{fig:numb}
	    \end{subfigure}%
	    ~ 
	    \begin{subfigure}[b]{0.32\textwidth}
	        \centering
	        \includegraphics[width=\textwidth]{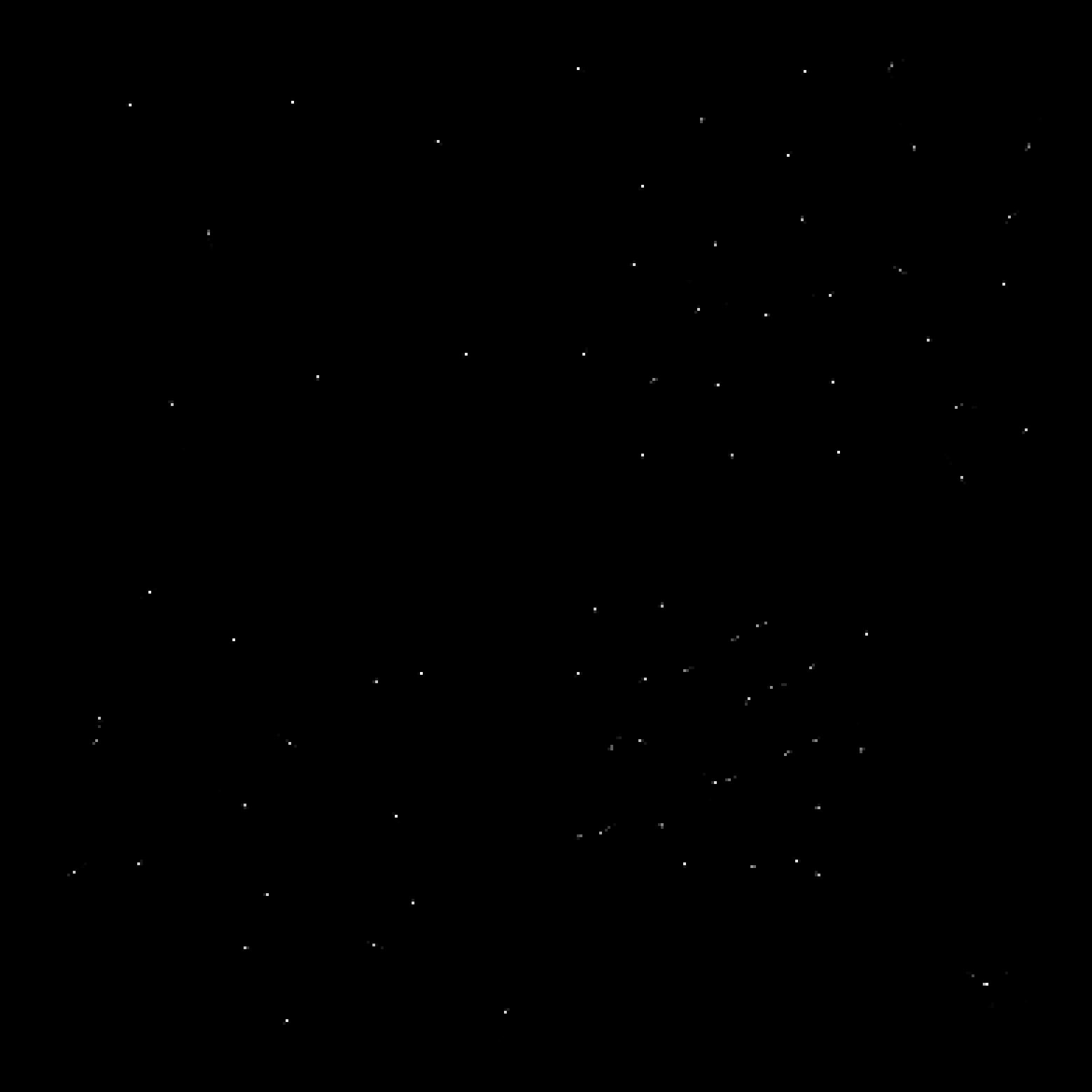}
	        \caption{}
			  \label{fig:numc}
	    \end{subfigure}
		 
	    \begin{subfigure}[b]{\textwidth}
	        \centering
	        \includegraphics[width=\textwidth]{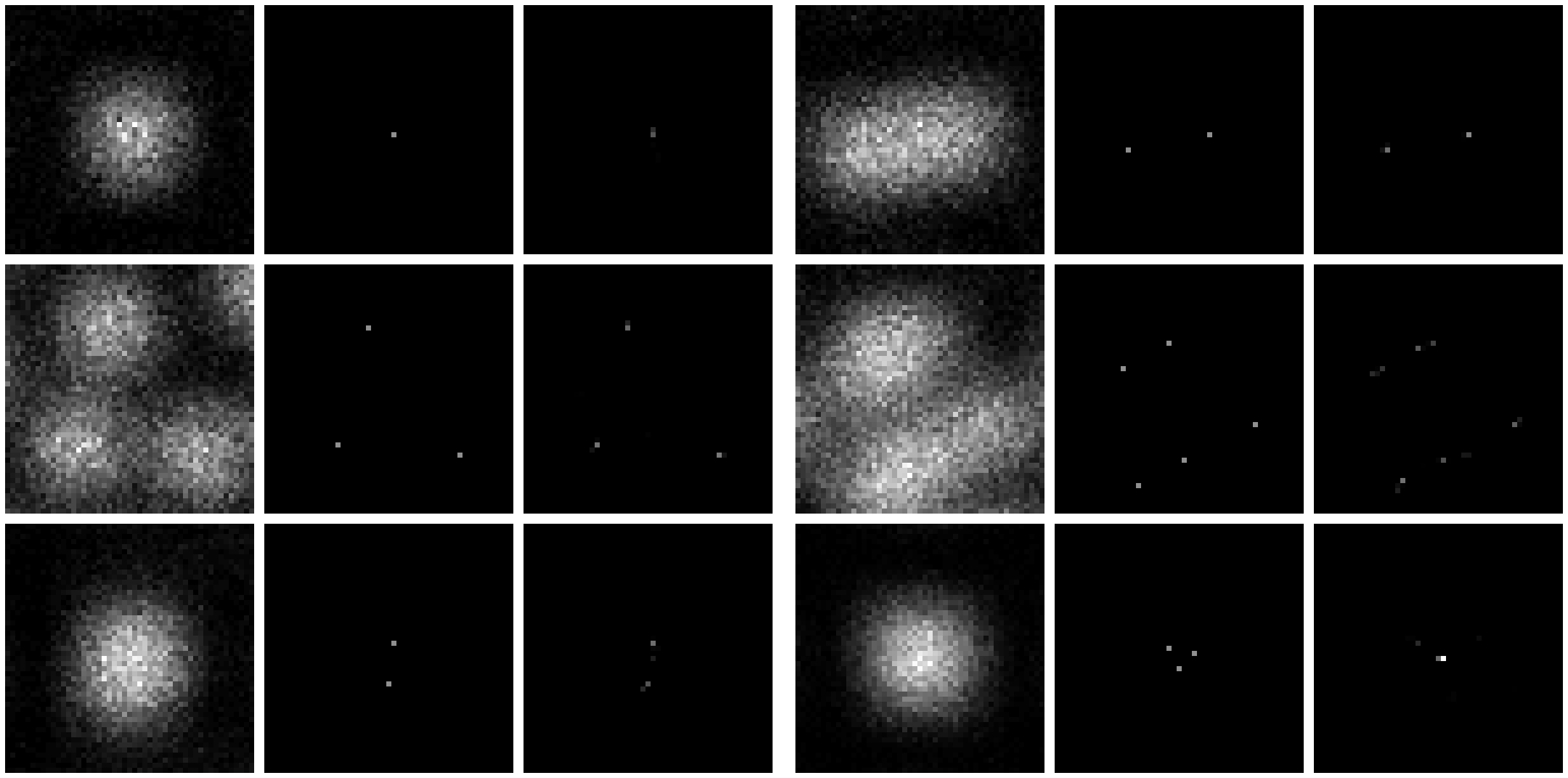}
	        \caption{}
			  \label{fig:numd}
	    \end{subfigure}%
            \caption{(a) Observed data
              $\mathrm{Pois}\left(\matQ_{\mathrm{tri,2D}}\vinp\right)$;
              (b) true signal $\vinp$; (c)  estimate produced
              by \fref{eq:find0}; (d) six zoomed-in vignettes
              corresponding to the boxes in (a); the
              rows-of-three in (d) show the observed data,
              the true signal and the estimate in this order.}
\label{fig:num}
\end{figure}

\section{Conclusion}
\label{sec:concl}

When a signal is positive and Rayleigh regular, then linear
programming solves the super-resolution problem with near-optimal
worst-case performance. Although the results presented in this paper
assume that the signal is supported on a discrete grid, extensions to
the continuum can be found in the companion paper~\cite{candes-14-2}.

A widely open research problem concerns the super-resolution of
complex-valued signals. In 1D, \cite{donoho92-09} shows that if the
signal belongs to $\rclass{4 \rsp}{\rsp}$, then stable
super-resolution is possible via exhaustive search. If the signal
belongs to $\rclass{4}{1}$, \cite{candes13} proves that stable
super-resolution can be achieved via $\ell_1$-minimization. Is there a
computationally feasible algorithm that achieves stable
super-resolution for signals in $\rclass{4 \rsp}{\rsp}$ with $\rsp>1$?
If no such algorithm is found, is it possible to show that this
problem is in some sense fundamentally difficult from a computational
viewpoint? 

\section*{Acknowledgments}

E.~C. is partially supported by NSF under grant CCF-0963835 and by the
Math + X Award from the Simons Foundation.  V.~M.~was supported by the
Swiss National Science Foundation fellowship for advanced researchers
under grant PA00P2\_139678. He is now supported by the Simons
Foundation. 

\appendix
	
\section{Proof of Theorem~\ref{thm:conv} }
\label{app:th3}
The proof uses the idea of \cite[Th.~4]{donoho90-06} with an important
difference: there, the authors provide a lower bound on a modulus of
continuity defined as $\sup_{\vinp_1,\vinp_2\in \setC}
{\onenorm{\vinp_1-\vinp_2}}/{\twonorm{\matQ_{\mathrm{flat,1D}}(\vinp_1-\vinp_2)}}$.
In our case, we are interested in a lower bound on
$\MC[\setC,\matQ]\define\sup_{\vinp_1,\vinp_2\in \setC}
{\onenorm{\vinp_1-\vinp_2}}/{\onenorm{\matQ(\vinp_1-\vinp_2)}}$.

Put $\vech=\tp{[h_0\cdots h_{\idim-1}]}=\vecx-\tilde\vecx$ and 
$\vecs=\tp{[s_0\cdots s_{\idim-1}]}=\matQ\vech$.  Then a standard
calculation shows that $\vinpfiltc_m$ can be written as
\begin{align}
	\label{eq:io4fref}
	\vinpfiltc_m=\sum_{l=0}^{\idim-1} \glow\lefto(\frac{m-l}{\idim}\right)  \vinperc_l,
\end{align}
where 
\begin{equation}
	\glow(t)= \frac{1}{(1+\fcn)\idim} \left(\frac{\sin((1+\fcn)\pi t)}{\sin(\pi t)}\right)^2
\end{equation}
is the Fej{\'e}r kernel.  The idea is to construct $\vech$ with at
most $2\rsp$ nonzero elements in such a way that for each $m$, the
terms in the sum~\fref{eq:io4fref} cancel each other out as much as
possible. One way to do this in a systematic way is to set
\[
	x_k  =
	\begin{cases}
          \frac{1}{2^{2\rsp-1}}{2\rsp-1 \choose k}, & k \in \{0, 2, \ldots, 2(r-1)\}\\
          0,& \text{otherwise}. 
	\end{cases}, \quad \tilde{x}_k = x_{k-1}
\] 
(with the periodic convention).  Obviously,
$\zeronorm{\vecx}=\zeronorm{\tilde \vecx}=\rsp$ and setting
$\rspt=2\rsp-1$ for convenience,
\begin{equation}
  \label{eq:Bfrej0}
	\onenorm{\vech}=\onenorm{\vecx-\tilde \vecx}=\frac{1}{2^\rspt}\sum_{l=0}^\rspt {\rspt \choose l}= 1. 
\end{equation}
With this, 
\begin{equation}
	\vinpfiltc_m
	=\sum_{l=0}^{\rspt}  \underbrace{ (-1)^l  \frac{1}{2^\rspt} {\rspt \choose l}}_{\vinperc_l} \glow\lefto(\frac{m-\rspt}{\idim}+\frac{\rspt-l}{\idim}\right) 
	= \frac{1}{2^\rspt}\FDOarg{1/\idim}{\frac{m-\rspt}{\idim}},
\end{equation}
where
\begin{equation}
	\FDO{\delta} = \sum_{l=0}^{\rspt} (-1)^l {\rspt \choose l} \glow\lefto(t+(\rspt-l)\delta\right)
\end{equation}
is the finite-difference operator of order $\rspt$ applied to the
kernel $\glow(\cdot)$. For large $\idim$ and large $\SRF$,
$\FDO{1/\idim} \approx \frac{1}{\idim^\rspt} \frac{d^\rspt \glow}{d
  t^\rspt}(t)$, a crucial fact allowing us to obtain closed-form
estimates on $\onenorm{\vecs}$.  Formally, write $s_m$ as a Fourier series 
\begin{equation}
		s_m=\frac{1}{\idim}\sum_{k=-\fcn}^{\fcn} e^{\iu 2\pi m k/\idim} q\lefto(\frac{k}{\fcn+1}\right) p_\rspt\lefto(\frac{k}{\idim}\right),
\end{equation}
where
\begin{equation}
		q(f)=\begin{cases}
			1-\abs{f}, &  f\in [-1,1],\\
			0,\ & \text{otherwise}, 
		\end{cases}
\end{equation}
and 
\begin{equation}
		p_\rspt(f)=\sum_{l=0}^\rspt e^{-\iu 2\pi l f} h_l.
\end{equation}
Now let 
\begin{equation}
	\label{eq:tm}
	t_m=\int_{-(\fcn+1)/\idim}^{(\fcn+1)/\idim} e^{\iu 2\pi m f} q\lefto(f\frac{\idim}{\fcn+1}\right) p_\rspt(f) df.
\end{equation}
It is not difficult to see that for all $\idim$
\begin{equation}
	\sum_{m=-\idim/2+1}^{\idim/2}\abs{t_m-s_m}\le \sum_{\abs{m}\ge \idim/2} \abs{t_m}
\end{equation}
and, therefore, since the series $\sum_{m=-\infty}^{\infty}\abs{t_m}$
converges, 
\begin{equation}
	\onenorm{\vecs}=\sum_{m=0}^{\idim-1}\abs{s_m}=\sum_{m=-\idim/2+1}^{\idim/2} \abs{s_m}\to \sum_{m=-\infty}^{\infty}\abs{t_m}
\end{equation}
when $\idim,\odimn\to\infty$ with $\idim/\odimn=\SRFn$ fixed. 
Using the fact that $q(f)=0$ for $\abs{f}>1/2$ and changing variables in the integral in~\fref{eq:tm} we can write
\begin{align}
	t_m&=\int_{-(\fcn+1)/\idim}^{(\fcn+1)/\idim} e^{\iu 2\pi m f} q\lefto(\frac{f\idim}{\fcn+1}\right) p_\rspt(f) df\\
	&=\frac{\fcn+1}{\idim} \int_{-1}^{1} e^{\iu 2\pi  ( m (\fcn+1)/\idim) f} q\lefto( f\right) p_\rspt\lefto( \frac{f (\fcn+1)}{\idim}\right) df.
\end{align}
We conclude that as $\idim,\odimn\to\infty$ with $\idim/\odimn=\SRFn$
fixed,
\begin{equation}
	\onenorm{\vecs}\to \left(\frac{1}{\SRFn}\right)^{\rspt}\frac{1}{\gfun{\rsp}{\SRFn}}
\end{equation}
where, 
\begin{equation}
	\frac{1}{\gfun{\rsp}{\eta}}\define \frac{1}{2^\rspt}\frac{1}{2\eta} \sum_{m=-\infty}^{\infty}\abs{ \int_{-1}^{1} e^{\iu 2\pi  (m/(2\eta)) f} q\lefto(f\right) (2\eta)^{2\rsp-1} p_{2\rsp-1}(f /(2\eta)) df}.
\end{equation}
Since the finite difference operator converges to the derivative operator as $\delta\to 0$:
\begin{equation}
	\frac{1}{\delta^\rspt}\FDOarg{\delta}{\cdot}\to  \frac{d^\rspt (\cdot)}{d t^\rspt },\quad \delta\to 0,
\end{equation}
it follows that for every fixed $f$,
\begin{equation}
	\eta^\rspt p_\rspt(f /\eta)\to \frac{1}{2^\rspt} (\iu 2 \pi f)^\rspt,\quad \eta\to\infty.
\end{equation}
Therefore, when $\SRFn\to\infty$,
\begin{equation}
	\gfun{\rsp}{\SRFn}\to \cnstl{\rsp}
\end{equation}
where 
\begin{equation}
	\label{eq:crsp}
	\cnstl{\rsp}\define 2^{2\rsp-1} \left(\pi^{2\rsp-1}  \int_{-\infty}^{\infty}\abs{ \int_{-1}^{1} e^{\iu 2\pi t f} q\lefto(f\right) f^{2 \rsp-1} df}d t\right)^{-1}.
\end{equation}
Direct numerical computation reveals
\begin{equation}
	\cnstl{\rsp}\ge
	\begin{cases}
		1.66 , & \rsp=1\\
		1.44 , & \rsp=2\\
		0.92, & \rsp=3\\
		0.48, & \rsp=4\\
		0.24, & \rsp=5.
	\end{cases}
\end{equation}

\section{Coherent Optics}
\label{app:cohopt}
When the illumination is perfectly coherent, the time-varying phasor amplitudes across the object plane differ only by complex constants so that we can write
\begin{equation}
	\Phi(\isp,\otime)=\Phi(\isp)\frac{\Phi(0,\otime)}{\sqrt{\left<\abs{\Phi(0,\otime)}^2\right>}}.
\end{equation}
Plugging this into~\fref{eq:quadratic1} we obtain
\begin{equation}
	\label{eq:quadratic2}
	\tilde s_\mathrm{coh}(\osp)\propto \left|\int h(\osp-\isp) \Phi(\isp) d\isp\right|^2.
\end{equation}
We see that in a coherent imaging system, the directly observable
received intensity, $\tilde s_\mathrm{coh}(\osp)$, is a nonlinear
(quadratic) function~\fref{eq:quadratic2} of the signal $\Phi(\isp)$.

\section{Proof of~\ref{eq:calpha}}
\label{app:technbnd}
By definition, $\matR=[\vecr_0\cdots\vecr_{\idim-1}]$ is a circulant matrix, and, hence, $\vecnorm{\matR}_{1,op}=\onenorm{\vecr_0}$. Further, by properties of circulant matrices, 
\begin{equation}
	\hat\vecr=\sqrt{\idim} \matF \vecr_0	
\end{equation}
or, equivalently,
\begin{equation}
	\vecr_0=\frac{1}{\sqrt{\idim}} \herm\matF\hat\vecr
\end{equation}
so that
\begin{equation}
	\vecnorm{\matR}_{1,op}=\frac{1}{\sqrt{\idim}}\onenorm{\herm\matF\hat\vecr}.
\end{equation}
We use the following lemma.
\begin{lem}
\label{lem:ftderiv}
Assume $\xseq_l$ is a discrete periodic signal with period
$\idim$. For each $l$, let
\begin{align}
	\yseq_l&=\xseq_l-\xseq_{l-1}, \label{eq:ydef}\\
	\zseq_l&=\yseq_l-\yseq_{l-1} \label{eq:zdef}
\end{align}
be the first and second differences of $\xseq_l$. Let $\xseqf_k,
\yseqf_k, \zseqf_k$ be $\idim$-periodic sequences of inverse \ac{DFT}
coefficients of $\xseq_l, \yseq_l, \zseq_l$, respectively. For
example,
\begin{equation}
	\xseqf_k = \frac{1}{\sqrt{\idim}} \sum_{l=-\idim/2+1}^{\idim/2} \xseq_l e^{2\pi\iu l k/\idim}. 
\end{equation}
Assume that 
\begin{equation}
	\sum_{k=-\idim/2+1}^{\idim/2}{\abs{\zseq_k}}\le A.
\end{equation} 
Then for all $k\ne 0 \mod \idim$,
\begin{equation}
	\abs{\xseqf_k}\le \frac{1}{\sqrt{\idim}}\frac{A}{2-2\cos\lefto(2 \pi k/\idim\right)}\quad \text{for all} \quad k\ne 0 \mod \idim.
\end{equation}
\end{lem}
\begin{proof}
	By the theorem about the \ac{DFT} of first differences~\cite[p.~223]{oppenheim96},
	\begin{equation}
          \zseqf_k = \left(1-e^{2\pi\iu k/\idim}\right)^2  \xseqf_k	\end{equation}
	and, consequently,
	\begin{equation}
		\label{eq:frm1}
		\xseqf_k = \frac{1}{\left(1-e^{2\pi\iu k/\idim}\right)^2} \zseqf_k\quad \text{for all}\quad k\ne 0, \pm \idim, \pm 2\idim, \ldots.
	\end{equation}
	Next, observe that 
	\begin{equation}
		\label{eq:frm2}
		\abs{\zseqf_k} \le \frac{1}{\sqrt{\idim}} \sum_{l=-\idim/2+1}^{\idim/2} \abs{\zseq_l} \le\frac{A}{\sqrt{\idim}}. 
	\end{equation}
	Substituting \fref{eq:frm2} into \fref{eq:frm1} and using that
        $|1-\exp(2\pi\iu k/\idim)|^2= 2 - 2\cos( 2 \pi k/\idim)$
        concludes the proof.
\end{proof}

Set $\xseq_l=\hat r_l$, continued periodically with period $\idim$,
and define $\yseq_l,\zseq_l$ as in~\fref{eq:ydef} and~\fref{eq:zdef}.
Observe the following facts:
\[
\begin{array}{ll} 
  \yseq_{k}=0, & \quad k\in\natseg{-\idim/2+1}{-\shld\fc},\\
  \yseq_{k}=(\fcn+1)\abs{a}, & \quad k\in\natseg{-\fcn+1}{-\cutfn\fcn},\\
  \yseq_{-\cutfn\fcn+1}=-\frac{\fcn+1}{((1-\cutfn)\fcn+1)((1-\cutfn)\fcn+2)}=-\yseq_{\cutfn\fcn}, & \\ 
  \yseq_k=-(\fcn+1)\abs{a}, & \quad k\in\natseg{\cutfn\fcn+1}{\fcn},\\
  \yseq_k=0, & \quad k\in \natseg{\fcn+1}{\idim/2}
\end{array}
\]
and note that $\yseq_k$ is monotonically increasing on the intervals
$\natseg{-\idim/2+1}{-\cutfn\fcn}$,
$\natseg{-\cutfn\fcn+1}{\cutfn\fcn}$ and
$\natseg{\cutfn\fcn+1}{\idim/2}$.  From this it immediately follows
that
\begin{align}
	\sum_{k={-\idim/2+1}}^{-\cutfn\fcn} \abs{\zseq_k} & =\sum_{k=-\fcn+1}^{-\cutfn\fcn} \abs{\zseq_k}=\yseq_{-\cutfn\fcn}-\yseq_{-\fcn}=(1+\fcn)\abs{a},\\
	\abs{\zseq_{-\cutfn\fcn+1}} & = \frac{\fcn+1}{((1-\cutfn)\fcn+1)((1-\cutfn)\fcn+2)}+(\fcn+1)\abs{a}=\abs{\zseq_{\cutfn\fcn+1}},\\
	\sum_{k=-\cutfn\fcn+2}^{\cutfn\fcn} \abs{\zseq_k} & =\yseq_{\cutfn\fcn}-\yseq_{-\cutfn\fcn+1}=2\frac{\fcn+1}{((1-\cutfn)\fcn+1)((1-\cutfn)\fcn+2)},\\
	\sum_{k=\cutfn\fcn+2}^{\idim/2} \abs{\zseq_k} & =\yseq_{\fcn+1}-\yseq_{\cutfn\fcn+1}=(\fcn+1)\abs{a}
\end{align}
so that 
\begin{equation}
	\sum_{k=-\idim/2+1}^{\idim/2} \abs{\zseq_k}
	=4\frac{\fcn+1}{((1-\cutfn)\fcn+1)((1-\cutfn)\fcn+2)}+4(\fcn+1)\abs{a}\le A
\end{equation}
where
\begin{equation}
	A\define\frac{1}{(\fcn+1)}\underbrace{\frac{4}{(1-\cutfn)^2}}_D+\frac{2}{\fcn}  \underbrace{\frac{2}{(1-\cutfn)^2} }_E.
\end{equation}
Furthermore, a direct calculation reveals that 
\begin{align}
		\abs{ \xseqf_k} 
		&\le \frac{1}{\sqrt{\idim}} \sum_{l=-\idim/2+1}^{\idim/2} \abs{\xseq_l}\\
		&=\frac{1}{\sqrt{\idim}} \left[1+2\left(\sum_{l=1}^{\cutfn\fcn} \abs{\xseq_l}+\sum_{l=\cutfn\fcn+1}^{\fcn} \abs{\xseq_l}\right)\right]\\
		&\le\frac{1}{\sqrt{\idim}} \left[1+2\left(\frac{\cutfn\fcn}{2}\left(1+\frac{\fcn+1}{(1-\cutfn)\fcn+1}\right) + \frac{\fcn+1}{2((1-\cutfn)\fcn+2)}-\frac{\fcn+1}{((1-\cutfn)\fcn+1)((1-\cutfn)\fcn+2)}\right)\right]\\
		&\le \frac{\fcn}{2\sqrt{\idim}} \underbrace{\left(2\cutfn+\frac{2}{1-\cutfn}\right)}_B.
\end{align}
Finally, \fref{eq:calpha} follows from
\begin{align}
	\frac{1}{\sqrt{\idim}} \onenorm{\herm\matF\hat\vecr}
	&=\frac{1}{\sqrt{\idim}}\sum_{k=-\idim/2+1}^{\idim/2}\abs{\xseqf_k}\\
	&\le \frac{2}{\sqrt{\idim}}\sum_{k=0}^{\idim/(\fcn+1)}\abs{ \xseqf_k}+\frac{2}{\sqrt{\idim}}\sum_{k=\idim/(\fcn+1)+1}^{\idim/2}\abs{\xseqf_k}\\
	&\le \frac{\fcn}{\idim} \frac{\idim}{\fcn+1} B +\frac{2A}{\idim}\sum_{k=\idim/(\fcn+1)+1}^{\idim/2}\frac{1}{2-2\cos\lefto(k 2 \pi/\idim\right)}\\
	&\le  B+\frac{D}{2 \pi^2} +\frac{E}{2 \pi^2} \frac{2(\fcn+1)}{\fcn}\\
	&\le  B+\frac{D}{2 \pi^2} +3\frac{E}{2 \pi^2}, \label{eq:auxbndproof}
\end{align}
where we used that 
\begin{align}
	\sum_{k=\idim/(\fcn+1)+1}^{\idim/2}\frac{1}{2-2\cos\lefto(k 2 \pi/\idim\right)}
	&\le \int_{\idim/(\fcn+1)}^{\idim/2}\frac{1}{2-2\cos\lefto(f 2 \pi/\idim\right)} df\\
	&=\frac{\idim \cot(\pi/(\fcn+1) )}{4 \pi}\\
	&=\frac{\idim (\fcn+1)}{4 \pi} \frac{\cot( \pi/(\fcn+1))}{(\fcn+1)}
\end{align}
and 
\begin{equation}
	\frac{\cot( \pi/(\fcn+1))}{(\fcn+1)}\le \frac{1}{\pi} \quad\text{for all}\quad (\fcn+1)\ge 0.
\end{equation}
\section{Basic Lemma in the~\ac{2D} case}
\label{app:proof2d}
\begin{lem}
	\label{lem:dualdd}
        Fix $\idim,\odim$ and assume
        $\setT\in\rclasstwo{4.76}{1}{\idim}{\odim}$. Set
        $\fc=(\odim-1)/2, \lambdac=1/\fc$ and suppose $\fc\ge 512$.
        Then there exists a real-valued trigonometric polynomial
		 \begin{equation}
                   \dualpc(\vect; \idim, \odim)=\sum_{k_1=-\fc}^{\fc}\sum_{k_2=-\fc}^{\fc} \hat \dualpc_{k_1,k_2} e^{-\iu 2\pi (k_1 t_1+k_2 t_2)},\quad \vect=\tp{[t_1, t_2]},
		 \end{equation}
                 such that $\infnorm{q}\le 1$, and
\begin{equation}
	\begin{cases}
		q(\vect)=0,\quad \text{for all} \quad \vect\in\setT\\
		q(\vect)\ge \phi(\vect),\quad \text{for all} \quad \vect,
	\end{cases}
\end{equation}
where
	\begin{align}
		\label{eq:LBqdd}
		\phi(\vect)=\begin{cases}
                  c_1 \fc^2 \twonorm{\vect_0-\vect}^2,\ \text{for all}\  \vect\  \text{s.t.~} \exists \vect_0\in\setT\ \text{with}\ \infnorm{\vect-\vect_0}\le c_2\lambdac\\
                  c_3 \ \text{for all} \ \vect\in\left\{\tilde \vect:
                    \infnorm{\tilde \vect-\vect_0}\ge c_2 \lambdac\
                    \text{for all} \ \vect_0\in\setT\right\}.
		\end{cases}
	\end{align}
	Above, $c_1$, $c_2$, and $c_3$ are numerical constants.
\end{lem}

\bibliographystyle{hieeetr}
\bibliography{IEEEabrv,publishers,confs-jrnls,vebib}

\end{document}